\documentclass[11pt]{article}
\usepackage{epsfig}
\usepackage{amsfonts}
\usepackage{amssymb}
\usepackage{amstext}
\usepackage{amsmath}
\usepackage{xspace}
\usepackage{graphicx}
\usepackage{graphics}
\usepackage{colordvi}
\usepackage{color}
\usepackage{psfrag}
\usepackage{subfigure}
\usepackage{algorithm, algorithmic}
\usepackage{wrapfig}
\usepackage{url}
\usepackage{fullpage}
\usepackage{enumitem}
\usepackage{braket}
\usepackage{multirow}
\usepackage{rotating,makecell}
\usepackage{comment}
\usepackage[latin1]{inputenc}
\usepackage{tikz}
\usetikzlibrary{shapes,arrows,positioning,shadows,snakes}

\usepackage{nameref}
\definecolor{ForestGreen}{rgb}{0.1333,0.5451,0.1333}
\definecolor{DarkRed}{rgb}{0.8,0,0}
\definecolor{Red}{rgb}{0.9,0,0}
\usepackage[linktocpage=true,
	pagebackref=true,colorlinks,
	linkcolor=DarkRed,citecolor=ForestGreen,
	bookmarks,bookmarksopen,bookmarksnumbered]
	{hyperref}

\newcommand{\gap}{{\sf Gap}\xspace}

\newcommand{\disj}{{\sf Disj}\xspace}
\newcommand{\eq}{{\sf Eq}\xspace}
\newcommand{\ham}{{\sf Ham}\xspace}
\newcommand{\conn}{{\sf Conn}\xspace}
\newcommand{\st}{{\sf ST}\xspace}
\newcommand{\mst}{{\sf MST}\xspace}
\newcommand{\ipmodthree}{{\sf IPmod3}\xspace}

\def\cP{\mathcal{P}}

\def\cA{\mathcal{A}}

\def\cX{\mathcal{X}}
\def\cY{\mathcal{Y}}

\newcommand{\size}[1]{\ensuremath{\left|#1\right|}}

\usepackage{amsthm}

\newtheorem{theorem}{Theorem}[section]
\newtheorem{lemma}[theorem]{Lemma}
\newtheorem{observation}[theorem]{Observation}
\newtheorem{corollary}[theorem]{Corollary}
\newtheorem{claim}[theorem]{Claim}

\def\square{\vbox{\hrule height.2pt\hbox{\vrule width.2pt height5pt \kern5pt
\vrule width.2pt} \hrule height.2pt}}

\theoremstyle{definition}\newtheorem{definition}[theorem]{Definition}
\theoremstyle{definition}\newtheorem{example}[theorem]{Example}

\newenvironment{prog}[1]{
\begin{minipage}{5.8 in}
{\sc\bf #1}
\begin{enumerate}}
{
\end{enumerate}
\end{minipage}
}

\renewcommand{\phi}{\varphi}
\newcommand{\eps}{\epsilon}

\newcommand{\poly}{\operatorname{poly}}

\newcommand{\reals}{{\mathbb R}}

\newcommand{\bias}{\operatorname{Bias}}

\newcommand{\server}{sv}

\ifdefined\ShowComment

\def\danupon#1{\marginpar{$\leftarrow$\fbox{D}}\footnote{$\Rightarrow$~{\sf #1 --Danupon}}}
\def\gopal#1{\marginpar{$\leftarrow$\fbox{D}}\footnote{$\Rightarrow$~{\sf #1 --Gopal}}}

\else

\def\danupon#1{}
\def\gopal#1{}
\fi

\newcommand{\squishlist}{
 \begin{itemize}
}

\newcommand{\squishend}{\end{itemize}
}

\newcounter{Lcount}
\newcommand{\squishlisttwo}{
\begin{list}{\arabic{Lcount}. }
{ \usecounter{Lcount} \setlength{\itemsep}{0pt}
\setlength{\parsep}{0pt} \setlength{\topsep}{0pt}
\setlength{\partopsep}{0pt} \setlength{\leftmargin}{2em}
\setlength{\labelwidth}{1.5em} \setlength{\labelsep}{0.5em} } }

\newcommand{\squishendtwo}{
\end{list} }

\usepackage{xcolor,setspace}

\usepackage{titling}
\author{Michael Elkin\thanks{Department of Computer Science,
Ben-Gurion University, Beer-Sheva, 84105, Israel. \hbox{E-mail}:~{\tt elkinm@cs.bgu.ac.il}.}
\and Hartmut Klauck\thanks{Division of Mathematical
Sciences, Nanyang Technological University, Singapore 637371 \& Centre for Quantum Technologies,
National University of Singapore, Singapore 117543. \hbox{E-mail}:~{\tt hklauck@gmail.com}.
Research at the Centre for Quantum Technologies is funded by the Singapore Ministry of Education
and the National Research Foundation.
}
\and Danupon Nanongkai\thanks{Faculty of Computer Science, University of Vienna, Austria. \hbox{E-mail}:~{\tt danupon@gmail.com}. Work partially done while at Division of Mathematical Sciences, Nanyang Technological University, Singapore 637371.}
\and Gopal Pandurangan\thanks{ Division of Mathematical
Sciences, Nanyang Technological University, Singapore 637371 \& Department of Computer Science, Brown University, Providence, RI 02912, USA.  \hbox{E-mail}:~{\tt gopalpandurangan@gmail.com}.
Supported in part by the following research grants: Nanyang Technological University grant M58110000, Singapore Ministry of Education (MOE) Academic Research Fund (AcRF) Tier 2 grant MOE2010-T2-2-082, 
Singapore MOE  AcRF Tier 1 grant MOE2012-T1-001-094,
and a grant from the US-Israel Binational Science
Foundation (BSF). 
}
}
\date{}
\title{Can Quantum Communication Speed Up Distributed  Computation?}

\begin{document}

\begin{titlepage}
\pagenumbering{roman}
\maketitle
\begin{abstract}
The focus of this paper is on {\em quantum distributed} computation, where we investigate whether quantum communication can help in {\em speeding up} distributed network algorithms.
Our main result is that for certain fundamental network problems such as minimum spanning tree, minimum cut, and shortest paths, quantum communication {\em does not}  help in substantially speeding up  distributed algorithms for these problems compared to the classical setting.

In order to obtain this result, we extend the technique of Das Sarma et al. [SICOMP 2012] to obtain a uniform approach to prove non-trivial lower bounds for quantum distributed algorithms for several graph optimization (both exact and approximate versions) as well as verification problems, some of which are new even in the classical setting, e.g. tight randomized lower bounds for Hamiltonian cycle and spanning tree verification, answering an open problem of Das Sarma et al., and a lower bound in terms of the weight aspect ratio, matching the upper bounds of Elkin [STOC 2004]. 
Our approach introduces the {\em Server model} and {\em Quantum  Simulation Theorem} which together provide a connection between distributed algorithms and communication complexity. The Server model is the standard two-party communication complexity model augmented with additional power; yet, most of the hardness in the two-party model is carried over to this new model. The Quantum Simulation Theorem carries this hardness further to quantum distributed computing. Our techniques, except the proof of the hardness in the Server model, require very little knowledge in quantum computing, and this can help overcoming a usual impediment in proving bounds on quantum distributed algorithms. In particular, if one can prove a lower bound for distributed algorithms for a certain problem using the technique of Das Sarma et al., it is likely that such lower bound can be extended to the quantum setting using tools provided in this paper and without the need of knowledge in quantum computing. 
\end{abstract}

\newpage
\tableofcontents
\end{titlepage}

\newpage
\pagenumbering{arabic}

%

\part{Overview}

\section{Introduction}
%

The  power and limitations of distributed (network) computation have been studied
extensively over the last three decades or so. In a distributed network,
each individual node can communicate only with its  neighboring nodes. 
Some distributed  problems can be solved entirely via local communication, e.g.,  maximal independent set, maximal matching, coloring, dominating set,  vertex cover, or approximations thereof. These are considered ``local" problems, as they can be shown to be solved using {\em small (i.e., polylogarithmic)}  communication  (e.g., see \cite{luby, peleg, Suomela09}). For example, a maximal independent set can be computed in  $O(\log n)$ time \cite{Luby86}.
However, many important problems   are ``global'' problems (which are the focus of this paper) from the distributed computation point of view. For example, to count the total number of nodes, to elect a leader,  to compute a  spanning tree (ST) or a minimum spanning tree (MST) or a shortest path tree (SPT), information necessarily must travel to the farthest nodes in a system. If exchanging a  message over a single edge costs one time unit, one needs $\Omega(D)$ time units to compute the result, where $D$ is the network diameter \cite{peleg}.
 If message size was unbounded, one can simply collect all the information in $O(D)$ time, and then compute the result. However, in many applications, there is {\em bandwidth restriction} on the size of the message (or the number of  bits) that can be exchanged over a communication link in one time unit.
This motivates studying global problems in the \cal{CONGEST} model \cite{peleg}, where each node can exchange  at most $B$ bits (typically $B$ is small, say  $O(\log n)$) among its neighbors in one time step. This is  one of the central  models in the study of distributed computation. The design of efficient algorithms for the \cal{CONGEST} model, as well as establishing lower bounds on the time complexity of various fundamental distributed computing problems, has been the subject of an active area of research called (locality-sensitive) {\em distributed computing} for the last three decades (e.g., \cite{peleg, Elkin-sigact04, Dubhashi,  KhanKMPT08,  Suomela09, DasSarmaHKKNPPW11}).  
In particular, it is now established that $\tilde{\Omega}(D+\sqrt{n})$ \footnote{$\tilde{\Omega}$ and $\tilde{O}$ notations hide polylogarithmic factors.} is a fundamental lower bound on the running time of many important graph optimization (both exact and approximate versions) and verification problems  such
as MST, ST, shortest paths, minimum cut, ST verification etc \cite{DasSarmaHKKNPPW11}.

The main focus of this paper is studying the power of distributed network computation in the {\em quantum setting}. More precisely, we consider the \cal{CONGEST} model in the quantum setting, where nodes can use quantum
processing, communicate over quantum links  using quantum bits, and use exclusively quantum
phenomena such as {\em entanglement} (e.g., see \cite{DenchevP08, BroadbentT08, GavoilleKM09}). 
A fundamental question that we would like to investigate is whether quantumness can help in speeding up distributed computation for graph optimization problems; in particular, whether  the  above mentioned lower bound of $\tilde{\Omega}(D+\sqrt{n})$ (that applies to many important problems in the classical setting) also applies to the quantum setting.

Lower bounds for local problems  (where the running time is  $O(\poly\log n)$) in the quantum setting   usually  follow directly from the same arguments as in the classical setting. This is because these lower bounds are proved using the ``limited sight'' argument: The nodes do not have time to get the information of the entire network. Since entanglement {\em cannot be used to replace communication} (by, e.g., Holevo's theorem \cite{Holevo73} (also see \cite{NielsenChuangBook,Nayak99})), the same argument holds in the quantum setting with prior entanglement. This argument is captured by the notion of {\em physical locality} defined by Gavoille et al. \cite{GavoilleKM09}, where it is shown that for many {\em local}  problems, quantumness does not give any significant speedup in time compared to the classical setting.

The above limited sight argument, however, does not seem to be extendible to {\em global} problems where the running time is usually $\Omega(D)$, since nodes have enough time to see the whole network in this case. In this setting, the argument developed in \cite{DasSarmaHKKNPPW11} (which follows the line of work in
\cite{PelegR00,LotkerPP06,Elkin06,KorKP11}) can be used to show tight lower bounds for many problems in the classical setting. However, this argument does not always hold in the quantum setting because it essentially relies on network ``congestion'': Nodes cannot communicate fast enough (due to limited bandwidth) to get important information to solve the problem. However, we know that the quantum communication and entanglement can potentially {\em decrease the amount of communication} and thus there might be some problems that can be solved faster. One  example that illustrates this point is the following {\em distributed verification of disjointness function} defined in \cite[Section 2.3]{DasSarmaHKKNPPW11}.

{
\begin{example}\label{example:disj}
Suppose we give $b$-bit string $x$ and $y$ to node $u$ and $v$ in the network, respectively, where $b=\sqrt{n}$. We want to check whether the inner product $\langle x, y\rangle$ is zero or not. This is called the {\em Set Disjointness} problem (\disj). It is easy to show that it is impossible to solve this problem in less than $D/2$ rounds since there will be no node having the information from both $u$ and $v$ if $u$ and $v$ are of distance $D$ apart. (This is the very basic idea of the limited sight argument.)
This argument holds for both classical and quantum setting and thus we have a lower bound of $\Omega(D)$ on both settings.
\cite[Lemma 4.1]{DasSarmaHKKNPPW11} shows that this lower bound can be significantly improved to $\tilde \Omega(b)=\tilde \Omega(\sqrt{n})$ in the classical setting, even when the network has diameter $O(\log n)$. 
This follows from the communication complexity of $\Omega(b)$ of \disj \cite{BabaiFS86,KalyanasundaramS92,Bar-YossefJKS04,Razborov92} and the {\em Simulation Theorem} of \cite{DasSarmaHKKNPPW11}.
This lower bound, however, does not hold in the quantum setting since we can simulate the known $O(\sqrt{b})$-communication quantum protocol of \cite{AaronsonA05} in $O(\sqrt{b}D)=O(n^{1/4}D)$ rounds. 
\qed
\end{example}
}

Thus we have an example of a global problem that quantum communication gives an advantage over classical communication. This example also shows that the previous techniques and results from \cite{DasSarmaHKKNPPW11} does not apply to the quantum setting since \cite{DasSarmaHKKNPPW11} heavily relies on the hardness of the above distributed disjointness verification problem. A fundamental question is: {\em ``Does this phenomenon occur for natural global distributed network problems?"}
%
%
%
%

Our paper  answers the above question where we show that this phenomenon does not occur for many global graph problems.
Our main result is that for  fundamental global  problems such as minimum spanning tree, minimum cut, and shortest paths, quantum communication {\em does not}  help significantly in  speeding up  distributed algorithms for these problems compared to the classical setting. More precisely, we show that  $\tilde{\Omega}(D+\sqrt{n})$ is a lower bound for these problems in the quantum setting as well.
An $\tilde{O}(D+ \sqrt{n})$ time algorithm for  MST problem in the classical setting  is well-known \cite{KuttenP98}.  Recently, it has been shown that minimum cut also admits a distributed $(1+\epsilon)$-approximation algorithm in the same time in the classical setting \cite{GhaffariK13,Su14,Nanongkai14-MinCut,NanongkaiSu14-MinCut}. Also, recently it has been shown that shortest paths admits an $\tilde O(D+ \sqrt{n}D^{1/4})$-time  $(1+\epsilon)$-approximation and $\tilde O(\sqrt{n}+D)$-time $O(\log n)$-approximation  distributed classical algorithms \cite{LenzenP13,Nanongkai13-ShortestPaths}.
Thus, our quantum lower bound shows that quantum communication does not speed up distributed algorithms for MST and minimum cut, while for shortest paths the speed up, if any, is bounded by $O(D^{1/4})$ (which is small for small diameter graphs).

\danupon{I rewrote this para since the previous one was inaccurate (as it was copied from the previous abstract). I still feel that this para is too similar to the abstract though.}
In order to obtain our  quantum lower bound results, we develop a uniform approach to prove non-trivial lower bounds for quantum distributed algorithms. 
This approach leads us to several non-trivial quantum distributed lower bounds (which are the first-known quantum bounds for problems such as minimum spanning tree, shortest paths etc.), some of which are new even in the classical setting. 
Our approach introduces the {\em Server model} and {\em Quantum  Simulation Theorem} which together provide a connection between distributed algorithms and communication complexity. The Server model is simply the standard two-party communication complexity model augmented with a powerful {\em Server} who can communicate for free but receives no input (cf. Def.~\ref{def:server model}). It is more powerful than the two-party model, yet captures most of the hardness obtained by the current quantum communication complexity techniques.  
The Quantum Simulation Theorem (cf. Theorem \ref{theorem:from server to distributed}) is an extension of the Simulation Theorem of Das Sarma et al. \cite{DasSarmaHKKNPPW11} from the classical setting to the quantum one. It carries this hardness from the Server model further to quantum distributed computing. 
Most of our techniques require very little knowledge in quantum computing, and this can help overcoming a usual impediment in proving bounds on quantum distributed algorithms.  
In particular, if one can prove a lower bound for distributed algorithms in the classical setting using the technique of Das Sarma et al., then it is possible that one can also prove the same lower bound in the quantum setting in essentially the same way -- the only change needed is that the proof has to start from problems that are hard on the server model that we provide in this paper.

\section{The Setting}
\subsection{Quantum Distributed Computing Model}
\label{sec:model}

We study problems in a natural quantum version of the \cal{CONGEST(B)} model \cite{peleg} (or, in short, the $B$-model), where each node can exchange  at most $B$ bits (typically $B$ is small, say  $O(\log n)$) among its neighbors in one time step. 
The main focus of the current work is to understand the time complexity  of  fundamental graph problems in the $B$-model in the {\em quantum setting}.
We now explain the model. We refer the readers to 
Appendix~\ref{sec:formal definition of network} 
for a more rigorous and formal definition of our model.


Consider a synchronous  network of processors  modeled by an undirected $n$-node graph, where nodes model the processors and  edges model the links between the processors. The processors  (henceforth, nodes) are assumed to have unique IDs. Each node has limited topological knowledge; in particular, it only knows the IDs of its neighbors and knows no other topological information (e.g., whether its neighbors are linked by an edge or not). The node may also accept some additional inputs as specified by the problem at hand.

The communication is synchronous, and occurs in discrete pulses, called {\em rounds}. All the nodes
wake up simultaneously at the beginning of each round. In each round each node $u$ is allowed to
send an arbitrary message of  $B$ bits through each edge $e = (u, v)$ incident
to $u$, and the message will arrive at $v$ at the end of the current round. Nodes then perform an internal computation, which finishes instantly since nodes have infinite computation power.
%
%
There are several measures to analyze the performance of distributed algorithms, a fundamental one being the {\em running time}, defined as the worst-case number of rounds of distributed communication.




In the quantum setting, a distributed network could be augmented with two additional resources: {\em quantum communication} and {\em shared entanglement} (see e.g., \cite{DenchevP08}).
Quantum communication allows nodes to communicate with each other using {\em quantum bits (qubits)}; i.e., in each round at most $B$ qubits can be sent through each edge in each direction.
Shared entanglement allows nodes to possess qubits that are entangled with qubits of other nodes\footnote{Roughly speaking, one can think of shared entanglement as a ``quantum version'' of shared randomness. For example, a well-known entangled state on two qubits is the {\em EPR} pair  \cite{EPR35,Bell64}
which is a pair of qubits that, when measured, will either both be zero or both be one, with probability $1/2$ each. An EPR pair shared by two nodes can hence be used to, among other things, generate a shared random bit for the two nodes. Assuming entanglement implies shared randomness (even among all nodes), but also allows for other operations such as quantum teleportation \cite{NielsenChuangBook}, which replaces quantum communication by classical communication plus entanglement.}.
%
%
%
Quantum distributed networks can be categorized based on which resources are assumed (see, e.g., \cite{GavoilleKM09}).
In this paper, we are interested in the {\em most powerful model}, where both quantum communication and the {\em most general form of shared entanglement} are assumed: in a technical term, we allow nodes to share an arbitrary $n$-partite entangled state as long as it does not depend on the input (thus, does not reveal any input information).
Throughout the paper, we simply refer to this model as quantum distributed network (or just distributed network, if the context is clear).  All lower bounds we show in this paper hold in this model, and thus also imply lower bounds in weaker models.
 %

\subsection{Distributed Graph Problems}
\label{sec:problems}

We focus on solving graph problems on distributed networks.
%
We are interested in two types of graph problems: optimization and verification problems. In both types of problems, we are given a distributed network $N$ modeled by a graph and some property $\cP$ such as ``Hamiltonian cycle'', ``spanning tree'' or ``connected component''.

In optimization problems, we are additionally given a (positive) weight function $w:E(N)\rightarrow \reals_+$ where every node in the network knows weights of edges incident to it.  Our goal is to find a subnetwork $M$ of $N$ of minimum weight that satisfies $\cP$ (e.g. minimum Hamiltonian cycle or MST) where every node knows which edges incident to it are in $M$ in the end of computation. Algorithms can sometimes depend on the {\em weight aspect ratio} $W$ defined as $W=\frac{\max_{e\in E(N)} w(e)}{\min_{e\in E(N)} w(e)}$.

In verification problems, we are additionally given a subnetwork $M$ of $N$ as the problem input (each node knows which edges incident to it are in $M$). We want to determine whether $M$ has some property, e.g., $M$ is a Hamiltonian cycle ($\ham(N)$), a spanning tree ($\st(N)$), or a connected component ($\conn(N)$), where every node knows the answer in the end of computation.

We use\footnote{We mention the reason behind our complexity notations. First, we use $*$ as in $Q^*$ in order to emphasize that our lower bounds hold even when there is a shared entanglement, as usually done in the literature. Since we deal with different models in this paper, we put the model name after $*$. Thus, we have $Q^{*, N}$ for the case of distributed algorithm on a distributed network $N$, and  $Q^{*, cc}$ and $Q^{*, \server}$ for the case of the standard communication complexity and the Server model (cf. Subsection~\ref{sec:techniques}), respectively.} $Q^{*, N}_{\epsilon_0, \epsilon_1}(\ham(N))$ to refer to the quantum time complexity of Hamiltonian cycle verification problem on network $N$ where for any $0$-input $M$ (i.e. $M$ is not a Hamiltonian cycle), the algorithm has to output zero with probability at least $1-\epsilon_0$ and for any $1$-input $M$ (i.e. $M$ is a Hamiltonian cycle), the algorithm has to output one with probability at least $1-\epsilon_1$. (We call this type of algorithm {\em $(\epsilon_0, \epsilon_1)$-error}.)
When $\epsilon_0=\epsilon_1=\epsilon$, we simply write $Q^{*, N}_{\epsilon}(\ham(N))$.\danupon{To do: Try to avoid $Q^{*, N}_{\epsilon}(\ham(N))$.}
Define $Q^{*, N}_{\epsilon_0, \epsilon_1}(\st(N)))$ and $Q^{*, N}_{\epsilon_0, \epsilon_1}(\conn(N))$ similarly.

We also study the {\em gap versions} of verification problems. For any integer $\delta\geq 0$, property $\cP$ and a subnetwork $M$ of $N$, we say that $M$ is {\em $\delta$-far}\footnote{We note that the notion of $\delta$-far should not be confused with the notion of {\em $\epsilon$-far} usually used in property testing literature where we need to add and remove at least $\epsilon$ {\em fraction} of edges in order to achieve a desired property. The two notions are closely related. The notion that we chose makes it more convenient to reduce between problems on different models.} from $\cP$ if  we have to add at least $\delta$ edges {\em from $N$} and remove any number of edges in order to make $M$ satisfy $\cP$.
%
%
We denote the problem of distinguishing between the case where the subnetwork $M$ satisfies $\cP$ and is $\delta$-far from satisfying $\cP$ {\em the $\delta$-$\cP$ problem} (it is promised that the input is in one of these two cases). When we do not want to specify $\delta$, we write $\gap$-$\cP$.

Other graph problems that we are interested in are those in \cite{DasSarmaHKKNPPW11} and their gap versions. We provide definitions in 
Appendix~\ref{sec:formal definition of network}
for completeness.


\section{Our Contributions}

\danupon{I just added this paragraph. Please have a look.} Our first contribution is lower bounds for various fundamental verification and optimization graph problems, some of which are new even in the classical setting and answers some previous open problems (e.g. \cite{DasSarmaHKKNPPW11}). We explain these lower bounds in detail in Section~\ref{sec:results}. 
The main implication of these lower bounds is that quantum communication does {\em not} help in substantially speeding up  distributed algorithms for many of these problems compared to the classical setting. Notable examples are MST, minimum cut, $s$-source distance, shortest path tree, and shortest $s$-$t$ paths. In Corollary~\ref{corollary:main optimization}, we show a lower bound of $\Omega(\sqrt{\frac{n}{B\log n}})$ for these problems which holds against any quantum distributed algorithm with any  approximation guarantee. 
Due to the seminal paper of Kutten and Peleg \cite{KuttenP98}, we know that MST can be computed exactly in $\tilde O(\sqrt{n}+D)$ time in the classical setting, and thus we cannot hope to use quantum communication to get a significant speed up for MST.
Recently, Ghaffari and Kuhn \cite{GhaffariK13} showed that minimum cut can be $(2+\epsilon)$-approximated in $\tilde O(\sqrt{n}+D)$ time in the classical setting, and Su \cite{Su14} and Nanongkai \cite{Nanongkai14-MinCut} independently improved the approximation ratio to $(1+\epsilon)$; this implies that, again, quantum communication does not help.
More recently, Nanongkai \cite{Nanongkai13-ShortestPaths} showed that $s$-source distance, shortest path tree, and shortest $s$-$t$ paths, can be $(1+o(1))$-approximated in $\tilde O(\sqrt{n}D^{1/4}+D)$ time in the classical setting; thus, the speedup that quantum communication can provide for these problems, if any, is bounded by $O(D^{1/4})$. 
Moreover, if we allow higher approximation factor, the result of Lenzen and Patt-Shamir \cite{LenzenP13} implies that we can $O(\log n)$-approximate these problems in $\tilde O(\sqrt{n}+D)$ time; this upper bound together with our lower bound leaves no room for quantum algorithms to improve the time complexity.
%
%
%
%
%
Besides the above lower bounds for optimization problems, we show the same lower bound of $\Omega(\sqrt{\frac{n}{B\log n}})$ for verification problems in Corollary~\ref{corollary:main verification}. Das Sarma et al. \cite{DasSarmaHKKNPPW11} showed that these problems, except least-element list verification, can be solved in $\tilde O(\sqrt{n}+D)$ time in the classical setting; thus, once again, quantum communication does not help. 

Our second contribution is the systematic way to prove lower bounds of quantum distributed algorithms. The high-level idea behind our lower bound proofs is by establishing a connection between quantum communication complexity and quantum distributed network computing. Our work is inspired by \cite{DasSarmaHKKNPPW11} (following a line of work in \cite{PelegR00,LotkerPP06,Elkin06,KorKP11}) which shows lower bounds for many graph verification and optimization problems in the classical distributed computing model. The main technique used to show the classical lower bounds in \cite{DasSarmaHKKNPPW11} is the {\em Simulation Theorem} (Theorem 3.1 in \cite{DasSarmaHKKNPPW11}) which shows how one can use lower bounds in the standard two-party classical communication complexity model \cite{KNbook} to derive lower bounds in the ``distributed'' version of communication complexity. We provide techniques of the same flavor for proving quantum lower bounds. In particular, we develop the {\em Quantum Simulation Theorem}. 
However, due to some difficulties in handling quantum computation (especially the entanglement) we need to introduce one more concept: instead of applying the Quantum Simulation Theorem to the standard two-party communication complexity model, we have to apply it to a slightly stronger model called {\em Server model}. We show that working with this stronger model does not make us lose much: several hard problems in two-party communication complexity remain hard in this model, so we can still prove hardness results  using these problems. Quantum Simulation Theorem together with the Server model give us a tool to bring the hardness in the quantum two-party setting to the distributed setting.  In Section~\ref{sec:techniques}, we give a more comprehensive overview of our techniques. 
%
%
%
Along the way, we also obtain new results in the standard communication complexity model, which we explain in Section~\ref{sec:additional results}. 

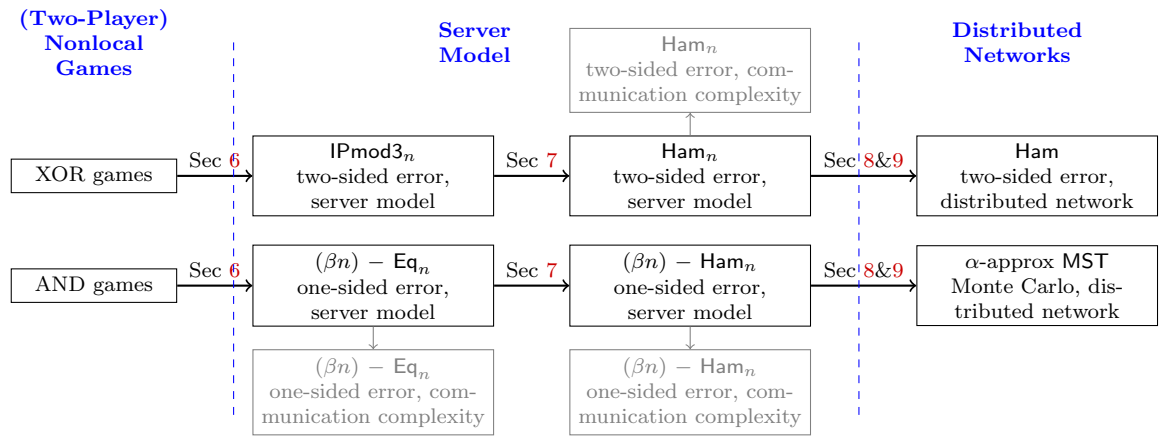
\begin{figure*}
\center
%
\begin{tikzpicture}[auto, node distance=0.1cm]
\scriptsize
\node (xor) [rectangle, draw=black, text centered, text=black, text width=2cm] {XOR games};
\node (and) [rectangle, draw=black, text centered, text=black, text width=2cm, below=1 of xor] {AND games};
\node (ipthree_server) [rectangle, draw=black, text centered, text=black, right=1 of xor, text width=3cm] {$\ipmodthree_n$ \\ {two-sided error, server model}};
\node (disj_server) [rectangle, draw=black, text centered, text=black, right=1 of and, text width=3cm] {$(\beta n)-\eq_n$\\ {one-sided error, server model}};
\node (ham_server) [rectangle, draw=black, text centered, text=black, right=1 of ipthree_server, text width=3cm] {$\ham_n$ \\ {two-sided error, server model}};
\node (conn_server) [rectangle, draw=black, text centered, text=black, right=1 of disj_server, text width=3cm] {$(\beta n)-\ham_n$ \\ {one-sided error, server model}};
\node (ham_distributed) [rectangle, draw=black, text centered, text=black, right=1.4 of ham_server, text width=3cm] {$\ham$ \\ {two-sided error, distributed network}};
\node (conn_distributed) [rectangle, draw=black, text centered, text=black, right=1.4 of conn_server, text width=3cm] {$\alpha$-approx $\mst$ \\ {Monte Carlo, distributed network}};
\node (ham_cc) [rectangle, draw=gray, text centered, text=gray, above=0.3 of ham_server, text width=3cm] {$\ham_n$ \\ {two-sided error, communication complexity}};
\node (conn_cc) [rectangle, draw=gray, text centered, text=gray, below=0.3 of conn_server, text width=3cm] {$(\beta n)-\ham_n$ \\ {one-sided error, communication complexity}};
\node (disj_cc) [rectangle, draw=gray, text centered, text=gray, below=0.3 of disj_server, text width=3cm] {$(\beta n)-\eq_n$\\ {one-sided error, communication complexity}};
\draw[->, thick, black] (xor.east) -- node [midway] {Sec~\ref{sec:basic lower bounds server model}}  (ipthree_server.west);
\draw[->, thick, black] (ipthree_server.east) -- node [midway] {Sec~\ref{sec:server graph}}  (ham_server.west);
\draw[->, thick, black] (ham_server.east) -- node [midway] {Sec~\ref{sec:from server to distributed}\&\ref{sec:proof of main}} (ham_distributed.west);
%
\draw[->, thick, black] (and.east) -- node [midway] {Sec~\ref{sec:basic lower bounds server model}}  (disj_server.west);
\draw[->, thick, black] (disj_server.east) -- node [midway] {Sec~\ref{sec:server graph}} (conn_server.west);
\draw[->, thick, black] (conn_server.east) --  node [midway] {Sec~\ref{sec:from server to distributed}\&\ref{sec:proof of main}} (conn_distributed.west);
\draw[->, gray] (ham_server.north) -- (ham_cc.south);
\draw[->, gray] (conn_server.south) -- (conn_cc.north);
\draw[->, gray] (disj_server.south) -- (disj_cc.north);
  \node (nonlocal) [above=1 of xor, text width=2.3cm, text centered, text=blue] {{\bf (Two-Player) Nonlocal Games}};
  \draw [-, dashed, blue] ([xshift=0.6cm] nonlocal.east) -- ++(0, -5);
  \node (server) [right=2.7 of nonlocal, text width=2cm, text centered, text=blue] {{\bf Server Model}};
  \draw [-, dashed, blue] ([xshift=4cm] server.east) -- ++(0, -5);
  \node (distributed) [right=5 of server, text width=2cm, text centered, text=blue] {{\bf Distributed Networks}};
\end{tikzpicture}
\caption{\small\it Our proof structure. Lines in gray show the implications of our results in communication complexity.}\label{fig:proof structure}
\end{figure*}

\subsection{Lower Bound Techniques for Quantum Distributed Computing}
\label{sec:techniques}


The high-level idea behind our lower bound proofs is establishing a connection between quantum communication complexity and quantum distributed network computing via a new communication model called the {\em Server model}, as shown in two middle columns of Fig.~\ref{fig:proof structure}.  This model is a generalization of the standard two-party communication complexity model in the sense that the Server model can simulate the two-party model; thus, lower bounds on this model imply lower bounds on the two-party network models. More importantly, we show that lower bounds on this model imply lower bounds on the quantum distributed model as well (cf. Section~\ref{sec:from server to distributed}~\&~\ref{sec:proof of main}). This is depicted by the rightmost arrows in Fig.~\ref{fig:proof structure}.
In addition, we prove quantum lower bounds in the server model, some of which also imply {\em new} lower bounds in the two-party model for problems such as Hamiltonian cycle and spanning tree, even in the classical setting. This is done by showing that certain techniques based on {\em nonlocal games} can be extended to prove lower bounds on the Server model (cf. Section~\ref{sec:basic lower bounds server model}) as depicted by leftmost arrows in Fig.~\ref{fig:proof structure}, and by reductions between problems in the Server models (cf. Section~\ref{sec:server graph}) as depicted by middle arrows in Fig.~\ref{fig:proof structure}.
%
%
%
\begin{definition}[Server Model]\label{def:server model} There are three players in the server model: Carol, David and the server. Carol and David receive the inputs $x$ and $y$, respectively, and want to compute $f(x, y)$ for some function $f$. (Observe that the server receives no input.) Carol and David can talk to each other. Additionally, they can talk to the server. The catch here is that the server can send messages for {\em free}. Thus, the communication complexity in the server model counts only messages sent by Carol and David.
\end{definition}

We let $Q^{*, \server}_{\epsilon_0, \epsilon_1}(f)$ denote the communication complexity --- in the quantum setting with entanglement --- of computing function $f$ where for any $i$-input (an input whose correct output is $i\in\{0, 1\}$) the algorithm must output $i$ with probability at least $1-\epsilon_i$. We will write $Q^{*, \server}_{\epsilon}(f)$ instead of $Q^{*, \server}_{\epsilon, \epsilon}(f)$.
For the standard two-party communication complexity model \cite{KNbook}, we use  $Q^{*, cc}_{\epsilon_0, \epsilon_1}(f)$ to denote the communication complexity in the quantum setting with entanglement.

To the best of our knowledge, the Server model is different from other known models in communication complexity. Clearly, it is different from multi-party communication complexity since the server receives no input and can send information for {\em free}. Moreover, 
it is easy to see that the Server model, even without prior entanglement, is at least as strong as the standard quantum communication complexity model with shared entanglement, since the server can dispense any entangled state to Carol and David.
Interestingly, it turns out that the Server model is {\em equivalent} to the standard two-party model in the classical communication setting, while it is not clear if this is the case in the quantum communication setting. This is the main reason that proving lower bounds in the quantum setting is more challenging in its classical counterpart.

To explain some issues in the quantum setting, let us sketch the proof of the fact that the two models are {\em equivalent} in the classical setting. Let us first consider the deterministic setting. The proof is by the following ``simulation'' argument. Alice will simulate Carol and the server. Bob will simulate David and the server. In each round, Alice will see all messages sent from the server to Carol and thus she can keep simulating Carol. However, she does not see the message sent from David to the server which she needs to simulate the server. So, she must get this message from Bob. Similarly, Bob will get from Alice the message sent from Carol to the server. These are the only messages we need in each round in order to be able to simulate the protocol. Observe that the total number of bits sent between Alice and Bob is exactly the number of bits sent by Carol and David to the server. Thus, the complexities of both models are exactly the same in the deterministic case.
We can conclude the same thing for the public coin setting (where all parties share a random string) since Alice and Bob can use their shared coin to simulate the shared coin of Carol, David and the server.

The above argument, however, does not seem to work in the quantum setting. The main issue with a simulation along the lines of the one sketched above is that Alice and Bob cannot simulate a ``copy'' of the server each.
For instance one could try to simulate the server's state in a distributed way by maintaining the state that results by applying CNOT to every qubit of the server and a fresh qubit, and distribute these qubits to Alice and Bob. But then if the server sends a message to Carol, Bob would have to disentangle the corresponding qubits in his copy, which would require a message to Alice.

While we leave as an open question whether the two models are equivalent in the quantum setting, we prove that many lower bounds in the two-party model extend to the server model, via a technique called nonlocal games.

\paragraph{Lower Bound Techniques on the Server Model (Details in Section~\ref{sec:basic lower bounds server model})} 
We show that many hardness results  in the two-party model  (where there is no server) carry over to the Server model. This is the only part that the readers need some background in quantum computing. 
%
%
%
The main difficulty in showing this is that, the Server model, even {\em without} prior entanglement, is clearly at least {\em as strong as} the standard quantum communication complexity model (where there is no server) {\em with} shared entanglement, since the server can dispense any entangled state to Carol and David. Thus, it is a challenging problem, which could be of an independent interest, whether {\em all} hard problems in the standard model remain hard in the server model.


While we do not fully answer the above problem, we identify a set of lower bound techniques in the standard quantum communication complexity model that can be carried over to the Server model, and use them to show that {\em many} problems remain hard. Specifically, we show that techniques based on the (two-player) {\em nonlocal} games (see, e.g., \cite{LeeS09,LeeZ10,fooling12}) can be extended to show lower bounds on the Server model.

Nonlocal games are games where two players, Alice and Bob, receive an input $x$ and $y$ from some distribution that is known to them and want to compute $f(x, y)$. Players cannot talk to each other; instead, they output one bit, say $a$ and $b$, which are then combined to be an output. For example, in {\em XOR games} and {\em AND games}, these bits are combined as $a\oplus b$ and $a\wedge b$, respectively. The players' goal is to maximize the probability that the output is $f(x, y)$.
We relate nonlocal games to the server model by showing that the XOR- and AND-game players can use an efficient server-model protocol to guarantee a good winning chance: 
\begin{lemma}\label{lemma:xor and simulate server}\label{LEMMA:XOR AND SIMULATE SERVER}{\sc (Server Model Lower Bounds via Non\-local Games)} For any boolean function $f$ and $\epsilon_0, \epsilon_1\geq 0$, there is an (two-player nonlocal) XOR-game strategy $\cA'$ (respectively, AND-game strategy $\cA''$) such that, for any input $(x, y)$, with probability $4^{-2Q^{*, \server}_{\epsilon_0, \epsilon_1}(f)}$, $\cA'$ (respectively, $\cA''$) outputs $f(x, y)$ with probability at least $1-\epsilon_{f(x, y)}$ (i.e. it outputs $1$ with probability at least $1-\epsilon_1$ and $0$ with probability at least $1-\epsilon_0$);
%
%
%
otherwise (with probability $1-4^{-2Q^{*, \server}_{\epsilon_0, \epsilon_1}(f)}$), $\cA'$ outputs $0$ and $1$ with probability $1/2$ each (respectively, $\cA''$ outputs $0$ with probability $1$).
\end{lemma}
Roughly speaking, the above lemma compares two cases: in the ``good'' case $\cA'$  outputs the correct vaue of $f(x, y)$ with high probability (the probability controlled by $\epsilon_0$ and $\epsilon_1$) and in the ``bad case'' $\cA'$ simply outputs a random bit. It shows that if $Q^{*, \server}_{\epsilon_0, \epsilon_1}(f)$  is small, then the ``good'' case will happen with a non-negligible probability. In other words, the lemma says that if $Q^{*, \server}_{\epsilon_0, \epsilon_1}(f)$ is small, then the probability that the nonlocal game players win the game will be high. 

This lemma gives us an access to several lower bound techniques via nonlocal games. For example, following the $\gamma_2$-norm techniques in \cite{LinialS09,Sherstov11,LeeZ10} and the recent method of \cite{fooling12}, we show one- and two-sided error lower bounds for many problems on the server model (in particular, we can obtain lower bounds in general forms as in \cite{Razborov03,Sherstov11,LeeZ10}). These lower bounds match the two-party model lower bounds.

\paragraph{Graph Problems and Reductions between Server-Model Problems (Details in Section~\ref{sec:server graph})} To bring the hardness in the Server model to the distributed setting, we have to prepare hardness for the right problems in the Server model so that it is easy to translate to the distributed setting. In particular, the problems that we need are the following graph problems.

\begin{definition}[Server-Model Graph Problems]\label{def:server model graph problems}
Let $G$ be a graph of $n$ nodes\footnote{To avoid confusion, throughout the paper we use $G$ to denote the input graph in the Server model and $N$ and $M$ to denote the distributed network and its subnetwork, respectively, unless specified otherwise. For any graph $H$, we use $V(H)$ and $E(H)$ to denote the set of nodes and edges in $H$, respectively.}. We partition edges of $G$ to $E_C(G)$ and $E_D(G)$, which are given to David and Carol, respectively. The two players have to determine whether $G$ has some property, e.g., $G$ is a Hamiltonian cycle ($\ham_n$)\footnote{$\ham_n$ is used for the Hamiltonian cycle verification problem in the Server models, where $n$ denotes the size of input graphs,  and $\ham(N)$ is used for the Hamiltonian cycle verification problem on a distributed network $N$ (defined in Section~\ref{sec:problems}).}, a spanning tree ($\st_n$), or is connected ($\conn_n$). For the purpose of this paper in proving lower bounds for distributed algorithms, we restrict the problem and assume that in the case of the Hamiltonian cycle problem $E_G(C)$ and $E_D(C)$ are both perfect matchings.
\end{definition}
%
%
%
We also consider the gap version in the case of communication complexity. The notion of $\delta$-far is slightly different from the distributed setting (cf. Section~\ref{sec:problems}) in that we can add {\em any} edges to $G$ instead of adding {\em only} edges in $N$ to $M$.
The {\em main challenge} in showing hardness results  for these graph problems is that some of them, e.g. Hamiltonian cycle and spanning tree verification, are not known to be hard, even in the classical two-party model (they are left as open problems in \cite{DasSarmaHKKNPPW11}). To get through this, we derive several new reductions (using novel gadgets) to obtain this:

\begin{theorem}{\sc (Server-Model Lower Bounds for $\ham_n$)}\label{theorem:graph lower bound server model}\label{theorem:main cc}\label{THEOREM:MAIN CC}
There  exist some constants $\eps,\beta > 0$ such that for any $n$, 
$Q^{*, \server}_{\epsilon,\epsilon}(\ham_n)$ and $Q^{*, \server}_{0, \epsilon}((\beta n)\mbox{-}\ham_n)$ are $\Omega(n)$.
\end{theorem}
We prove Theorem~\ref{theorem:graph lower bound server model} using elementary (but intricate) gadget-based reductions. Thus, no knowledge in quantum computing is required to understand this proof. 
Theorem~\ref{theorem:graph lower bound server model} also leads to lower bounds that are new even in the classical two-party model. We discuss this in Section~\ref{sec:additional results}.

\paragraph{Quantum Simulation Theorem: From Server Model to Distributed Algorithms (Details in Section~\ref{sec:from server to distributed})} 

%

To show the role of the Server model in proving distributed algorithm lower bounds, we prove a {\em quantum version} of the {\em Simulation Theorem} of \cite{DasSarmaHKKNPPW11} (cf. Section \ref{sec:from server to distributed})
which shows that the hardness of graph problems of our interest in the Server model implies the hardness of these problems in the  quantum distributed setting (the theorem below holds for several graph problems but we state it only for the Hamiltonian Cycle verification problem since it is sufficient for our purpose):

\begin{theorem}[Quantum Simulation Theorem]\label{theorem:from server to distributed}\label{THEOREM:FROM SERVER TO DISTRIBUTED}
For any $B$, $L$, $\Gamma\geq \log L$, $\beta\geq 0$ and $\epsilon_0, \epsilon_1>0$, there exists a $B$-model quantum network $N$ of diameter $\Theta(\log L)$ and $\Theta(\Gamma L)$ nodes such that if $Q_{\epsilon_0, \epsilon_1}^{*, N}((\beta\Gamma)\mbox{-}\ham(N))\leq \frac{L}{2}-2$ then $Q_{\epsilon_0, \epsilon_1}^{*, \server}((\beta\Gamma)\mbox{-}\ham_\Gamma)= O((B\log L)Q_{\epsilon_0, \epsilon_1}^{*,N}((\beta\Gamma)\mbox{-}\ham(N)))$. \end{theorem}
%
%
%
%
In words, the above theorem states that if there is an $(\epsilon_0, \epsilon_1)$-error quantum distributed algorithm that solves the Hamiltonian cycle verification problem on $N$ in at most $(L/2)-2$ time, i.e.
$Q_{\epsilon_0, \epsilon_1}^{*, N}(\ham(N))\leq (L/2)-2\,,$
then the $(\epsilon_0, \epsilon_1)$-error communication complexity in the Server model of the Hamiltonian cycle problem on $\Gamma$-node graphs is  
$Q_{\epsilon_0, \epsilon_1}^{*, \server}(\ham_\Gamma)= O((B\log L)Q_{\epsilon_0, \epsilon_1}^{*, N}(\ham(N)))\,.$
The same statement also holds for its gap version ($(\beta\Gamma)\mbox{-}\ham(N)$).
We note that the above theorem can be extended to a large class of graph problems. 
\danupon{One generalization of this which is not yet done is that we can state the theorem for all graphs with some certain properties. This will unnecessarily make the paper more complicated but will be useful as a reference in the future.}
The proof of the above theorem does not need any knowledge in quantum computing to follow. In fact, it can be viewed as a simple modification of the Simulation Theorem in the classical setting \cite{DasSarmaHKKNPPW11}. The main difference, and the most difficult part to get our Quantum Simulation Theorem to work, is to realize that we must start from the Server model instead of the two-party model.

\subsection{Quantum Distributed Lower Bounds}
\label{sec:results} \label{sec:quantum distributed lower bound results}

%

We present specific lower bounds for various fundamental verification and optimization graph problems. Some of these bounds are new even in the classical setting. To the best of our knowledge, our bounds are the first non-trivial lower bounds for fundamental global problems.

\paragraph{1. Verification problems} We prove a {\em tight} two-sided error quantum lower bound of $\tilde\Omega(\sqrt{n})$ time, where $n$ is the number of nodes in the distributed network and $\tilde \Theta(x)$ hides $\poly\log x$, for the {\em Hamiltonian cycle} and {\em spanning tree verification problems}. Our lower bound holds even in a network of small ($O(\log n)$) diameter.
\begin{theorem}[Verification Lower Bounds]\label{theorem:main verification}\label{THEOREM:MAIN VERIFICATION}
For any $B$ and large $n$, there exists $\epsilon>0$ and a $B$-model $n$-node network $N$ of diameter $\Theta(\log n)$ such that any $(\epsilon, \epsilon)$-error quantum algorithm with prior entanglement for Hamiltonian cycle and spanning tree verification on $N$ requires $\Omega(\sqrt{\frac{n}{B\log n}})$ time. That is, $Q^{*, N}_{\epsilon, \epsilon}(\ham(N))$ and $Q^{*, N}_{\epsilon, \epsilon}(\st(N))$ are $\Omega(\sqrt{\frac{n}{B\log n}})$.
\end{theorem}
Our bound implies a new bound on the classical setting which answers the open problem in \cite{DasSarmaHKKNPPW11}, and is the {\em first randomized lower bound} for both graph problems, subsuming the deterministic lower bounds for Hamiltonian cycle verification \cite{DasSarmaHKKNPPW11}, spanning tree verification \cite{DasSarmaHKKNPPW11} and minimum spanning tree verification \cite{KorKP11}. It is also shown in \cite{DasSarmaHKKNPPW11} that \ham can be reduced to several problems via deterministic classical-communication reductions. Since these reductions can be simulated by quantum protocols, we can use these reductions straightforwardly to show that all lower bounds in \cite{DasSarmaHKKNPPW11} hold even in the quantum setting.
\begin{corollary}\label{corollary:main verification}
The statement in Theorem~\ref{theorem:main verification} holds for the following verification problems: Connected component, spanning connected subgraph, cycle containment, $e$-cycle containment, bipartiteness, $s$-$t$ connectivity, connectivity, cut, edge on all paths, $s$-$t$ cut and least-element list. (See \cite{DasSarmaHKKNPPW11} and Appendix~\ref{sec:formal definition of network} for definitions.)
\end{corollary}
 Fig.~\ref{fig:main_verification} compares our results with previous results for verification problems.

\begin{figure*}
\center
\small
\begin{tabular}{m{1cm}|m{4cm}|m{5cm}|m{5cm}}
  \cline{2-4}
  &\bf Problems & \bf Previous results &\bf Our results\\
  \cline{2-4}
  \hline
  \multirow{3}{1cm}{\begin{sideways}\makecell[c]{\parbox{2cm}{\center \tiny \bf $B$-model distributed network}}\end{sideways}}&\ham, \st, {\sf MST} verification & $\Omega(\sqrt{n/B\log n})$ deterministic, classical communication \cite{DasSarmaHKKNPPW11,KorKP11} & \multirow{2}{5cm}{$\Omega(\sqrt{n/B\log n})$ two-sided error, quantum communication with entanglement} \\
  \cline{2-3}
  &\conn and other verification problems from \cite{DasSarmaHKKNPPW11} & $\Omega(\sqrt{n/B\log n})$ two-sided error, classical communication \cite{DasSarmaHKKNPPW11} & \\
  \cline{2-4}
  &$\alpha$-approx {\sf MST} and other optimization problems from \cite{DasSarmaHKKNPPW11} & $\Omega(\sqrt{n/B\log n})$ Monte Carlo, classical communication for $W=\Omega(\alpha n)$ \cite{DasSarmaHKKNPPW11} & $\Omega(\min(\sqrt{n}, W/\alpha)/\sqrt{B\log n})$ Monte Carlo, quantum communication with entanglement\\
  \hline
  \hline
  \multirow{2}{1cm}{\begin{sideways}\makecell[c]{\parbox{1.6cm}{\center \tiny \bf Com\-mu\-ni\-cation Complexity}}\end{sideways}}&\ham, \st, and other verification problems & $\Omega(n)$ one-sided error,  classical communication \cite{RazS95} & $\Omega(n)$ two-sided error,  quantum communication with entanglement\\
  \cline{2-4}
  & \gap-\ham, \gap-\st, \gap-\conn, and other gap problems for $\Omega(n)$ gap& unknown & $\Omega(n)$ one-sided error,  quantum communication with entanglement\\ 
  \hline
\end{tabular}
\caption{\small\it Previous and our new lower bounds. We note that $n$ is the number of nodes in the network in the case of distributed network and the number of nodes in the input graph in the case of communication complexity.}\label{fig:main_verification}
\end{figure*}


\paragraph{2. Optimization problems} We show a  {\em tight} $\tilde \Omega(\min(W/\alpha, \sqrt{n}))$-time lower bound for any $\alpha$-approximation quantum randomized (Monte Carlo and Las Vegas) distributed algorithm for the MST problem.
\begin{theorem}[Optimization Lower Bounds]\label{theorem:main optimization}\label{THEOREM:MAIN OPTIMIZATION}
For any $n$, $B$, $W$ and $\alpha<W$ there exists $\epsilon>0$ and a $B$-model $\Theta(n)$-node network $N$ of diameter $\Theta(\log n)$ and weight aspect ratio $W$ such that any $\epsilon$-error $\alpha$-approximation quantum algorithm with prior entanglement for computing the minimum spanning tree problem on $N$ requires $\Omega(\frac{1}{\sqrt{B\log n}}\min(W/\alpha$, $\sqrt{n}))$ time.
\end{theorem}
This result generalizes the bounds in \cite{DasSarmaHKKNPPW11} to the quantum setting. Moreover, this lower bound implies the same bound in the classical model, which improves \cite{DasSarmaHKKNPPW11} (see Fig.~\ref{fig:current distributed bounds}) and matches the deterministic upper bound of $O(\min(W/\alpha, \sqrt{n}))$ resulting from a combination of Elkin's $\alpha$-approximation $O(W/\alpha)$-time deterministic algorithm \cite{Elkin06} and Peleg and Rubinovich's $O(\sqrt{n})$-time exact deterministic algorithm \cite{GarayKP93,KuttenP98} in the classical communication model. Thus this bound is tight up to a $\Theta(\sqrt{B \log n})$ factor. It is the first bound that is {\em tight for all values of the aspect ratio $W$}. Fig.~\ref{fig:current distributed bounds} compares our lower bounds with previous bounds.
\begin{figure}[t]
  \centering
  \includegraphics[width=\linewidth, clip=true, trim= 0cm 0.7cm 0cm 4cm]{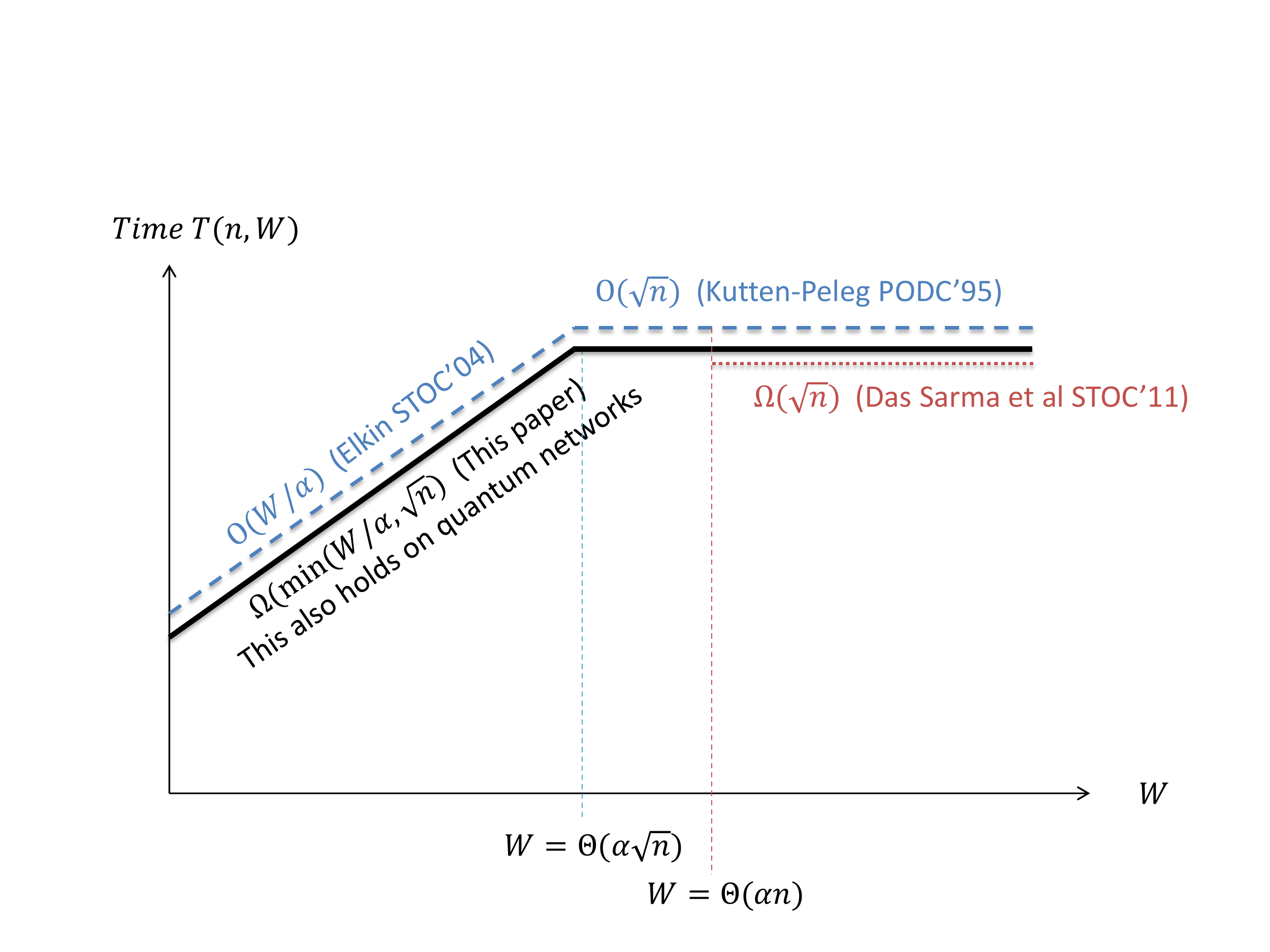}
  \caption{\small\it Previous and our new bounds (cf. Theorem~\ref{theorem:main optimization} and Corollary \ref{corollary:main optimization}) for approximating the MST problem in distributed networks when $N$ and $\alpha$ are fixed. The dashed line (in blue) represents the deterministic upper bounds (Algorithms). The dotted line (in red) is the previous lower bound for randomized algorithms. The solid line (in black) represents the bounds shown in this paper. Note that the previous lower bounds hold only in the classical setting while the new lower bounds hold in the quantum setting even when entanglement is allowed.}
  \label{fig:current distributed bounds}
\end{figure}
By using the same reduction as in \cite{DasSarmaHKKNPPW11}, our bound also implies that all lower bounds in \cite{DasSarmaHKKNPPW11} hold even in the quantum setting.
\begin{corollary}
\label{corollary:main optimization}
The statement in Theorem~\ref{theorem:main optimization} also holds for the following problems: minimum spanning tree, shallow-light tree, $s$-source distance, shortest path tree,  minimum routing cost spanning tree, minimum cut, minimum $s$-$t$ cut, shortest $s$-$t$ path and generalized Steiner forest. (See \cite{DasSarmaHKKNPPW11} and Appendix~\ref{sec:formal definition of network} for definitions.)
\end{corollary}
%

\subsection{Additional Results: Lower Bounds on Communication Complexity}\label{sec:additional results}
In proving the results in previous subsections, we prove several bounds on the Server model. Since the Server model is stronger than the standard communication complexity model (as discussed in Subsection~\ref{sec:techniques}), we obtain lower bounds in the communication complexity model as well. Some of these lower bounds are new even in the classical setting. In particular, our bounds in Theorem~\ref{theorem:graph lower bound server model} lead to the following corollary. (Note that we use  $Q^{*, cc}_{\epsilon_0, \epsilon_1}(\cP_n)$ to denote the communication complexity of verifying property $\cP$ of $n$-node graphs on the standard quantum communication complexity model with entanglement.)
\begin{corollary}\label{corollary:main cc}
For any $n$ and some constants $\epsilon, \beta>0$, $Q^{*, cc}_{\epsilon,\epsilon}({\sf P}_n)=\Omega(n),$  and $Q^{*, cc}_{0, \epsilon}((\beta n)-{\sf P}_n)\geq Q^{*, \server}_{0, \epsilon}((\beta n)-{\sf P}_n) =\Omega(n)$, where ${\sf P}_n$ can be any of the following verification problems: Hamiltonian cycle, spanning tree, connectivity, $s$-$t$ connectivity, and bipartiteness.
\end{corollary}
%
%
%
%
%
%
To the best of our knowledge, the lower bounds for Hamiltonian cycle and spanning tree verification problems are the first two-sided error lower bounds for these problems, even in the classical two-party setting (only nondeterministic, thus one-sided error, lower bounds are previously known \cite{RazS95}). The bounds for Bipartiteness and $s$-$t$ connectivity follow from a reduction from Inner Product given in \cite{BabaiFS86}, and a lower bound for Connectivity was recently shown in \cite{graphcomm12}.
%
%
%
We note that we prove the gap versions via a reduction from recent lower bounds in \cite{fooling12} and observe new lower bounds for the gap versions of Set Disjointness and Equality.
%
%


\section{Other Related Work}
\label{sec:role}\label{sec:related}

While our work focuses on  solving graph problems in quantum distributed networks, there are several prior works focusing on other distributed computing problems (including communication complexity in the two-party or multiparty communication model) using quantum effects.
We note that fundamental distributed computing problems such as leader election and byzantine agreement have been shown to solved better using quantum
phenomena (see e.g., \cite{DenchevP08,TaniKM05,ben-or}).
Entanglement has  been used to reduce the amount of communication of a specific function of input data  distributed among 3 parties \cite{cleve} (see also the work of \cite{buhrman1, wolf1,shma} on multiparty quantum communication complexity).

There are several results showing that quantum communication complexity in the two-player model can be more efficient than classical randomized communication complexity (e.g. \cite{buhrman2, raz}). These results also easily extend to the so-called number-in-hand multiparty model (in which players have separate inputs).
As of now no separation between quantum and randomized communication complexity is known
in the number-on-the-forehead multiparty model, in which players' inputs overlap.
Other papers concerning quantum distributed computing include \cite{rohrig,ChlebusKS10,KobayashiMT09,KobayashiMT10,PalKK03,GaertnerBKCW08}.

\section{Conclusion and Open Problems}
In this paper, we derive several lower bounds for important network problems in a quantum distributed network. We show that quantumness does not really help in obtaining faster distributed algorithms for fundamental problems such as minimum spanning tree, minimum cut, and shortest paths. Our approach gives a uniform way to prove lower bounds for various problems. Our technique closely follows the Simulation Theorem introduced by Das Sarma et al. \cite{DasSarmaHKKNPPW11}, which shows how to use the two-party communication complexity to prove lower bounds for distributed algorithms. The main difference of our approach is the use of the Server model. We show that many problems that are hard in the quantum two-party communication setting (e.g. \ipmodthree) are also hard in the Server model, and show new reductions from these problems to graph verification problems of our interest. Some of these reductions give tighter lower bounds even in the classical setting. 

Since the technique of Das Sarma et al. can be used to show lower bounds of many problems that are not covered in this paper (e.g. \cite{FrischknechtHW12,HolzerW12,NanongkaiDP11,LenzenP13,DasSarmaMP13,Ghaffari14,Censor-HillelGK13}), it is interesting to see if these lower bounds remain valid in the quantum setting. Since most of these problems rely on a reduction from the set disjointness problem, the main challenge is to obtain new reductions that start from problems that are proved hard on the Server model such as \ipmodthree. One problem that seems to be harder than others is the random walk problem \cite{NanongkaiDP11,DasSarmaNPT13} since the previous lower bound in the classical setting requires a {\em bounded-round} communication complexity \cite{NanongkaiDP11}. Proving lower bounds for the random walk problem thus requires proving a bounded-round communication complexity in the Server model as the first step. This requires different techniques since the nonlocal games used in this paper destroy the round structure of protocols. 

\danupon{Moreover, it is not clear if the known proofs work in the presence of entanglement.}


It is also interesting to better understand the role of the Server model: 
%
Can we derive a quantum two-party version of the Simulation Theorem, thus eliminating the need of the Server model?
%
Is the Server model {\em strictly} stronger than the two-party quantum communication complexity model?  
Also, it will be interesting to explore {\em upper} bounds in the quantum setting:  Do quantum distributed algorithms help in solving other fundamental graph problems ?

%



\newpage
\part{Proofs}

\section{Server Model Lower Bounds via Nonlocal Games (Lemma~\ref{lemma:xor and simulate server})} \label{sec:basic lower bounds server model}
In this section, we prove Lemma~\ref{lemma:xor and simulate server} which shows how to use nonlocal games to prove server model lower bounds. Then, we use it to show server-model lower bounds for two problems called {\em Inner Product mod 3} (denoted by $\ipmodthree_n$) and  {\em Gap Equality} with parameter $\delta$ (denoted by $\delta$-$\eq_n$). These lower bounds will be used in the next section.

Our proof makes use of the relationship between the server model and {\em nonlocal games}. In such games, Alice and Bob receive input $x$ and $y$ from some distribution $\pi$ that is known to the players. As usual they want to compute a boolean function $f(x, y)$ such as Equality or Inner Product mod 3. However, they cannot communicate to each other. Instead, each of them can send one bit, say $a$ and $b$, to a referee. The referee then combines $a$ and $b$ using some function $g$ to get an output of the game $g(a,b)$. The goal of the players is to come up with a strategy (which could depend on distribution $\pi$ and function $g$) that maximizes the probability that $g(a, b)=f(x, y)$. We call this the {\em winning probability}. One can define different nonlocal games based on what function $g$ the referee will use. Two games of our interest are {\em XOR}- and {\em AND}-games where $g$ is {\em XOR} and {\em AND} functions, respectively.

Our proof follows the framework of proving two-party quantum communication complexity lower bounds via nonlocal games (see, e.g., \cite{LeeS09,LeeZ10,fooling12}). The key modification is the following lemma which shows that the XOR- and AND-game players can make use of an efficient server-model protocol to guarantee a good winning probability.

{
\renewcommand{\thetheorem}{\ref{lemma:xor and simulate server}}
\begin{lemma}[Restated]
For any boolean function $f$ and $\epsilon_0, \epsilon_1\geq 0$, there is an (two-player nonlocal) XOR-game strategy $\cA'$ (respectively, AND-game strategy $\cA''$) such that, for any input $(x, y)$, with probability $4^{-2Q^{*, \server}_{\epsilon_0, \epsilon_1}(f)}$, $\cA'$ (respectively, $\cA''$) outputs $f(x, y)$ with probability at least $1-\epsilon_{f(x, y)}$ (i.e. it outputs $1$ with probability at least $1-\epsilon_1$ and $0$ with probability at least $1-\epsilon_0$);
%
%
%
otherwise (with probability $1-4^{-2Q^{*, \server}_{\epsilon_0, \epsilon_1}(f)}$), $\cA'$ outputs $0$ and $1$ with probability $1/2$ each (respectively, $\cA''$ outputs $0$ with probability $1$).
%
\end{lemma}
\addtocounter{theorem}{-1}
}

\begin{proof}
\danupon{I didn't try say that we can assume that the communication is $00\ldots 0$. We can try this later on if it helps.}\danupon{TO DO: Change proof in Appendix accordingly.}
We prove the lemma in a similar way to the proof of Theorem 5.3 in \cite{LeeS09} (attributed to Buhrman).
Consider any boolean function $f$. Let $\cA$ be any $(\epsilon_0, \epsilon_1)$-error server-model protocol for computing $f$ with communication complexity $T$. We will construct (two-player) nonlocal XOR-games and AND-games strategies, denoted by $\cA'$ and $\cA''$, respectively, that {\em simulate} $\cA$. First we simulate $\cA$ with an additional assumption that there is a ``fake server'' that sends messages to players (Alice and Bob) in the nonlocal games, but the two players in the games do not send any message to the fake server. Later we will eliminate this fake server. We will refer to parties in the server model as Carol, David, and the {\em real} server, while we call the nonlocal game players Alice, Bob, and the {\em fake} server.

%
Using teleportation (where we can replace a qubit by two classical bits when there is an entanglement; see, e.g., \cite{NielsenChuangBook}), it can be assumed that Carol and David send $2T$ {\em classical} bits to the real server instead of sending $T$ qubits (the server can set up the necessary entanglement for free). Assume that, on an input $(x, y)$, Carol and David send bits $c_t$ and $d_t$ in the $t^{th}$ round, respectively. (We note one detail here that in reality $c_t$ and $d_t$, for all $t$, are random variables. We will ignore this fact here to illustrate the main idea. More details are in Appendix~\ref{sec:ipmodthree full}.\danupon{Don't forget to point to appendix in the full version.})

Now, Alice, Bob and the fake server generate shared random strings $a_1\ldots a_t$ and $b_1\ldots b_t$ (this can be done since their states are entangled). These strings serve as a ``guessed'' communication sequence of $\cA$. Alice, Bob and the fake protocol simulate Carol, David and the real protocol, respectively. However, in each round $t$, instead of sending bit $c_t$ that Carol sends to the real server, Alice simply looks at $a_t$ and continues playing if her guessed communication is the same as the real communication, i.e. $c_t=a_t$. Otherwise, she ``aborts'': In the XOR-game protocol $\cA'$ she outputs $0$ and $1$ with probability 1/2 each, and in the AND-game protocol $\cA''$ she outputs $0$. Bob does the same thing with $d_t$ and $b_t$.


The fake server simply assumes it receives $a_t$ and $b_t$ and continues sending messages to Alice and Bob. Observe that the probability of never aborting is $4^{-T}$ (i.e., when the random strings $a_1\ldots a_T$ and $b_1\ldots b_T$ are the same as the communication sequences $c_1\ldots c_T$ and $d_1\ldots d_T$, respectively). If no one aborts, Alice will output Carol's output while Bob will output $0$ in the XOR-game protocol $\cA'$ and $1$ in the AND-game protocol $\cA''$. If no one aborts, Alice, Bob and the fake server perfectly simulate $\cA$ and thus output $f(x, y)$ with probability at least $1-\epsilon_{f(x, y)}$ in both protocols\footnote{That is, if $f(x, y)=0$, they output $0$ with probability at least $1-\epsilon_0$ and, if $f(x, y)=1$, they output $1$ with probability at least $1-\epsilon_1$}. Otherwise (with probability at most $1-4^{-T}$) one or both players will abort and the output will be randomly $0$ and $1$ in $\cA'$ and $0$ in $\cA''$. This is exactly what we claim in the theorem except that there is a fake server.


Now we eliminate the fake server. Notice that the fake server never receives anything from Alice and Bob. Hence we can assume that the fake server sends all his messages to Alice and Bob before the game starts (before the input is given), and those messages can be viewed as prior entanglement. We thus get standard XOR- and AND-game strategies without a fake server.
%
%
\end{proof}

Now we define and prove lower bounds for $\ipmodthree_n$ and $\delta$-$\eq_n$. In both problems Carol and David are given $n$-bit strings $x$ and $y$, respectively. In  $\ipmodthree_n$, they have to output $1$ if $(\sum_{i=1}^n x_iy_i) \mod 3 = 0$ and $0$ otherwise. In $\delta$-$\eq_n$, the players are {\em promised} that either $x=y$ or the hamming distance $\Delta(x, y)>\delta$ where $\Delta(x, y)=|\{i\mid x_i\neq y_i\}|$. They have to output $1$ if and only if $x=y$. This theorem will be used in the next section.
%
%
%
%
%
\begin{theorem}\label{theorem:basic lower bounds server model}
For some $\beta, \epsilon>0$ and any large $n$, $Q^{*, \server}_{\epsilon,\epsilon}(\ipmodthree_n)$ and $Q^{*, \server}_{0, \epsilon}((\beta n)\mbox{-}\eq_n)$ are $\Omega(n)$.
\end{theorem}

Now we give a high-level idea of the proof of Theorem~\ref{theorem:basic lower bounds server model} (see Appendix~\ref{sec:ipmodthree full} for detail).

To show that $Q^{*, \server}_{\epsilon, \epsilon}(\ipmodthree_n)=\Omega(n)$, we use an XOR-game strategy $\cA'$ and $\epsilon_0=\epsilon_1=\epsilon$ from Lemma~\ref{lemma:xor and simulate server}. Using this we can extend the theorem of Linial and Shraibman \cite{LinialS09} from the two-party model to the server model and show that $Q^{*, \server}_{\epsilon, \epsilon}(f)$ is lower bounded by an {\em approximate $\gamma_2$ norm}: $Q^{*, \server}_{\epsilon, \epsilon}(f)=\Omega(\log \gamma_2^{2\epsilon}(A_f))$ for some matrix $A_f$ defined by $f$.
%
%
%
%
%
%
Using $f=\ipmodthree_n$, one can then extend the proof of Lee and Zhang \cite[Theorem 8]{LeeZ10} to lower bound $\log \gamma_2^{2\epsilon}(A_f)$ by an {\em approximate degree} $\deg_{2\epsilon}(f')$ of some function $f'$. Finally, one can follow the proof of Sherstov \cite{Sherstov11} and Razborov \cite{Razborov03} to prove that $\deg_{2\epsilon}(f')=\Omega(n).$ Combining these three steps, we have 
%
$$Q^{*, \server}_{\epsilon, \epsilon}(\ipmodthree_n)=\Omega(\log \gamma_2^{2\epsilon}(A_{\ipmodthree_n}))=\Omega(\deg_{2\epsilon}(f'))=\Omega(n).$$
We note that this technique actually extends all lower bounds we are aware of on the two-party model (e.g. those in \cite{Razborov03,Sherstov11,LeeZ10}) to the server model.


%

To prove that $Q^{*, \server}_{0, \epsilon}((\beta n)\mbox{-}\eq_n)=\Omega(n)$ for some $\beta, \epsilon>0$, we use an AND-game strategy $\cA''$ with $\epsilon_0=0$ and $\epsilon_1=\epsilon=1/2$ from Lemma~\ref{lemma:xor and simulate server}. We adapt a recent result by Klauck and de Wolf \cite{fooling12}, which shows that $Q^{*, cc}_{0, 1/2}(f) \geq(\log \mbox{\rm fool}^1(f))/4-1/2$. Here $\mbox{\rm fool}^1(f)$ refers to the size of the {\em 1-fooling set} of $f$, which is defined to be a set $F=\{(x,y)\}$ of input pairs with the following properties.
\squishlist
\item If $(x,y)\in F$ then $f(x,y)=1$
\item If $(x,y),(x',y')\in F$ then $f(x,y')=0$ or $f(x',y)=0$
\squishend

We observe that the lower bound in \cite{fooling12} actually applies to AND-games as follows. Suppose Alice and Bob receive inputs $(x,y)$, then perform local measurements on a shared entangled state, and output bits $a,b$. Then the probability that $a\wedge b=1$ for a uniformly random $x,y\in F$ is at most $1/\mbox{\rm fool}^1(f)$, if the probability that $a\wedge b=1$ for $(x,y)$ with $f(x,y)=0$ is always 0.


Lemma~\ref{lemma:xor and simulate server} for the case of AND-games implies that there is an AND-game strategy $\cA''$ such that if $f(x, y)=0$ then $\cA''$ always output $0$ and if $f(x,y)=1$ then $\cA''$ outputs $1$ with probability at least $(1-\epsilon)4^{-2Q^{*,\server}_{0, \epsilon}(f)}$. This implies that $(1-\epsilon)4^{-2Q^{*,\server}_{0, \epsilon}(f)}\leq 1/\mbox{\rm fool}^1(f)$. In other words, if $\mbox{\rm fool}^1(f)=2^{\Omega(n)}$ then $Q^{*,\server}_{0, 1/2}(f)=\Omega(n)$.

All that remains is to define a good fooling set for $(\beta n)$-$\eq_n$. Fix any $1/4>\beta> 0$.
The idea is to use a good error-correcting code to construct the fooling set.  Recall that $\Delta(x, y)$ denote the Hamming distance between $x$ and $y$. Let $C$ be a set of $n$-bit strings such that the Hamming distance between any distinct $x, y\in C$ is at least $2\beta n$. Due to the Gilbert-Varshamov bound such codes $C$ exist with $|C|\geq 2^{(1-H(2\beta))n}=2^{\Omega(n)}$, where $H$ denotes the binary entropy function. Hence we have $Q^{*, \server}_{0, 1/2}((\beta n)\mbox{-}\eq_n) = \Omega(n)$.\danupon{I removed the exact number since I'm still not very sure about the constant.}

\section{Server-model Lower Bounds for {\lowercase{\large $\mathbf{\ham_n}$}} (Theorem~\ref{theorem:main cc})} \label{sec:server graph}
In this section, we prove Theorem~\ref{theorem:main cc}, which leads to new lower bounds for several graph problems as discussed in Section~\ref{sec:additional results}. The proof uses gadget-based reductions between problems on the Server model. 

{
\renewcommand{\thetheorem}{\ref{theorem:main cc}}
\begin{theorem}[Restated]
For any $n$ and some constants $\epsilon, \beta>0$, 
\begin{align}
Q^{*, \server}_{\epsilon,\epsilon}(\ham_n) &=\Omega(n)~~~\mbox{and}~~~\label{eq:serverham}\\
Q^{*, \server}_{0, \epsilon}((\beta n)\mbox{-}\ham_n) &=\Omega(n)\label{eq:serverham2}\,.
\end{align}
\end{theorem}
\addtocounter{theorem}{-1}
}

\begin{figure}
\center
   \includegraphics[clip=true, trim=1cm 5cm 5cm 6cm, width=0.95\linewidth]{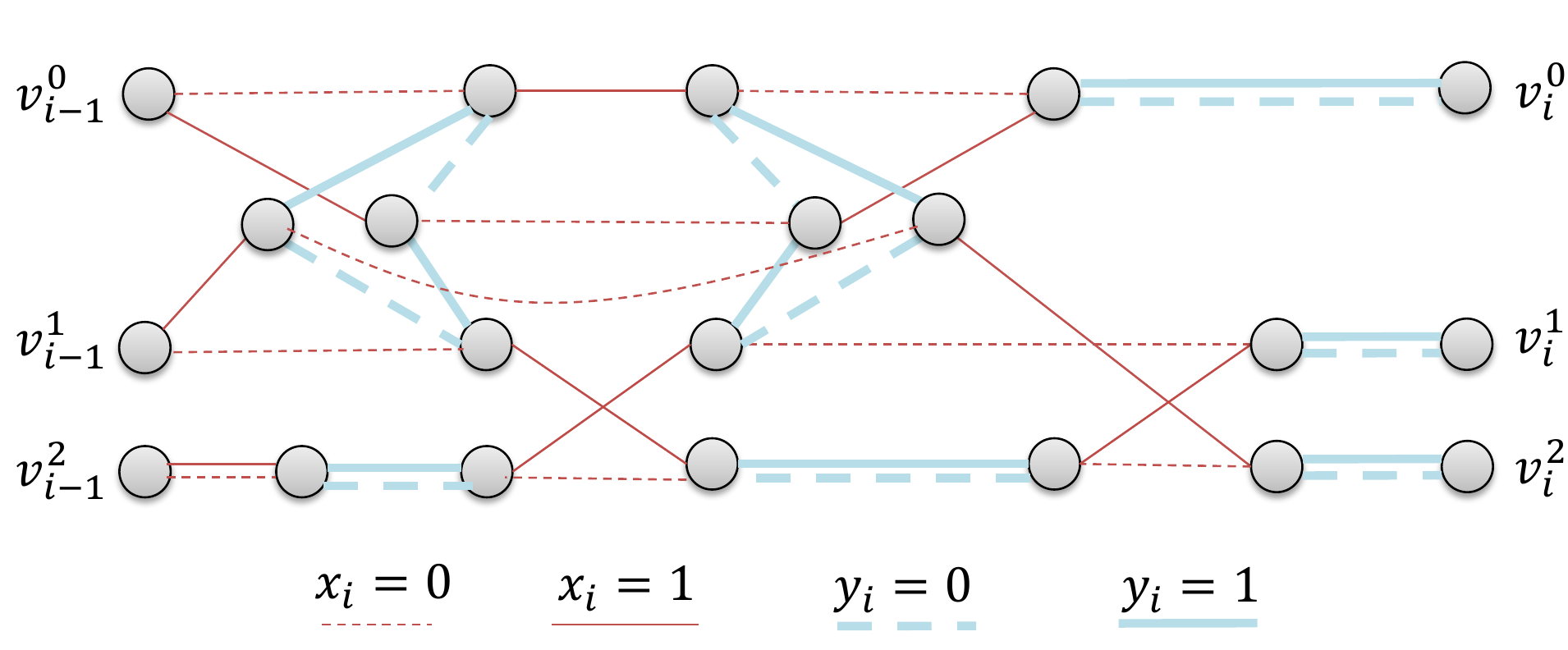}
   \caption{The construction of gadget $G_i$. If $x_i=0$ then Alice adds dashed thin edges (in red); otherwise she adds solid thin edges (in red). If $y_i=0$ then Bob adds dashed thick edges (in blue); otherwise he adds solid thick edges (in blue).
  }
  \label{fig:Hamiltonian Gadget}
\end{figure}


We first sketch the lower bound proof of $Q^{*, \server}_{\epsilon, \epsilon}(\ham_{n})$ and show later how to extend to the gap version. More detail can be found in Section~\ref{sec:server graph Appendix}.
We will show that for any $0\leq \epsilon\leq 1$ and some constant $c$, $Q^{*, \server}_{\epsilon,\epsilon}(\ipmodthree_n)=O(Q^{*, \server}_{\epsilon,\epsilon}(\ham_{cn}))$. The theorem then immediately follows from the fact that $Q^{*, \server}_{\epsilon,\epsilon}(\ipmodthree_n)=\Omega(n)$ (cf. Theorem~\ref{theorem:basic lower bounds server model}).

Let $x=x_1\ldots x_n$ and $y=y_1\ldots y_n$ be the input of $\ipmodthree_n$. We construct a graph $G$ which is an input of $\ham_{cn}$ as follows. The graph $G$ consists of $n$ gadgets, denoted by $G_1, \ldots, G_n$. For any $1\leq i\leq n-1$, gadgets $G_i$ and $G_{i+1}$ share exactly three nodes denoted by $v_i^0, v_i^1, v_i^2$. Each gadget $G_i$ is constructed based on the values of $x_i$ and $y_i$ as outlined in Fig.~\ref{fig:Hamiltonian Gadget}. The following observation can be checked by drawing $G_i$ for all cases of $x_i$ and $y_i$  (as in Fig.~\ref{fig:Hamiltonian Detail}).

\begin{figure}
\subfigure[$x_iy_i=00$]{
\includegraphics[clip=true, trim=1cm 5cm 5cm 6cm, width=0.45\linewidth]{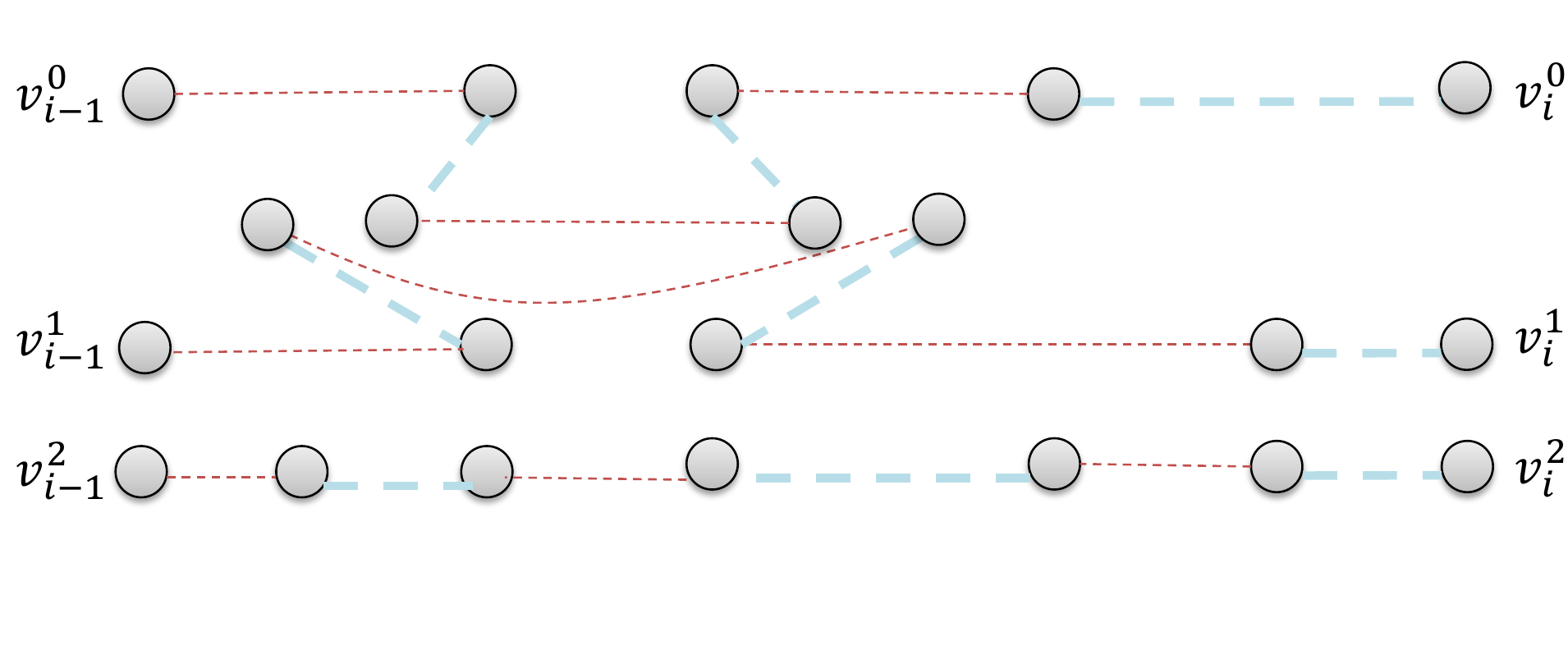}
\label{fig:ham00}
}
\subfigure[$x_iy_i=01$]{
\includegraphics[clip=true, trim=1cm 5cm 5cm 6cm, width=0.45\linewidth]{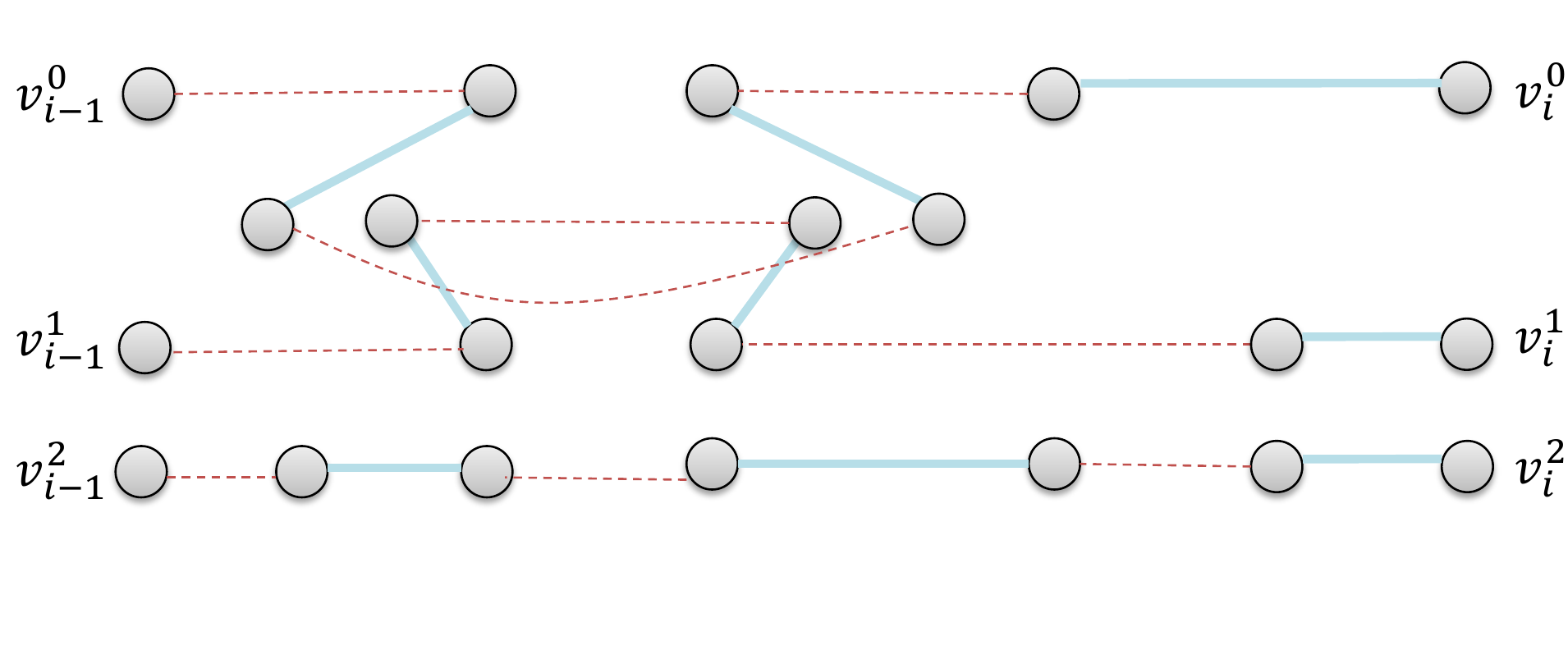}
\label{fig:ham01}
}\\
\subfigure[$x_iy_i=10$]{
\includegraphics[clip=true, trim=1cm 5cm 5cm 6cm, width=0.45\linewidth]{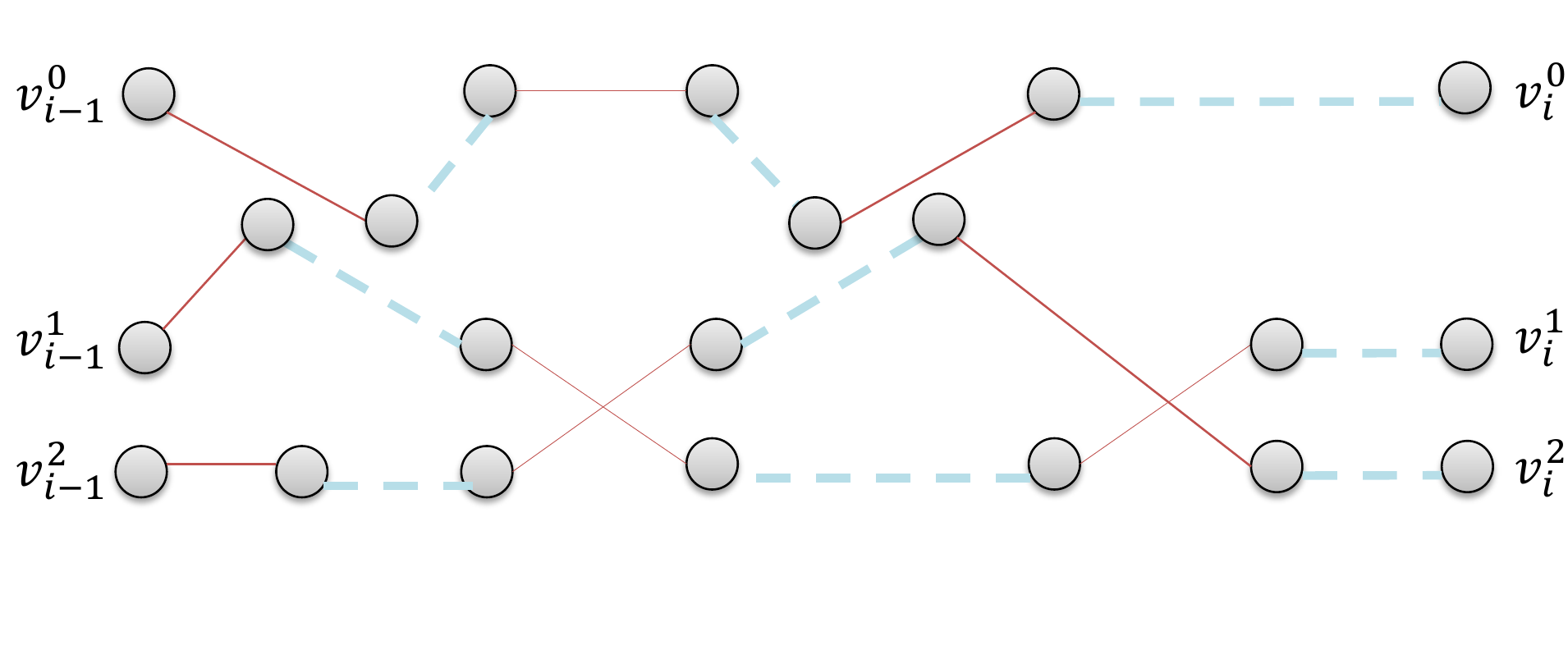}
\label{fig:ham10}
}
\subfigure[$x_iy_i=11$]{
\includegraphics[clip=true, trim=1cm 5cm 5cm 6cm, width=0.45\linewidth]{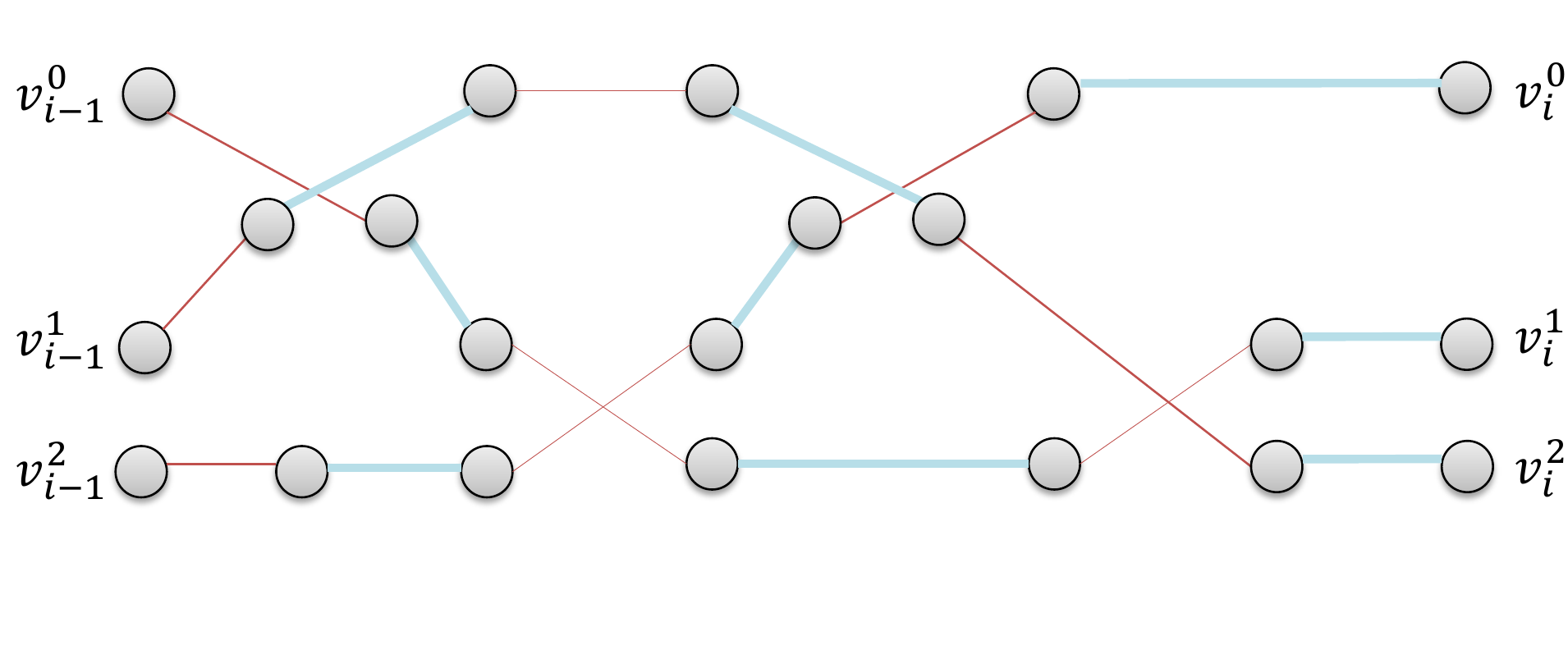}
\label{fig:ham11}
}
  \caption{Gadget $G_i$ for different values of $x_i$ and $y_i$. The main observation is that if $x_i\cdot y_i=0$ then $G_i$ consists of paths from $v_{i-1}^j$ to $v_i^j$ for all $0\leq j\leq 2$. Otherwise, it consists of paths from $v_{i-1}^j$ to $v_i^{(j+1) \mod 3}$.}\label{fig:Hamiltonian Detail}
\end{figure}

\begin{observation}\label{observation:ham}
For any value of $(x_i, y_i)$, $G_i$ consists of three paths where $v_{i-1}^j$ is connected by a path to $v_i^{(j+x_i\cdot y_i)\mod 3}$, for any $0\leq j\leq 2$. Moreover, Alice's (respectively Bob's) edges, i.e. thin (red) lines (respectively thick (blue) lines)  in Fig.~\ref{fig:Hamiltonian Gadget}, form a matching that covers all nodes except $v_i^j$ (respectively $v_{i-1}^j$) for all $0\leq j\leq 2$.
\end{observation}

\begin{figure}
\center
  \includegraphics[clip=true, trim=1cm 6cm 2cm 3cm, width=0.9\linewidth]{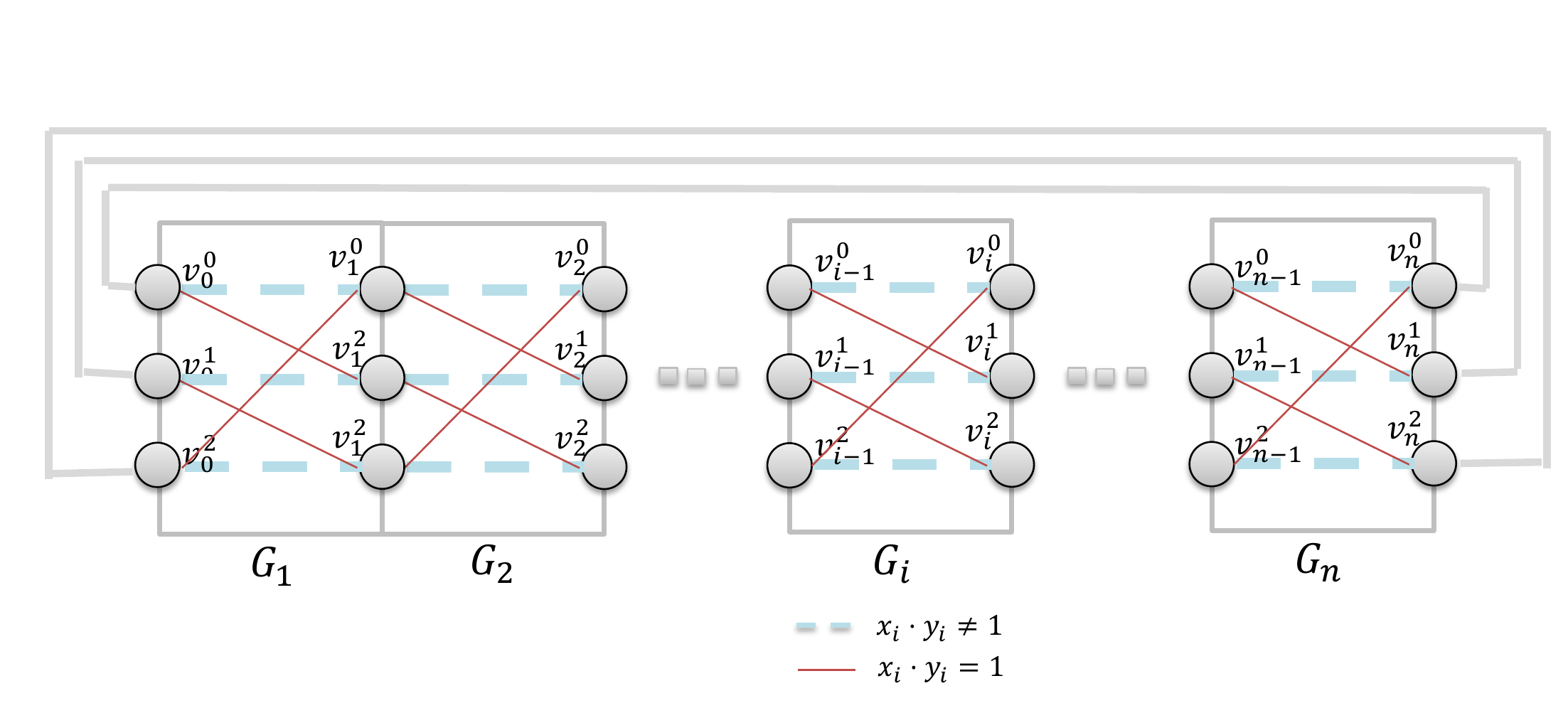}
  \caption{The graph $G$ consists of gadgets $G_1, \ldots G_n$. The solid thick edges (in gray) linking between $v_0^j$ and $v_n^j$, for $0\leq j\leq 2$ represent the fact that $v_0^j=v_n^j$. Lines that appear in each gadget $G_i$ depicts what we observe in Observation~\ref{observation:ham}: solid thin lines (in red) represent paths that will appear in $G_i$ if $x_i\cdot y_i=0$, and dashed thick lines (in blue) represent paths that will appear in $G_i$ if $x_i\cdot y_i=1$.}\label{fig:Hamiltonian Graph}
\end{figure}

Thus, when we put all gadgets together, graph $G$ will consist of three paths connecting between nodes in $\{v_0^j\}_{0\leq j\leq 2}$ on one side and nodes in $\{v_n^j\}_{0\leq j\leq 2}$ on the other. How these paths look like depends on the structure of each gadget $G_i$ which depends on the value of $x_i\cdot y_i$.  The following lemma follows trivially from Observation~\ref{observation:ham}.

\begin{lemma}\label{lem:ham}
$G$ consists of three paths $P^0$, $P^1$ and $P^2$ where for any $0\leq j\leq 2$, $P^j$ has $v_0^j$ as one end vertex and $v_n^{(j+ \sum_{1\leq i\leq n} x_i\cdot y_i)\mod 3}$ as the other.
\end{lemma}

%

%
Now, we complete the description of $G$ by letting $v_0^j=v_n^j$ for all $0\leq j\leq 2$.
It then follows that $G$ is a Hamiltonian cycle if and only if $\sum_{1\leq i\leq n} x_i\cdot y_i \mod 3\neq 0$ (see Fig.~\ref{fig:Hamiltonian Graph}; also see Lemma~\ref{lem:ham if and only if appendix} and Fig.~\ref{fig:ham three graph results Appendix} in Section~\ref{sec:server graph Appendix}).
%
%
%
Thus we can check that $\sum_{1\leq i\leq n} x_i\cdot y_i \mod 3$ is zero or not by checking whether $G$ is a Hamiltonian cycle or not. Theorem~\ref{theorem:graph lower bound server model} now follows from Theorem~\ref{theorem:basic lower bounds server model}.

\begin{figure}
  \begin{center}
    \includegraphics[clip=true, trim=0cm 5.3cm 16.5cm 6.5cm, width=0.4\linewidth]{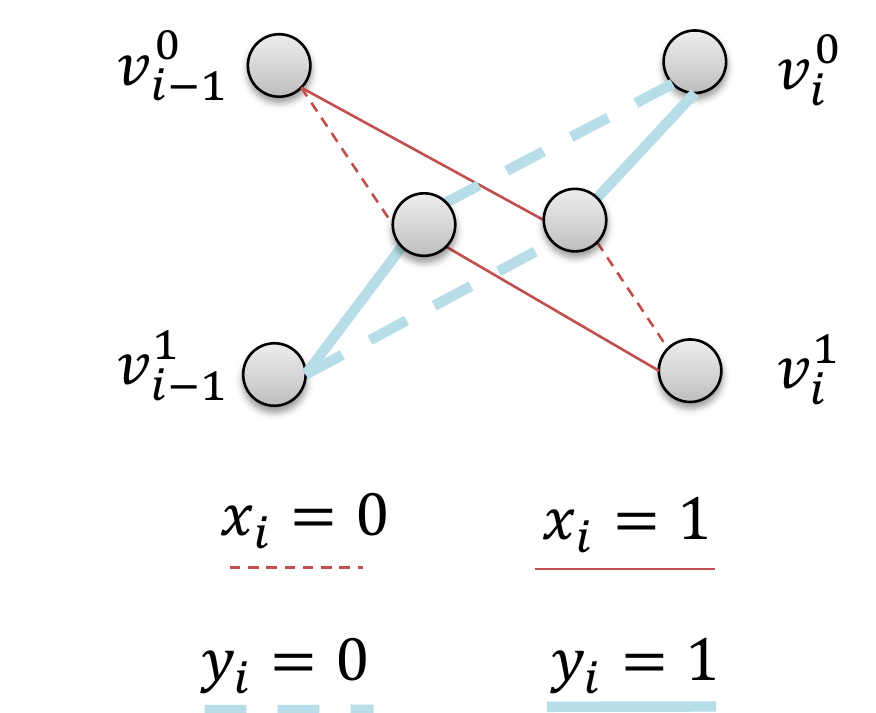}
  \end{center}
  \caption{Gadget $G_i$ to reduce from $(\beta n)$-$\eq_n$ to $(\beta n)$-$\ham_n$.
   }\label{fig:gap ham gadget}
\end{figure}
To show a lower bound of $Q^{*, \server}_{0, \epsilon}((\beta n)\mbox{-}\ham_n)$, we reduce from $(\beta n)$-$\eq_n$ in a similar way using gadget $G_i$ shown in Fig.~\ref{fig:gap ham gadget}. For any $1\leq i \leq n-1$, gadget $G_i$ and $G_{i+1}$ share $v_i^0$ and $v_i^1$, and we let $v_0^0=v_0^1$ and $v_n^0=v_n^1$. It is straightforward to show that if $x=y$, then $G$ is a Hamiltonian cycle, and if $x_{i_j}\neq y_{i_j}$ for some $i_1<i_2<\ldots <i_\delta$, then $G$  consists of $\delta$ cycles where each cycle starts at gadget $G_{i_j}$ and ends at gadget $G_{i_{j+1}}$. Note that our reduction gives a simplification of the rather complicated reduction in \cite[Section 6]{DasSarmaHKKNPPW11}.\danupon{TO DO: details in Appendix.}

\section{The Quantum Simulation Theorem (Theorem~\ref{theorem:from server to distributed})}\label{sec:from server to distributed}

In this section, we show that in the quantum setting, a server-model lower bound implies a $B$-model lower bound, as in Theorem~\ref{theorem:from server to distributed}.
{
\renewcommand{\thetheorem}{\ref{theorem:from server to distributed}}
\begin{theorem}[Restated]
For any $B$, $L$, $\Gamma\geq \log L$, $\beta\geq 0$ and $\epsilon_0, \epsilon_1>0$, there exists a $B$-model quantum network $N$ of diameter $\Theta(\log L)$ and $\Theta(\Gamma L)$ nodes such that if $Q_{\epsilon_0, \epsilon_1}^{*, N}((\beta\Gamma)\mbox{-}\ham(N))\leq \frac{L}{2}-2$ then $Q_{\epsilon_0, \epsilon_1}^{*, \server}((\beta\Gamma)\mbox{-}\ham_\Gamma)= O((B\log L)Q_{\epsilon_0, \epsilon_1}^{*,N}((\beta\Gamma)\mbox{-}\ham(N)))$. \end{theorem}
\addtocounter{theorem}{-1}
}
In words, the above theorem states that
%
if there is an $(\epsilon_0, \epsilon_1)$-error quantum distributed algorithm that solves the Hamiltonian cycle verification problem on $N$ in at most $(L/2)-2$ time, i.e.
$Q_{\epsilon_0, \epsilon_1}^{*, N}(\ham(N))\leq (L/2)-2\,,$
then the $(\epsilon_0, \epsilon_1)$-error communication complexity in the server model of the Hamiltonian cycle problem on $\Gamma$-node graphs is  
$Q_{\epsilon_0, \epsilon_1}^{*, \server}(\ham_\Gamma)= O((B\log L)Q_{\epsilon_0, \epsilon_1}^{*, N}(\ham(N)))\,.$
The same statement also holds for its gap version.
%
%
We note that the above theorem can be extended to a large class of graph problems with some certain properties. We state it for only \ham for simplicity.
\danupon{One generalization of this which is not yet done is that we can state the theorem for all graphs with some certain properties. This will unnecessarily make the paper more complicated but will be useful as a reference in the future.}

%
%
We give the proof idea here and provide full detail in Appendix~\ref{sec:full proof from server to distributed}. Although we recommend the readers to read this before the full proof and believe that it is enough to reconstruct the full proof, this proof idea can be skipped without loss of continuity.


We note again that the main idea of this theorem essentially follows the ideas developed in a line of work in \cite{PelegR00,Elkin06,LotkerPP06,KorKP11,DasSarmaHKKNPPW11}. In particular, we construct a network following ideas in \cite{PelegR00,Elkin06,LotkerPP06,KorKP11,DasSarmaHKKNPPW11}, and our argument is based on simulating the network by the three players of the server model. This idea follows one of many ideas implicit in the proof of the Simulation Theorem in \cite{DasSarmaHKKNPPW11} which shows how two players can simulate some class of networks. However, as we noted earlier, the previous proof does not work in the quantum setting, and it is still open whether the Simulation Theorem holds in the quantum setting. We instead use the server model. Another difference is that we prove the theorem for graph problems instead of problems on strings (such as Equality or Disjointness). This leads to some simplified reductions since reductions can be done easier in the communication complexity setting.



%
\begin{figure}
\centering
\includegraphics[width=0.8\linewidth, clip=true, trim= 0cm 5cm 0cm 0cm]{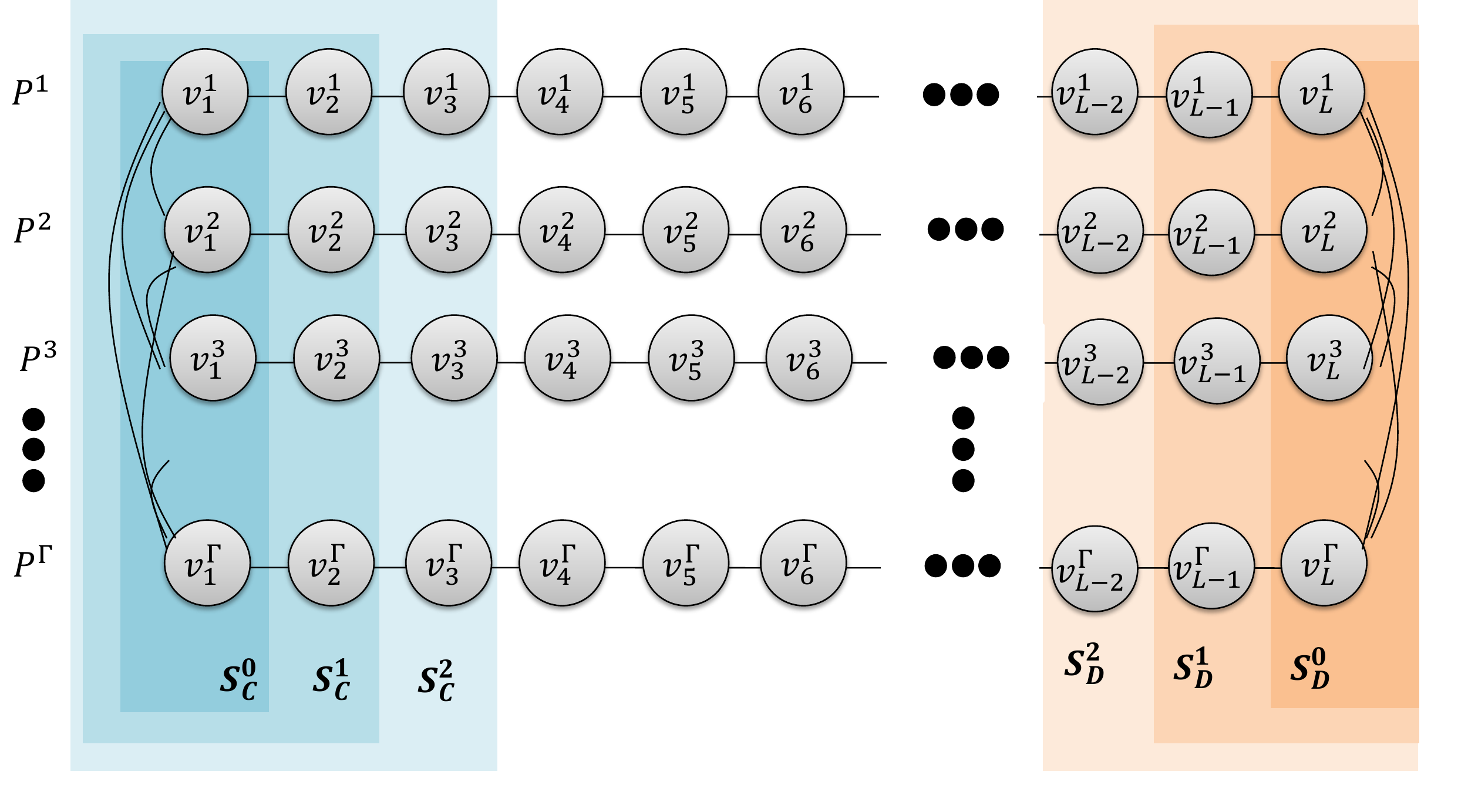}
\caption{{\small \it The network $N'$ used in the proof idea of Theorem~\ref{theorem:from server to distributed} with sets $S_C^t$ and $S_D^t$.}} \label{fig:Network one}
\end{figure}

To explain the main idea, let us focus on the non-gap version of Hamiltonian cycle verification and consider a $B$-model network $N'$ in Fig.~\ref{fig:Network one} consisting of $\Gamma$ paths, each of length $L$, where we have an edge between any pair of the leftmost (respectively, rightmost) nodes of paths. Now we will prove that if $Q_{\epsilon_0, \epsilon_1}^{*, N}(\ham(N))\leq (L/2)-2$ then $Q_{\epsilon_0, \epsilon_1}^{*, \server}(\ham_\Gamma)=0$ (i.e. no communication is needed from Carol and David to the server!). Note that this statement is stronger than the theorem statement but it is not useful since $N'$ has diameter $\Theta(L)$ which is too large. We will show how to modify $N'$ to get the desired network $N$  later.


Let paths in $N'$ be $P^1, \ldots, P^\Gamma$ and nodes in path $P^i$ be $v^i_1, \ldots, v^i_L$.
Let $\cA$ be an $(\epsilon_0, \epsilon_1)$-error quantum distributed algorithm that solves the Hamiltonian cycle verification problem on network $N'$ ($\ham(N')$) in at most $(L/2)-2$ time.

We show that Carol, David and the server can solve the Hamiltonian cycle problem on a $\Gamma$-node input graph without any communication, essentially by ``simulating'' $\cA$ on some input subnetwork $M$ corresponding to the server-model input graph $G=(U, E_C\cup E_D)$ in the following sense. When receiving $E_C$ and $E_D$, the three parties will construct a subnetwork $M$ of $N'$ (without communication) in such a way that $M$ is a Hamiltonian cycle if and only if $G=(U, E_C\cup E_D)$ is. Next, they will simulate algorithm $\cA$ in such a way that, at any time $t$ and for each node $v^i_j$ in $N'$, there will be exactly one party among Carol, David and the server that knows {\em all information that $v^i_j$ should know in order to run algorithm $\cA$}, i.e., the state of $v^i_j$ as well as the messages (each consisting of $B$ quantum bits) sent to $v^i_j$ from its neighbors at time $t$. The party that knows this information will pretend to be $v^i_j$ and apply algorithm $\cA$ to get the state of $v^i_j$ at time $t+1$ as well as the messages that $v^i_j$ will send to its neighbors at time $t+1$. We say that this party {\em owns $v^i_j$} at time $t$. Details are as follows.

Initially at time $t=0$, we let Carol own all leftmost nodes, and David own all rightmost nodes while the server own the rest, i.e. Carol, David and the server own the following sets of nodes respectively (see Fig.~\ref{fig:Network one}):
%
\begin{align}\label{eq:Simulate node set one}
\fbox{\parbox{.6\linewidth}{
\centering
$S_C^0=\{v^i_1\mid 1\leq i\leq \Gamma\}$,\\
$S_D^0=\{v_L^i \mid 1\leq i\leq \Gamma\}$,\\
$S_S^0=V(N')\setminus (S_C^0\cup S_D^0)$\,.
}}
\end{align}
After Carol and David each receive a perfect matching, denoted by $E_C$ and $E_D$ respectively, on the node set $U=\{u_1, \ldots, u_\Gamma\}$, they construct a subnetwork $M$ of $N'$ as follows. For any $i\neq j$, Carol marks $v^i_1v^j_1$ as participating in $M$ if and only if $u_iu_j\in E_C$. Similarly, David marks $v^i_Lv^j_L$ as participating in $M$ if and only if $u_iu_j\in E_D$. The server marks all edges in all paths as participating in $M$. Fig.~\ref{fig:Network two} shows an example.
%
%
We note the following observation which relies on the fact that $E_C$ and $E_D$ are perfect matchings.
%

\begin{figure}
\centering
\includegraphics[width=0.8\linewidth, clip=true, trim= 0cm 8cm 0cm 0cm]{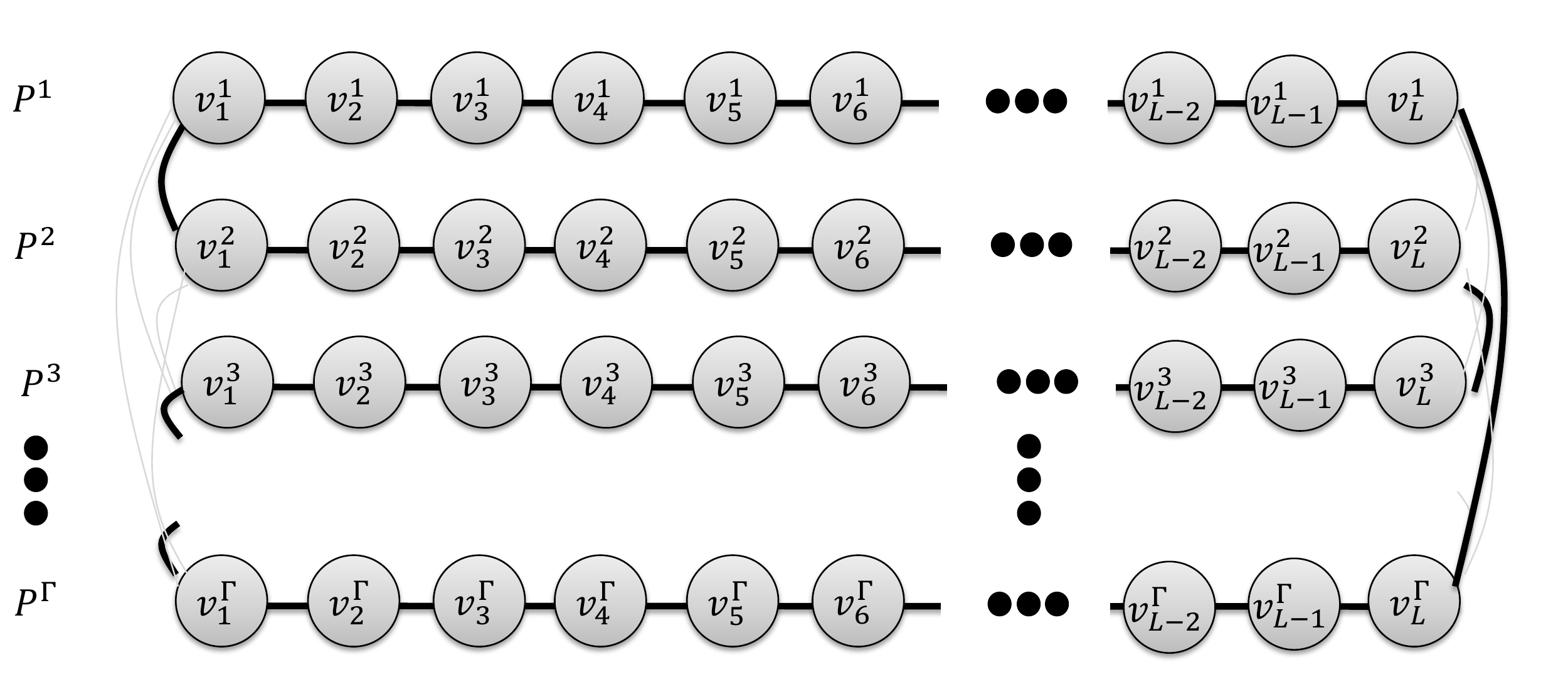}
\caption{\small \it The subnetwork $M$ when the input perfect matchings are $E_C=\{(u_1, u_2), (u_3, u_4), \ldots, (u_{\Gamma-1}, u_{\Gamma})\}$ and $E_D=\{(u_2, u_3), (u_4, u_5), \ldots, (u_\Gamma, u_1)\}$ ($M$ consists of all bold edges).}\label{fig:Network two}
\end{figure}

\begin{observation}\label{observation:G is ham if and only if M is ham}
The number of cycles in $G=(U, E_C\cup E_D)$ is the same as the number of cycles in $M$.\danupon{TO DO: Change Appendix Accordingly} 
\end{observation}


%
Now the three parties start a simulation. Recall that at time $t=0$ the three parties own nodes in the sets $S_C^0$, $S_D^0$ and $S_S^0$ as in Eq.\eqref{eq:Simulate node set one}. Our goal it to simulate $\cA$ for one time step and make sure that Carol, David and the server own the following sets respectively (see Fig.~\ref{fig:Network one}):
%
%
%
\begin{align}\label{eq:Simulate node set two}
\fbox{
\parbox{.6\linewidth}{
\centering
$S_C^1=\{v_1^i, v_2^i \mid 1\leq i\leq \Gamma\},$\\
$S_D^1=\{v_{L-1}^i, v_{L}^i \mid 1\leq i\leq \Gamma\},$\\
$S_S^1=V(N')\setminus (S_C^1\cup S_D^1)\,.$
}
}
\end{align}
To do this, the parties simulate $\cA$ on the nodes they own for one time step. This means that each of them will know the states and out-going messages at time $t=1$ (i.e., after $\cA$ is executed once) of nodes they own. Observe that although Carol knows the state of $v_1^i$, for any $i$, at time $t=1$, she is not able to simulate $\cA$ on $v_1^i$ for one more step since she does not know the message sent from $v_2^i$ to $v_1^i$ at time $t=1$. This information is known by the server who owns $v_2^i$ at time $t=0$. Thus, we let the server send this message to Carol. Additionally, for Carol to own node $v_2^i$ at time $t=1$, it suffices to let the server send the state of $v_2^i$ and the message sent from $v_3^i$ to $v_2^i$ at time $t=1$ (which are known by the server since it owns $v_2^i$ and $v_3^i$ at time $t=0$). The messages sent from the server to David can be constructed similarly. It can be checked that after this communication the three parties own nodes as in Eq.\eqref{eq:Simulate node set two} and thus they can simulate $\cA$ for one more step.


Using a similar argument as the above we can guarantee that at any time $t\leq (L/2)-2$, Carol, David and the server own nodes in the following sets respectively:
\begin{align*}
\fbox{
\parbox{.7\linewidth}{
\centering
$S_C^t=\{v^i_j \mid 1\leq i\leq \Gamma,~1\leq j\leq t+1\},$\\
$S_D^t=\{v^i_j \mid 1\leq i\leq \Gamma,~L-t\leq j\leq L\},$\\
$S_S^t=V(N')\setminus (S_C^t\cup S_D^t)\,.$
}
}
\end{align*}
Thus, if algorithm $\cA$ terminates in $(L/2)-2$ steps then Carol, David and the server will know whether $M$ is a Hamiltonian cycle or not with $(\epsilon_0, \epsilon_1)$-error by reading the output of nodes they own. By Observation~\ref{observation:G is ham if and only if M is ham}, they will know whether $G=(U, E_C\cup E_D)$ is a Hamiltonian cycle or not with the same error bound.

\begin{figure}
\centering
\includegraphics[width=0.7\linewidth]{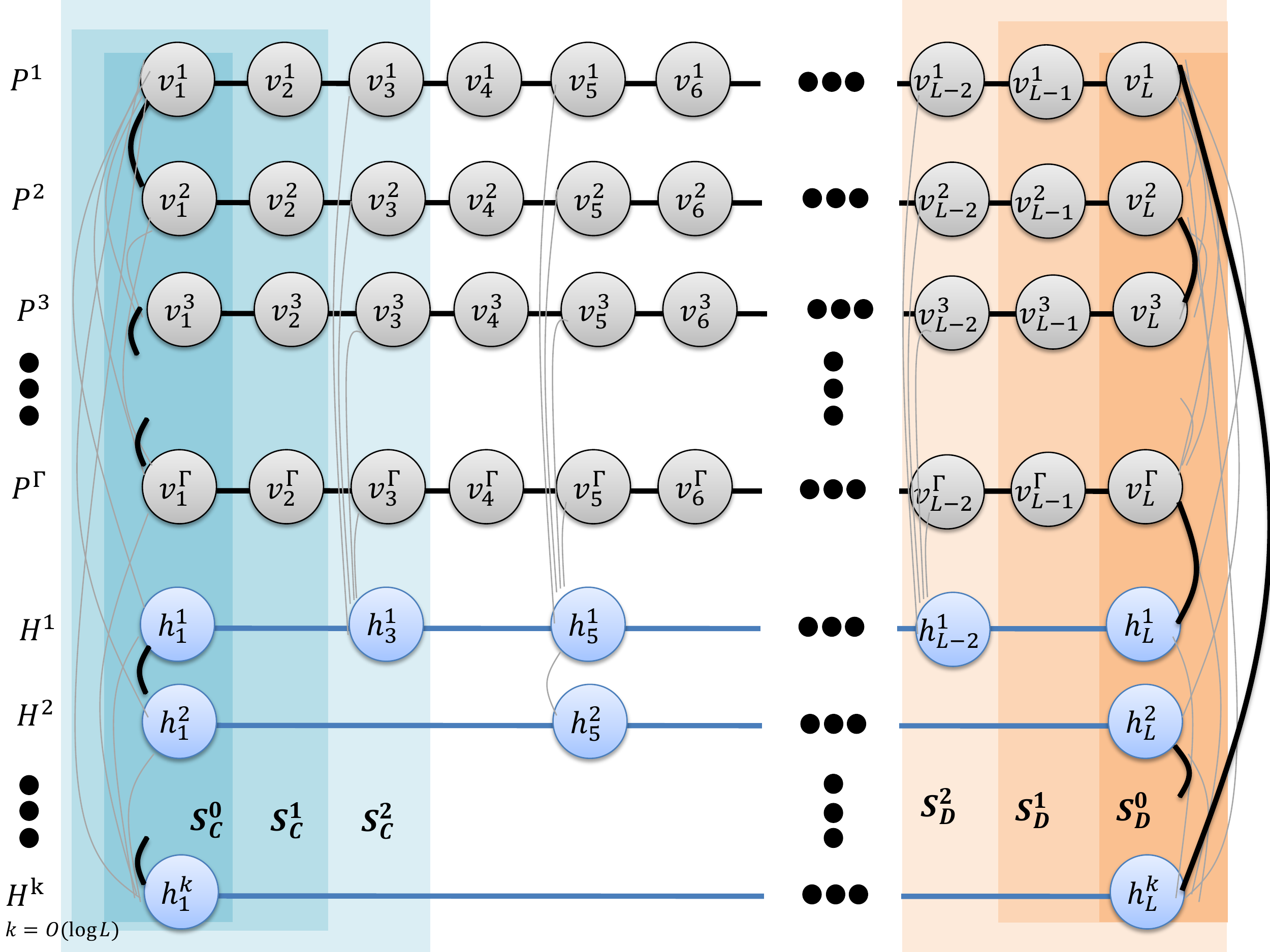}
\caption{\small \it The network $N$ consisting of network $N'$ and some ``highways'' which are paths with nodes $h^i_j$ (i.e., nodes in blue). Bold edges show an example of subnetwork $M$ when the input perfect matchings are $E_C=\{(u_1, u_2), (u_3, u_4), \ldots, (u_{\Gamma+k-1}, u_{\Gamma+k}\}$ and $E_D=\{(u_2, u_3), (u_4, u_5), \ldots, (u_{\Gamma+k}, u_1)\}$. Pale edges are those in $N$ but not in $M$.}\label{fig:Network small diameter}
\end{figure}

Now we modify $N'$ to get network $N$ of small diameter. A simple idea to slightly reduce the diameter is to add a path having half the number of nodes of other paths and connect its nodes to every other node on the other paths (see path $H^1$ in Fig.~\ref{fig:Network small diameter}). This path helps reducing the diameter from $L$ to roughly $(L/2)-2$ since any pair of nodes can connect in roughly $(L/2)-2$ hops through this path. By adding about $O(\log L)$ such paths (with $H^i$ having half the number of nodes of $H^{i-1}$) as in Fig.~\ref{fig:Network small diameter}, we can reduce the diameter to $O(\log L)$. We call the new paths {\em highways}.

We can use almost the same argument as before to prove the theorem, by modifying sets $S_C^t$, $S_D^t$ and $S_S^t$ appropriately as in Fig.~\ref{fig:Network small diameter} and consider the input graph $G=(U, E_C\cup E_D)$ of $\Gamma+k$ nodes, where $k$ is the number of highways. The exception is that now Carol and David have to speak a little. For example, observe that if the three parties want to own the states of $S_C^1$, $S_D^1$ and $S_S^1$ at time $t=1$, Carol has to send to the server the messages sent from node $h^i_1$ to its right neighbor, for all $i$. Since this message has size at most $B$, and the simulation is done for $Q^{*, N}_{\epsilon_0, \epsilon_1}(\ham(N))$ steps, Carol will send $O((B\log n)Q^{*, N}_{\epsilon_0, \epsilon_1}(\ham(N)))$ qubits to the server. David will have to send the same amount of information and thus the complexity in the server model is as claimed.

%
%
%
%
%
%


%
\section{Proof of main theorems (Theorem \ref{theorem:main verification} \& \ref{theorem:main optimization})}\label{sec:proof of main}
\subsection{Proof of Theorem \ref{theorem:main verification}}
%
%
%
{
\renewcommand{\thetheorem}{\ref{theorem:main verification}}
\begin{theorem}[Restated]
For any $B$ and large $n$, there exists $\epsilon>0$ and a $B$-model $n$-node network $N$ of diameter $\Theta(\log n)$ such that any $(\epsilon, \epsilon)$-error quantum algorithm with prior entanglement for Hamiltonian cycle and spanning tree verification on $N$ requires $\Omega(\sqrt{\frac{n}{B\log n}})$ time. That is, $Q^{*, N}_{\epsilon, \epsilon}(\ham(N))$ and $Q^{*, N}_{\epsilon, \epsilon}(\st(N))$ are $\Omega(\sqrt{\frac{n}{B\log n}})$.
\end{theorem}
\addtocounter{theorem}{-1}
}
We note from Theorem~\ref{theorem:graph lower bound server model} that 
\begin{align}
Q^{*,\server}_{\epsilon, \epsilon}(\ham_\Gamma)>c'\Gamma\label{eq:lower bound in server model}
\end{align}
for some $\epsilon>0$ and $c'>0$.  Let $c$ be the constant in the big-Oh in Theorem~\ref{theorem:from server to distributed}.
Let $L=\lfloor\frac{c'}{c}\sqrt{\frac{n}{B\log n}}\rfloor$ and $\Gamma=\lceil\sqrt{Bn\log n}\rceil$. Assume that
\begin{align}
Q_{\epsilon, \epsilon}^{*,N}(\ham(N)) &\leq L/2\leq \frac{c'}{2c}\sqrt{\frac{n}{B\log n}}\,.\label{eq:assume small time}
\end{align}
By Theorem~\ref{theorem:from server to distributed}, there is a network $N$ of diameter $O(\log L)=O(\log n)$ and $\Theta(L\Gamma)=\Theta(n)$ nodes such that
%
$Q_{\epsilon, \epsilon}^{*,\server}(\ham_\Gamma)$ $\leq (c B\log L)Q_{\epsilon, \epsilon}^{*,N}(\ham(N))$
$\leq  (c B\log L)\left(\frac{c'}{2c}\sqrt{\frac{n}{B\log n}}\right)$
$\leq c' \sqrt{Bn\log n}$
where the second equality is by Eq.~\eqref{eq:assume small time}. This contradicts Eq.\eqref{eq:lower bound in server model}, thus proving that $Q_{\epsilon, \epsilon}^{*,N}(\ham(N)) > L/2\geq \frac{c'}{4c}\sqrt{\frac{n}{B\log n}}$.

To show a lower bound of $Q_{\epsilon, \epsilon}^{*,N}(\st(N))$, let $\cA$ be an algorithm that solves spanning tree verification on $N$ in $T_\cA$ time. We can use $\cA$ to verify if a subnetwork $M$ is a Hamiltonian cycle as follows. First, we check that all nodes have degree two in $M$ (this can be done in $O(D)$ time). If not, $M$ is not a Hamiltonian cycle. If it is, then $M$ consists of cycles. Now we delete one edge $e$ in $M$ arbitrarily, and use $\cA$ to check if this subnetwork is a spanning tree. It is easy to see that this subnetwork is a spanning tree if and only if $M$ is a Hamiltonian cycle. The running time of our algorithm is $T_\cA+O(D)$. The lower bound of $Q_{\epsilon, \epsilon}^{*,N}(\ham(N))$ implies that $T_\cA=\Omega(\sqrt{\frac{n}{B\log n}})$.


\subsection{Proof of Theorem \ref{theorem:main optimization}}

{
\renewcommand{\thetheorem}{\ref{theorem:main optimization}}
\begin{theorem}[Restated]
For any $n$, $B$, $W$ and $\alpha<W$ there exists $\epsilon>0$ and a $B$-model $\Theta(n)$-node network $N$ of diameter $\Theta(\log n)$ such that any $\epsilon$-error $\alpha$-approximation quantum algorithm with prior entanglement for computing the minimum spanning tree problem on $N$ with weight function $w:E(N)\rightarrow \reals_+$ such that $\frac{\max_{e\in E(N)} w(e)}{\min_{e\in E(N)} w(e)}\leq W$ requires $\Omega(\frac{1}{\sqrt{B\log n}}\min(W/\alpha, \sqrt{n}))$ time.
\end{theorem}
\addtocounter{theorem}{-1}
}
We note from Theorem~\ref{theorem:graph lower bound server model} that
\begin{align}
Q^{*,\server}_{0, \epsilon}((\beta\Gamma)\mbox{-}\ham_\Gamma)>c'\Gamma\label{eq:lower bound in server model beta}
\end{align}
for some constant $\beta>0$, $\epsilon>0$ and $c'>0$.  Let $c$ be the constant in the big-Oh in Theorem~\ref{theorem:from server to distributed}.
Let $L=\lfloor\frac{c'}{c\sqrt{B\log n}}\min(\frac{W}{\alpha}, \sqrt{n})\rfloor$ and $\Gamma=\lceil \sqrt{B\log n}\max(\frac{n\alpha}{W}, \sqrt{n})\rceil$. We prove the following claim the same way we prove Theorem~\ref{theorem:main verification} in the previous section.

\begin{claim}
$Q_{0, \epsilon}^{*,N}((\beta\Gamma)\mbox{-}\ham) > \frac{L}{2}\geq \frac{c'}{4c}\min(W/\alpha, \sqrt{\frac{n}{B\log n}})$
\end{claim}
\begin{proof}
Assume that
\begin{align}
Q_{0, \epsilon}^{*,N}((\beta\Gamma)\mbox{-}\ham) &\leq \frac{L}{2} \leq \frac{c'}{2c}\min(W/\alpha, \sqrt{\frac{n}{B\log n}})\,.
\label{eq:assume small time beta}
\end{align}
By Theorem~\ref{theorem:from server to distributed}, there is a network $N$ of diameter $\Theta(\log L)=O(\log n)$ and $\Theta(L\Gamma)=\Theta(n)$ nodes such that
\begin{align*}
Q_{0, \epsilon}^{*,\server}((\beta\Gamma)\mbox{-}\ham_\Gamma) &\leq (c B\log L)Q_{0, \epsilon}^{*,N}((\beta\Gamma)\mbox{-}\ham)\\
&\leq  (c B\log L)(L/2)\\
&\leq \frac{c'\sqrt{B\log n}}{2}\min(\frac{W}{\alpha} ,\sqrt{n})\\
&\leq \frac{c'\sqrt{B\log n}}{2}\max(\frac{n\alpha}{W} ,\sqrt{n})\\
&\leq c' \Gamma
\end{align*}
where the second equality is by Eq.~\eqref{eq:assume small time beta} and the fourth inequality is because if $\frac{W}{\alpha}\leq \sqrt{n}$ then $\alpha\geq W/\sqrt{n}$ and thus $n\alpha/W\geq \sqrt{n}\geq W/\alpha$. This contradicts Eq.\eqref{eq:lower bound in server model beta}.
\end{proof}
Now assume that there is an $\epsilon$-error quantum distributed algorithm $\cA$ that finds an $\alpha$-approximate MST in $T_\cA$ time. We use $\cA$ to construct an $(0, \epsilon)$-error algorithm that solves $(\beta\Gamma)$-$\ham(N)$ in $T_\cA+O(D)$ time as follows. Let $M$ be the input subnetwork. First we check if all nodes have degree exactly two in $M$. If not then $M$ is not a Hamiltonian cycle and we are done. If it is then $M$ consist of one cycle or more. It is left to check whether $M$ is connected or not. To do this, we assign weight $1$ to all edges in $H$ and weight $W$ to the rest edges. We use $\cA$ to compute an $\alpha$-approximate MST $T$. Then we compute the weight of $T$ in $O(D)=O(\log n)$ rounds. If $T$ has weight at most $\alpha (n-1)$ then we say that $H$ is connected; otherwise we say that it is $(\beta \Gamma)$-far from being connected.

To show that this algorithm is $(0, \epsilon)$-error, observe that, for any $i$, if $H$ is $i$-far from being connected then the MST has weight at least $(n-1-i)+iW$ since the MST will contain at least $i$ edges of weight $W$. If $H$ is connected then the MST has weight exactly $n-1$ which means that $T$ will have weight at most $\alpha(n-1)$ with probability at least $1-\epsilon$, and we will say that $H$ is connected with probability at least $1-\epsilon$. Otherwise, if $H$ is $(\beta\Gamma)$-far from connected then $T$ always have weight at least
$$(n-1-\beta\Gamma)+\beta\Gamma W\geq \beta\Gamma W\geq \beta (\sqrt{B\log n}\max(\frac{n\alpha}{W}, \sqrt{n})) W \geq \beta \sqrt{B\log n} \frac{n\alpha}{W} W\geq \alpha n>\alpha (n-1)$$
for large enough $n$ (note that $\beta$ is a constant), and we will always say that $H$ is $(\beta \Gamma)$-far from being connected. Thus algorithm is $(0, \epsilon)$-error as claimed.

\newpage
\part{Appendix}

\appendix

\section{Detailed Definitions}\label{sec:definitions}

\subsection{Quantum Distributed Network Models}\label{sec:formal definition of network}

\subsubsection*{Informal descriptions}
We first describe a {\em general} model which will later make it easier to define some specific models we are considering. We assume some familiarity with quantum computation (see, e.g., \cite{NielsenChuangBook,Watrous11} for excellent resources). A general distributed network $N$ is modeled by a set of $n$ processors, denoted by $u_1, \ldots, u_n$, and a set of {\em bandwidth} parameters between each pair of processors, denoted by $B_{u_iu_j}$ for any $i\neq j$, which is used to bound the size of messages sent from $u_i$ to $u_j$. Note that $B_{u_iu_j}$ could be zero or infinity. To simplify our formal definition, we let $B_{u_iu_i}=\infty$ for all $i$.

In the beginning of the computation, each processor $u_i$ receives an input string $x_{i}$, each of size $b$. The processors want to cooperatively compute a global function $f(x_{1}, \ldots, x_{n})$. They can do this by communicating in {\em rounds}. In each rounds, processor $u_i$ can send a message of $B_{u_iu_j}$ bits or qubits to processor $u_j$. (Note that $u_i$ can send different messages to $u_j$ and $u_k$ for any $j\neq k$.) We assume that each processor has unbounded computational power. Thus, between each round of communication, processors can perform any computation (even solving an {\sf NP}-complete problem!). The {\em time complexity} is the minimum number of rounds needed to compute the function $f$.
We can categorize this model further based on the type of communication (classical or quantum) and computation (deterministic or randomized).

In this paper, we are interested in quantum communication when errors are allowed and nodes share entangled qubits. In particular, for any $\epsilon>0$ and function $f$, we say that a quantum distributed algorithm $\cA$ is {\em $\epsilon$-error} if for any input $(x_1, \ldots, x_n)$, after $\cA$ is executed on this input any node $u_i$ knows the value of $f(x_1, \ldots, x_n)$ correctly with probability at least $1-\epsilon$. We let $Q^{*, N}_\epsilon(N)$ denote the time complexity (number of rounds) of computing function $f$ on network $N$ with $\epsilon$-error.

In the special case where $f$ is a boolean function, for any $\epsilon_0, \epsilon_1>0$ we say that $\cA$ computes $f$ with {\em $(\epsilon_0, \epsilon_1)$-error} if, after $\cA$ is executed on any input $(x_1, \ldots, x_n)$, any node $u_i$ knows the value of $f(x_1, \ldots, x_n)$ correctly with probability at least $1-\epsilon_0$ if $f(x_1, \ldots, x_n)=0$ and with probability at least $1-\epsilon_1$ otherwise. We let $Q^{*, N}_{\epsilon_0,\epsilon_1}(N)$ denote the time complexity of computing boolean function $f$ on network $N$ with $(\epsilon_0, \epsilon_1)$-error.

Two main models of interest are the the $B$-model (also known as $\mathcal{CONGEST}(B)$) and a new model we introduce in this paper called the {\em server model}. The $B$-model is modeled by an undirected $n$-node graph, where vertices model the processors and edges model the links between the processors. For any nodes (processors) $u_i$ and $u_j$, $B_{u_iu_j}=B_{u_ju_i}=B$ if there is an edge $u_iu_j$ in the graph and $B_{u_iu_j}=B_{u_ju_i}=0$ otherwise.
%
%

In the server model, there are three processors, denoted by {\em Carol}, {\em David} and the {\em server}. In each round, Carol and David can send one bit to each other and to the server while receiving an arbitrarily large message from the server, i.e. $B_{Carol, David}=B_{David, Carol}=B_{Carol, Server}=B_{David, Server}=1$ and $B_{Server, Carol}=B_{Server, David}=\infty$.

We will also discuss the {\em two-party communication complexity model} which is simply the network of two processors called Alice and Bob with bandwidth parameters $B_{Alice,Bob}=B_{Bob,Alice}=1$. (Note that, this model is sometimes defined in such a way that only one of the processors can send a message in each round. The communication complexity in this setting might be different from ours, but only by a factor of two.)

When $N$ is the server or two-party communication complexity model, we use $Q^{*, \server}_{\epsilon}(f)$ and $Q^{*, cc}_\epsilon(f)$ instead of $Q^{*, N}_\epsilon(f)$.



\subsubsection*{Formal definitions}

\paragraph{Network States} The {\em pure state} of a quantum network of $n$ nodes with parameters $\{B_{u_iu_j}\}_{1\leq i, j\leq n}$ is represented as a vector in a Hilbert space
$$\bigotimes_{1\leq i,j\leq n} H_{u_iu_j} = H_{u_1u_1}\otimes H_{u_1u_2}\otimes \ldots \otimes H_{u_1u_n}\otimes H_{u_2u_1}\otimes \ldots \otimes H_{u_2u_n}\otimes \ldots \otimes H_{u_nu_n}$$
where $\otimes$ is the tensor product. Here,
$H_{u_iu_i}$, for any $i$, is a Hilbert space of arbitrary finite dimension representing the ``workspace'' of processor $u_i$. In particular, we let $K$ be an arbitrarily large number (thus the complexity of the problem cannot depend on $K$) and $H_{u_iu_i}$ be a $2^K$-dimensional Hilbert space.
Additionally, $H_{u_iu_j}$, for any $i\neq j$, is a Hilbert space representing the $B_{u_iu_j}$-qubit communication channel from $u_i$ to $u_j$. Its dimension is $2^{B_{u_iu_j}}$ if $B_{u_iu_j}$ is finite and $2^K$ if $B_{u_iu_j}=\infty$.


The {\em mixed state} of a quantum network $N$ is a probabilistic distribution over its pure states
\begin{align*}
\{(p_i, \ket{\psi_i})\}~~\mbox{with}~~ p_i\geq 0 ~~\mbox{and} \sum_i p_i=1\,.
\end{align*}
We note that it is sometimes convenient to represent a mixed state by a {\em density matrix} $\rho=\sum_i p_i\ket{\psi_i}\bra{\psi_i}$.

\paragraph{Initial state} 
In the model {\em without} prior entanglement, the initial (pure) state of a quantum protocol on input $(x_1, \ldots, x_n)$ is the vector
\begin{align*}
\ket{\psi^0_{x_1, \ldots, x_n}}=\bigotimes_{1\leq i,j\leq n}\ket{\psi^0_{x_1, \ldots, x_n}(i,j)}
=\ket{\psi^0_{x_1, \ldots, x_n}(1,1)}\ket{\psi^0_{x_1, \ldots, x_n}(1,2)}\ldots  \ket{\psi^0_{x_1, \ldots, x_n}(n,n)}
\end{align*}
where $\ket{\psi^0_{x_1, \ldots, x_n}(i,j)}$ for any $1\leq i, j\leq n$ is a vector in $H_{u_iu_j}$ such that $\ket{\psi^0_{x_1, \ldots, x_n}(i,i)}=\ket{x_i, 0}$ for any $i$ and $\ket{\psi^0_{x_1, \ldots, x_n}(i,j)}=\ket{0}$ for any $i\neq j$ (here, $\ket{0}$ represents an arbitrary unit vector independent of the input). Informally, this corresponds to the case where each processor $u_i$ receives an input $x_i$ and workspaces and communication channel are initially ``clear''.

With prior entanglement, the initial (pure) state is a unit vector of the form
\begin{align}
\ket{\psi^0_{x_1, \ldots, x_n}}&=\sum_w \left(\alpha_w \bigotimes_{1\leq i,j\leq n}\ket{\psi^0_{w, x_1, \ldots, x_n}(i,j)}\right)\label{eqinitial state}
\end{align}
where $\ket{\psi^0_{w, x_1, \ldots, x_n}(i,j)}$ for any $1\leq i, j\leq n$ is a vector in $H_{u_iu_j}$ such that $\ket{\psi^0_{w, x_1, \ldots, x_n}(i,i)}=\ket{x_i, w}$ for any $i$ and $\ket{\psi^0_{w, x_1, \ldots, x_n}(i,j)}=\ket{0}$ for any $i\neq j$. Here, the coefficients $\alpha_w$ are arbitrary real numbers satisfying $\sum_w \alpha_w^2=1$ that is independent of the input $(x_1, \ldots, x_n)$. Informally, this corresponds to the case where processors share entangled qubits in their workspaces.


Note that we can assume the global state of the network to be always a pure state, since any mixed state can be purified by adding qubits to the processor's workspaces, and ignoring these in later computations.


\paragraph{Communication Protocol} The communication protocol consists of rounds of {\em internal computation} and {\em communication}. In each internal computation of the $t^{th}$ round, each processor $u_i$ applies a {\em unitary transformation} to its incoming communication channels and its own memory, i.e. $H_{u_ju_i}$ for all $j$. That is, it applies a unitary transformation of the form
\begin{align}
C_{t, u_i}\otimes \left(\bigotimes_{1\leq j\leq n, k\neq i} I_{u_ju_k}\right)\label{eq:transformation}
\end{align}
which acts as an identity on $H_{u_ju_k}$ for all $1\leq j\leq n$ and $k\neq i$. At the end of the internal computation, we require the communication channel to be clear, i.e. if we would measure any communication channel in the computational basis then we would get $\ket{0}$ with probability one. This can easily be achieved by swapping some fresh qubits from the private workspace into the communication channel.
%
%
Note that the processors can apply the transformations corresponding to an internal computation simultaneously since they act on different parts of the network's state. 

To define communication, let us divide the workspace $H_{u_iu_i}$ of processor $u_i$ further to
$$H_{u_iu_i}=H_{u_iu_i, 1}\otimes H_{u_iu_i, 2}\otimes \ldots \otimes H_{u_iu_i, n}$$
where $H_{u_iu_i, j}$ has the same dimension as $H_{u_iu_j}$. The space $H_{u_iu_i, j}$ can be thought of as a place where $u_i$ prepares the messages it wants to send to $u_j$ in each round, while $H_{u_iu_i, i}$ holds $u_i$'s remaining workspace. Now, for any $j\neq i$, $u_i$ sends a message to $u_j$ simply by swapping the qubits in $H_{u_iu_i, j}$ with those in $H_{u_iu_j}$. Note that $u_i$ does not receive any information in this process since the communication channel $H_{u_iu_j}$ is clear after the internal computation. Also note that we can perform the swapping operations between any pair $i\neq j$ simultaneously since they act on different part of the network state. This completes one round of communication.\danupon{To Hartmut: I changed the network state to a pure state.} We let
\begin{align}\label{eq:network state at time t}
\ket{\psi^{t}_{x_1, \ldots, x_n}}
\end{align}
denote the network state after $t$ rounds of communication.

%

At the end of a $T$-round protocol, we compute the {\em output of processor $u_i$} as follows. We view part of $H_{u_iu_i}$ as an output space of $u_i$, i.e. $H_{u_iu_i}=H_{O_i}\otimes H_{W_i}$ for some $H_{O_i}$ and $H_{W_i}$. We compute the output of $u_i$ by measuring $H_{O_i}$ in the computational basis.\danupon{QUESTION: Do we need other types of measurement? Do we need to gives details here?} That is, if we let $K'$ be the number of qubits in $H_{O_i}$ and the network state after a $T$-round protocol be
$\psi^{T}_{x_1, \ldots, x_n}$
%
%
then, for any $w\in \{0, 1\}^{K'}$,
\begin{align*}
Pr[\mbox{Processor $u_i$ outputs $w$}]= |\braket{\psi^{T}_{x_1, \ldots, x_n}|w}|^2.
\end{align*}
Fig.~\ref{fig:general model circuit} depicts a quantum circuit corresponding to a communication protocol on three processors.
\begin{figure}
  \center
  \includegraphics[width=0.9\linewidth]{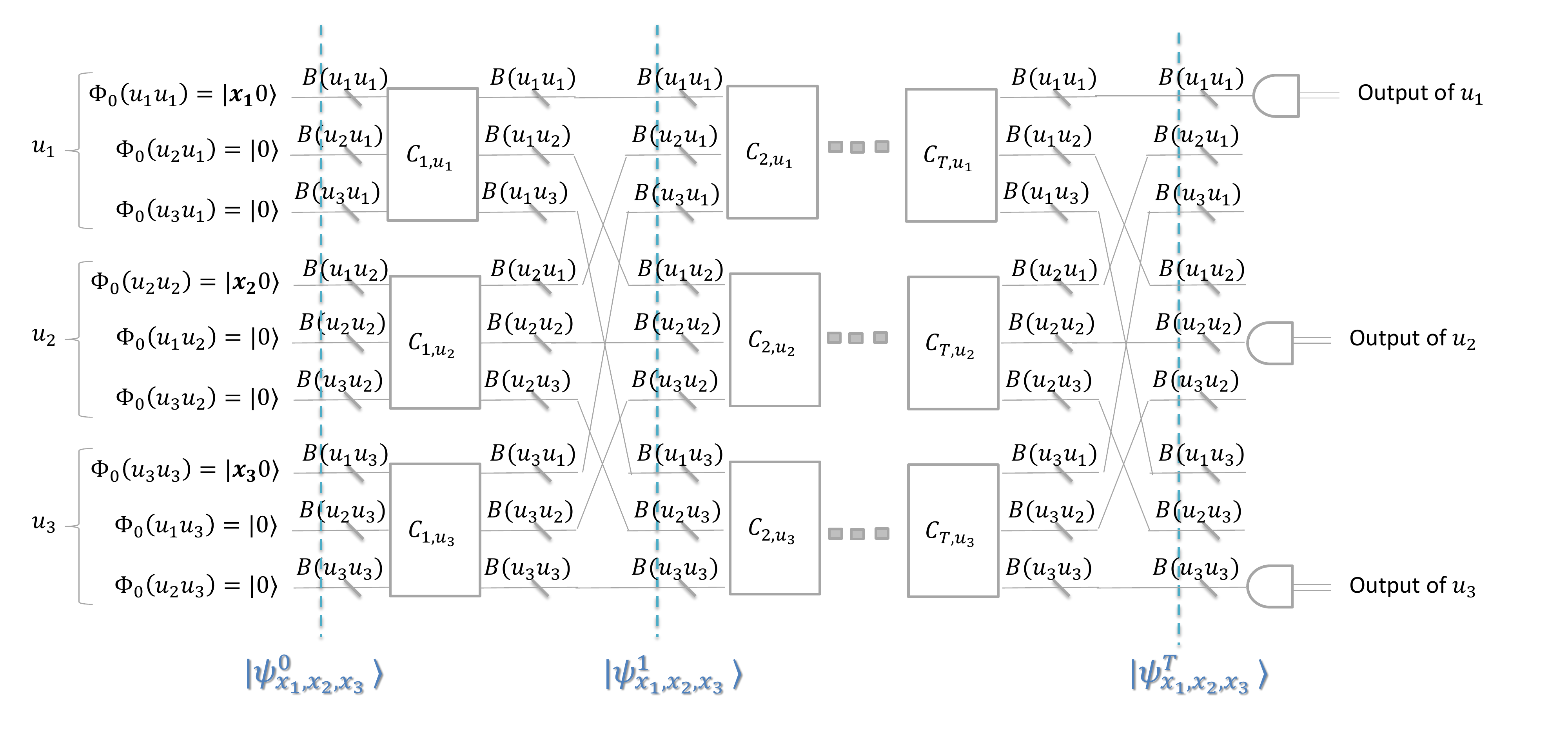}\\
  \caption{A circuit corresponding to $T$ rounds of communication on general distributed network having $3$ processors. The information flows from left to right and the line crossing each wire with a number $B_{u_iu_j}$ means that there are $B_{u_iu_j}$ qubits of information flowing through such wire. We note that the initial state in the picture is without entanglement.}\label{fig:general model circuit}
\end{figure}

\paragraph{Error and Time Complexity} For any $0\leq \epsilon \leq 1$, we say that a quantum protocol $\cA$ on network $N$ computes function $f$ with {\em $\epsilon$-error} if for any input $(x_1, \ldots, x_n)$ of $f$ and any processor $u_i$, $u_i$ outputs $f(x_1, \ldots, x_n)$ with probability at least $1-\epsilon$ after $\cA$ is executed. The {\em $\epsilon$-error time complexity} of computing function $f$ on network $N$, denoted by $Q^{*, N}_\epsilon(f)$, is the minimum $T$ such that there exists a $T$-round quantum protocol on network $N$ that computes function $f$ with {\em $\epsilon$-error}. We note that we allow the protocol to start with an entangled state. The $*$ in the notation follows the convention to contrast with the case that we do not allow prior entanglement (which is not considered in this paper).
When $N$ is the server model and two-party communication complexity model mentioned earlier, we use $Q^{*, \server}_\epsilon(f)$ and $Q^{*, cc}_\epsilon(f)$ respectively to denote the $\epsilon$-error time complexity.

If $f$ is a boolean function, we will sometimes distinguish between the error of outputting $0$ and $1$. For any $0\leq \epsilon_0, \epsilon_1\leq 1$ we say that $\cA$ computes $f$ with {\em $(\epsilon_0, \epsilon_1)$-error} if for any input $(x_1, \ldots, x_n)$ of $f$ and any processor $u_i$, if $f(x_1, \ldots, x_n)=0$ then $u_i$ outputs $0$ with probability at least $1-\epsilon_0$ and otherwise $u_i$ outputs $1$ with probability at least $1-\epsilon_1$.
The time complexity, denoted by $Q^{*, N}_{\epsilon_0, \epsilon_1}(f)$ is defined in the same way as before. We will also use $Q^{*, \server}_{\epsilon_0, \epsilon_1}(f)$ and $Q^{*, cc}_{\epsilon_0, \epsilon_1}(f)$.

%

\subsection{Distributed Graph Verification Problems} \label{subsec:distributed verification problem definition}\label{subsec:verification_network def}
In the distributed network $N$, we describe its subgraph $M$ as an input as follows. Each node $u_i$ in $N$ receives an $n$-bit binary string $x_{u_i}$ as an input. We let $x_{u_i, u_1}, \dots, x_{u_i, u_n}$ be the bits of $x_{u_i}$. Each bit $x_{u_i, u_j}$ indicates whether edge $u_iv_j$ participates in the subgraph $M$ or not.
The indicator variables must be consistent, i.e., for every edge $u_iu_j\in E(N)$, $x_{u_i,u_j}=x_{u_ju_i}$ (this is easy to verify with a single round of communication) and if there is no edge between $u_i$ and $u_j$ in $N$ then $x_{u_i,u_j}=x_{u_ju_i}=0$.

We define $M_{x_{u_1}, \ldots, x_{u_n}}$, or simply $M$, to be subgraph of $N$ having edges whose indicator variables are $1$; that is,
$$E(M)=\{(u_i, u_j)\in E \mid \forall i\neq j,\ x_{u_i,u_j}=x_{u_ju_i}=1\}.$$

 We list the following problems concerning the verification of properties of subnetwork $M$ on distributed network $N$ from \cite{DasSarmaHKKNPPW11}.

\begin{itemize}
\item
{\bf connected spanning subgraph verification:} We want to verify whether $M$ is connected and spans all nodes of $N$, i.e., every node in $N$ is incident to some edge in $M$.

\item
{\bf cycle containment verification:} We want to verify if $M$ contains a cycle. 

\item
{\bf $e$-cycle containment verification:} Given an edge $e$ in $M$ (known to vertices adjacent to it), we want to verify if $M$ contains a cycle containing $e$. 

\item
{\bf bipartiteness verification}: We want to verify whether $M$ is bipartite. 

\item
{\bf $s$-$t$ connectivity verification}: In addition to $N$ and $M$, we are given two vertices $s$ and $t$ ($s$ and $t$ are known by every vertex). We would like to verify whether $s$ and $t$ are in the same connected component of $M$. 

\item
{\bf connectivity verification}: We want to verify whether $M$ is connected.


\item
{\bf cut verification:} We want to verify whether $M$ is a cut of $N$, i.e., $N$ is not connected when we remove edges in $M$. 

\item
{\bf edge on all paths verification:} Given two nodes $u$, $v$ and an edge $e$. We want to verify whether $e$ lies on all paths between $u$ and $v$ in $M$. In other words, $e$ is a $u$-$v$ cut in $M$. 

\item
{\bf $s$-$t$ cut verification}: We want to verify whether $M$ is an $s$-$t$ cut, i.e., when we remove all edges $E(M)$ of $M$ from $N$, we want to know whether $s$ and $t$ are in the same connected component or not. 

\item
{\bf least-element list verification~\cite{cohen,KhanKMPT08}:} The input of this problem is different from other problems and is as follows. Given a distinct rank (integer) $r(v)$ to each node $v$ in the weighted graph $N$, for any nodes $u$ and $v$, we say that $v$ is the {\em least element} of $u$ if $v$ has the lowest rank among vertices of distance at most $d(u, v)$ from $u$. Here, $d(u, v)$ denotes the weighted distance between $u$ and $v$. The {\em Least-Element List} (LE-list) of a node $u$ is the set $\{\langle v, d(u, v)\rangle \mid \mbox{$v$ is the least element of u}\}$. 

In the least-element list verification problem, each vertex knows its rank as an input, and some vertex $u$ is given a set $S=\{\langle v_1, d(u, v_1)\rangle, \langle v_2, d(u, v_2)\rangle, \ldots \}$ as an input. We want to verify whether $S$ is the least-element list of $u$. 

%
%
%

\item
{\bf Hamiltonian cycle verification:} We would like to verify whether $M$ is a Hamiltonian cycle of $N$, i.e., $M$ is a simple cycle of length $n$.

\item
{\bf spanning tree verification:} We would like to verify whether $M$ is a tree spanning $N$.

\item
{\bf simple path verification:} We would like to verify that $M$ is a simple path, i.e., all nodes have degree either zero or two in $M$ except two nodes that have degree one and there is no cycle in $M$.
\end{itemize}

\subsection{Distributed Graph Optimization Problems} \label{subsec:distributed optimization problem definition}

In the graph optimization problems $\cP$ on distributed networks, such as finding MST, we are given a positive weight $\omega(e)$ on each edge $e$ of the network (each node knows the weights of all edges incident to it). Each pair of network and weight function $(N, \omega)$ comes with a nonempty set of {\em feasible solution} for problem $\cP$; e.g., for the case of finding MST, all spanning trees of $N$ are feasible solutions. The goal of $\cP$ is to find a feasible solution that minimizes or maximize the total weight. We call such solution an {\em optimal solution}. For example, the spanning tree of minimum weight is the optimal solution for the MST problem. We let $W=\max_{e\in E(N)} \omega(e)/\min_{e\in E(N)}\omega(e)$.

For any $\alpha\geq 1$, an {\em $\alpha$-approximate solution} of $\cP$ on weighted network $(N, \omega)$ is a feasible solution whose weight is not more than $\alpha$ (respectively,  $1/\alpha$) times of the weight of the optimal solution of $\cP$ if $\cP$ is a minimization (respectively, maximization) problem. We say that an algorithm $\cA$ is an $\alpha$-approximation algorithm for problem $\cP$ if it outputs an $\alpha$-approximate solution for any weighted network $(N, \omega)$.
In case we allow errors, we say that an $\alpha$-approximation $T$-time algorithm is $\epsilon$-error if it outputs an answer that is not $\alpha$-approximate with probability at most $\epsilon$ and always finishes in time $T$, regardless of the input.

Note the following optimization problems on distributed network $N$ from \cite{DasSarmaHKKNPPW11}.

\begin{itemize}
\item In the {\bf minimum spanning tree} problem~\cite{Elkin06,PelegR00}, we want to compute the weight of the minimum spanning tree (i.e., the spanning tree of minimum weight). In the end of the process all nodes should know this weight.

\item Consider a network with two cost functions associated to edges, weight and length, and a root node $r$. For any spanning tree $T$, the radius of $T$ is the maximum length (defined by the length function) between $r$ and any leaf node of $T$. Given a root node $r$ and the desired radius $\ell$, a {\bf shallow-light tree}~\cite{peleg} is the spanning tree whose radius is at most $\ell$ and the total weight is minimized (among trees of the desired radius).

\item Given a node $s$, the {\bf $s$-source distance} problem~\cite{Elkin05} is to find the distance from $s$ to every node. In the end of the process, every node knows its distance from $s$.

\item In the {\bf shortest path tree} problem~\cite{Elkin06}, we want to find the shortest path spanning tree rooted at some input node $s$, i.e., the shortest path from $s$ to any node $t$ must have the same weight as the unique path from $s$ to $t$ in the solution tree. In the end of the process, each node should know which edges incident to it are in the shortest path tree.

\item The {\bf minimum routing cost spanning tree} problem (see e.g., ~\cite{KhanKMPT08}) is defined as follows. We think of the weight of an edge as the cost of routing messages through this edge. The routing cost between any node $u$ and $v$ in a given spanning tree $T$, denoted by $c_T(u, v)$, is the distance between them in $T$. The routing cost of the tree $T$ itself is the sum over all pairs of vertices of the routing cost for the pair in the tree, i.e., $\sum_{u, v\in V(N)} c_T(u, v)$. Our goal is to find a spanning tree with minimum routing cost.

\item A set of edges $E'$ is a {\bf cut} of $N$ if $N$ is not connected when we delete $E'$. The {\bf minimum cut} problem~\cite{Elkin-sigact04} is to find a cut of minimum weight. A set of edges $E'$ is an {\em $s$-$t$ cut} if there is no path between $s$ and $t$ when we delete $E'$ from $N$.
    The {\bf minimum $s$-$t$ cut} problem is to find an $s$-$t$ cut of minimum weight.
    %

\item Given two nodes $s$ and $t$, the {\bf shortest $s$-$t$ path} problem is to find the length of the shortest path between $s$ and $t$.

\item The {\bf generalized Steiner forest} problem~\cite{KhanKMPT08} is defined as follows. We are given $k$ disjoint subsets of vertices $V_1, ..., V_k$ (each node knows which subset it is in). The goal is to find a minimum weight subgraph in which each pair of vertices belonging to the same subsets is connected. In the end of the process, each node knows which edges incident to it are in the solution.

\end{itemize}

\section{Detail of Section~\ref{sec:basic lower bounds server model}}\label{sec:ipmodthree full}

\subsection{Two-player XOR Games}
\danupon{TO DO: Define AND Games as well!}
We give a brief description of XOR games. AND game can be described similarly (their formal description is not needed in this paper). For a more detailed description as well as the more general case of nonlocal games   see, e.g., \cite{LeeS09,BrietThesis11} and references therein. An XOR game is played by three parties, Alice, Bob and a referee. The game is defined by $\cX$ and $\cY$ which is the set of input to Alice and Bob, respectively, $\pi$, a joint probability distribution $\pi:\cX\times \cY\rightarrow [0, 1]$, and a boolean function $f:\cX\times \cY\rightarrow \{0, 1\}$.

At the start of the game, the referee picks a pair $(x, y)\in \cX\times \cY$ according to the probability distribution $\pi$ and sends $x$ to Alice and $y$ to Bob. Alice and Bob then answer the referee with one-bit message $a$ and $b$. The players win the game if the value $a\oplus b$ is equal to $f(x, y)$. In other words, Alice and Bob want the XOR of their answers to agree with $f$, explaining the name ``XOR game.''

The goal of the players is to maximize the {\em bias} of the game, denoted by $\bias_\pi(f)$, which is the probability that Alice and Bob win minus the probability that they lose. In the classical setting, this is
\begin{align*}
\bias_\pi(f)&=\max_{\substack{a: \cX\rightarrow \{-1, 1\},\\ b: \cY\rightarrow \{-1, 1\}}} \sum_{(x,y)\in \cX\times \cY} (-1)^{f(x, y)}\pi(x, y)(-1)^{a(x)} (-1)^{b(y)}\\
&=\max_{\substack{a \in \{-1, 1\}^{|\cX|},\\ b \in \{-1, 1\}^{|\cY|}}} \mathbb{E}_{(x, y)\sim\pi} [ (-1)^{a(x)} (-1)^{b(y)}(-1)^{f(x, y)}]\,.
\end{align*}
%
In the quantum setting, Alice and Bob are allowed to play an {\em entangled strategy} where they may make use of an entangled state they share prior to receiving the input. That is, Alice and Bob start with some shared pure quantum state which is independent of the input and after they receive input $(x, y)$ they make some projective measurements depending on $(x, y)$ and return the result of their measurements to the referee.
%
%
Formally, an XOR entangled strategy is described by a shared (pure) quantum state $\ket{\psi}\in \mathbb{C}^{d\times d}$ for some $d\geq 1$ and a choice of projective measurements $\{A^{0}_x, A^1_x\}$ and  $\{B^{0}_y, B^1_y\}$ for all $x\in \cX$ and $y\in \cY$. When receiving input $x$ and $y$, the probability that Alice and Bob output $(a, b)\in \{0, 1\}^2$ is $\bra{\psi}A_x^a\otimes B_y^{b}\ket{\psi}$.
%
Thus, the maximum correlation can be shown to be (see \cite{BrietThesis11} for details)
\begin{align*}
\bias_\pi(f) = \max \mathbb{E}_{(x, y)\sim\pi} [\bra{\psi}(A^1_x-A^{0}_x)\otimes (B^1_y-B^{0}_y)\ket{\psi} (-1)^{f(x, y)}]
\end{align*}
where the maximization is over pure states $\ket{\psi}$ and projective measurements $\{A^{0}_x, A^1_x\}$ and  $\{B^{0}_y, B^1_y\}$ for all $x\in \cX$ and $y\in \cY$. In the rest of this paper, $\bias_\pi(f)$ always denotes the maximum correlation in the quantum setting.
%
We let
$$Q^{*, XOR}(f)=\min_\pi \bias_\pi(f)\,.$$
We note that while the players could start the game with a mixed state, it can be shown that pure entangled states suffice in order to maximize the winning probability (see, e.g., \cite{BrietThesis11}).


\subsection{From Nonlocal Games to Server-Model Lower Bounds}

\begin{lemma}[Lemma~\ref{lemma:xor and simulate server} restated]\label{lemma:xor simulate server appendix}  For any boolean function $f$ and $\epsilon_0, \epsilon_1\geq 0$, there is an XOR-game strategy $\cA'$ and AND-game strategy $\cA''$ such that, for any input $(x, y)$,
\begin{itemize}
\item with probability $4^{-2Q^{*, \server}_{\epsilon_0, \epsilon_1}(f)}$, $\cA'$ and $\cA''$ are able to simulate a protocol in the server model and hence output $f(x, y)$ with probability at least $1-\epsilon_{f(x,y)}$;
\item otherwise $\cA'$ outputs $0$ and $1$ with probability $1/2$ each, and $\cA''$ outputs $0$ with probability $1$.
\end{itemize}
\end{lemma}


\begin{proof} We have sketched the proof in Section~\ref{theorem:basic lower bounds server model}. We now provide more detail.

Let $c=Q^{*, \server}_{\epsilon_0,\epsilon_1}(f)$, i.e. Carol and David communicate with the server for $c$ rounds where each of them sends one qubit to the server per round while the server sends them messages of arbitrary size. While Alice and Bob cannot run a protocol $\cA$ in the server model since they cannot communicate to each other, we show that they can obtain the output of $\cA$ with probability $\frac{1}{4^{2c}}$. To be precise, for any input $(x, y)$ let $p_{x, y}$ and $q_{x, y}$ be the probability that $\cA(x, y)$ is zero and one respectively. We will show that
\begin{align}
\parbox{0.85\linewidth}{Alice and Bob can obtain the final state of $\cA$ with probability $4^{-2c}$ and in that case output the correct answer with high probability. If they do not obtain that state one of them will output a random bit for XOR games and one of them will output 0 for AND games.
}\label{eq:xor game two}
\end{align}

Hence the XOR game will accept with probability $\frac{1}{2}(1-4^{-2c})+4^{-2c}q_{x,y}=\frac{1}{2}+(q_{x,y}-\frac{1}{2})4^{-2c}$ and thus have a bias of
at least $4^{-2c}\cdot\min\{1/2-\epsilon_0, 1/2-\epsilon_1\}$.

The AND game will accept 1-inputs with probability at least
  $q'_{x,y}\geq \frac{q_{x, y}}{4^{2c}}$. Furthermore if $\cA$ never accepts a 0-input, then neither will the AND game.

Let us first prove Statement \eqref{eq:xor game two} with an additional assumption that there is a ``fake'' server that Alice and Bob can receive a message from but cannot talk to (we will eliminate this fake server later). We will call this a fake server to distinguish it from the ``real'' server in the server model.

First let us note the Carol and David need not talk to each other, but can send their messages to the server who can pass them to the other player.
Since the server can also set up entanglement between the three parties without cost, Carol, David and the server can use {\em teleportation} (see \cite{NielsenChuangBook} for details) and we can assume that in protocol $\cA$ Carol and David send $2$ classical bits per round to the server instead of one qubit.
These two bits are also uniformly distributed, regardless of the state of the qubit.

 Thus, for any input $(x, y)$, the messages sent by Carol and David in protocol $\cA$ will be $a, b\in \{0, 1\}^{2c}$ with some probability, say $p_{x, y, a, b}$. For simplicity, let us assume that each communication sequence $(a, b)$ leads to a unique output of $\cA$ on input $(x, y)$ (e.g., by requiring Carol and David to send their result to the server in the last round). Let $\cA(x, y, a, b)$ be the output of the protocol $\cA$ on input $(x, y)$ with communication sequence $(a, b)$. Then the probability that $\cA$ outputs zero and one is, respectively,
$$p_{x, y}=\sum_{(a, b):\ \cA(x, y, a, b)=0} p_{x, y, a , b} ~~~\mbox{and}~~~ q_{x, y}=\sum_{(a, b):\ \cA(x, y, a, b)=1} p_{x, y, a , b}\,.$$
The strategy of Alice and Bob who play the XOR and AND games is trying to ``guess'' this sequence.

In particular, Alice, Bob and the fake server will pretend to be Carol, David and the real server as follows. Before receiving the input, Alice, Bob and the fake server use their shared entanglement to create two shared random strings of length $2c$, denoted by $a'$ and $b'$, and start their initial entangled states with the same states of Carol, David and the server. In each round $t$ of $\cA$, Alice, Bob and the fake server will simulate Carol, David and the real server, respectively, as follows. Let $c_{t,1}$ and $c_{t,2}$ be two bits sent by Carol to the real server at round $t$. Alice will check whether the guessed communication sequence $a'$ is correct by checking if $c_{t,1}$ and $c_{t,1}$ are the same as $a'_{2t-1}$ and $a'_{2t}$ which are the $(2t-1)^{th}$ and $(2t)^{th}$ bits of $a'$. If they are not the same then she will `abort' which means that
\begin{itemize}
\item Alice will output $0$ and $1$ uniformly random if she is playing an XOR game, and
\item Alice will output $0$ if she is playing an AND game.
\end{itemize}

Similarly, Bob will check whether the guessed communication sequence $b'$ is correct by checking $b'_{2t-1}$ and $b'_{2t}$ with two classical bits sent by David to the server. Moreover, the fake server will pretend that it receives $a'_{2t-1}$, $a'_{2t}$, $b'_{2t-1}$ and $b'_{2t}$ to execute $\cA$ and send huge quantum messages to Alice and Bob. Alice and Bob then execute $\cA$ using these messages. After $2c$ rounds (if no player aborts), the players output the following.
\begin{itemize}
\item In XOR games, Alice will send Carol's output to the referee, and Bob will send $0$ to the referee.
\item In AND games, Alice will send Carol's output to the referee, and Bob will send $1$ to the referee.
\end{itemize}

Thus, if one or both players aborts then the output of an XOR game will be uniformly random in $\{0, 1\}$. For an AND game in case of a abort the players reject. Otherwise, the result of the XOR and AND games will be $\cA(x, y, a, b)$. The probability that Alice and Bob do not abort, given that the communication sequence of $\cA$ on input $(x, y)$ is $a$ and $b$ is $Pr[a'=a \wedge b'=b]=\frac{1}{4^{2c}}$.


This almost proves Statement \eqref{eq:xor game two} (thus the lemma) except that there is a fake server sending information to Alice and Bob in the XOR and AND game strategy. To remove the fake server, observe that we do not need an input in order to generate the messages the fake server sent to Alice and Bob. Thus, we change the strategy to the following. As previously done, before Alice, Bob and the fake server receive an input they generate shared random strings $(a', b')$ and start with the initial states of Carol, David and the real server. In addition to this, the fake server use the string $a'$ and $b'$ to generate the messages sent by the real server to Carol and David. It then sends this information to Alice and Bob. We now remove the fake server completely and mark this point as a starting point of the XOR and AND games. After Alice and Bob receive input $(x, y)$, they simulate protocol $\cA$ as before. In each round, when they are supposed to receive messages from the fake server, they read messages that the fake server sent before the game starts. Since the fake server sends the same messages, regardless of when it sends, the result is the same as before. Thus, we achieve Statement~\eqref{eq:xor game two} even when there is no fake server. This completes the proof of Lemma~\ref{lemma:xor simulate server appendix}.
\danupon{Do we need to be more formal?} 
\end{proof}

\subsection{Lower Bound for $\ipmodthree_n$}


Using the above lemma, we prove the following lemma which extends the theorem of Linial and Shraibman \cite{LinialS09} from the two-party model to the server model. Our proof makes use of XOR games as in \cite{LeeS09} (attributed to Buhrman). For any boolean function  $f: \cX\times \cY \rightarrow \{0, 1\}$, let $A_f$ be a $|\cX|$-by-$|\cY|$ matrix such that $A_f[x, y]=(-1)^{f(x, y)}$. Recall that for any matrix $A$, $\|A\|_1=\sum_{i,j}|A_{i,j}|$.

\begin{lemma}
\label{lem:xor to server}\danupon{This is a bit weaker (we use $4\epsilon$ instead of $2\epsilon$.}
For boolean function $f$ and $0\leq \epsilon< 1/4$
$$4^{2Q^{*, \server}_\epsilon(f)}\geq \max_M \frac{\langle A_f, M\rangle-2\epsilon\|M\|_1}{\gamma_2^*(M)}= \gamma_2^{2\epsilon}(A_f)\,.$$
\end{lemma}

\begin{proof}
We first prove the following claim.

\begin{claim}
\label{claim:eliminate xor}
For any boolean functions $f, g$ on the same domain, probability distribution $\pi$ and $0\leq \epsilon\leq 1,$
\begin{align*}
\bias_\pi(g)\geq \frac{\langle A_f, A_g\circ \pi\rangle - 2\epsilon}{4^{2Q^{*, \server}_\epsilon(f)}}\,.
\end{align*}
\end{claim}
\begin{proof}\danupon{This is still informal right now.}
First, suppose that when receive input $(x, y)$, Alice and Bob can somehow compute $f(x, y)$ and use this as an answer to the XOR game (e.g., Alice and Bob returns $f(x,y)$ and $1$ to the referee respectively). What is the bias this strategy can achieve? Since the probability of winning is $\sum_{\substack{x, y:f(x, y)=g(x,y)}} \pi(x, y)$, the bias is straightforwardly
$$\sum_{\substack{x, y\\f(x, y)=g(x,y)}} \pi(x, y) - \sum_{\substack{(x, y)\\ f(x, y)\neq g(x,y)}} \pi(x, y) = \sum_{x,y} \pi(x, y) A_f[x,y]A_g[x,y] = \langle A_f, A_g\circ \pi\rangle$$

Let $\cA$ be an $\epsilon$-error protocol for computing $f$ in the server model and $\cA(x, y)$ be the output of $\cA$ (which could be randomized) on input $(x, y)$. Now suppose that Alice and Bob use $\cA(x, y)$ to play the XOR game. Then the winning probability will decrease by at most $\epsilon$. Thus the bias is at least
%
\begin{align}
\langle A_f, A_g\circ \pi\rangle-2\epsilon\,.\label{eq:xor game one}
\end{align}

Now suppose that Alice and Bob use protocol $\cA'$ from Lemma~\ref{lemma:xor simulate server appendix} with $\epsilon_0=\epsilon_1=\epsilon$ to play the XOR game. With probability $1-4^{-2Q^{*, \server}_\epsilon(f)}$,  $\cA'$ will output randomly; this means that the bias is 0. Otherwise, $\cA'$ will behave as an $\epsilon$-error algorithm. Thus, we conclude from Eq.\eqref{eq:xor game one} that the bias is at least
\begin{align*}
4^{-2Q^{*, \server}_\epsilon(f)}\left(\langle A_f, A_g\circ \pi\rangle-2\epsilon\right)\,.
\end{align*}
This completes the claim.
\end{proof}

Thus, for any $\pi$
$$4^{2Q^{*, \server}_\epsilon(f)}\geq \frac{ \langle A_f, A_g\circ \pi\rangle - 2\epsilon}{\bias_\pi(g)}\,.$$
Note that $\bias_\pi(g)=\gamma^*_2(A_g\circ \pi)$ \cite{Tsirelson87} (also see \cite[Theorem 5.2]{LeeS09}). So,
$$4^{2Q^{*, \server}_\epsilon(f)}\geq \frac{\langle A_f, A_g\circ \pi\rangle - 2\epsilon}{\gamma^*_2(A_g\circ \pi)}\,.$$
Since this is true for any $\pi$ and $g$,\danupon{To to: Check this carefully again}
$$4^{2Q^{*, \server}_\epsilon(f)}\geq \max_{\substack{\pi, g}}\frac{\langle A_f, A_g\circ \pi\rangle - 2\epsilon}{\gamma^*_2(A_g\circ \pi)}=\max_M \frac{\langle A_f, M\rangle - 2\epsilon\|M\|_1}{\gamma^*_2(M)}\,.$$
This proves the first inequality in Lemma~\ref{lem:xor to server}.

For the second inequality, we use Proposition 1 in \cite{LeeZ10} (proved in \cite{LeeS09})\danupon{Find where this is in \cite{LeeS09}} which states that for any norm $\Phi$, matrix $A$ and $0\leq \alpha<1$, the $\alpha$-approximate norm is
$$\Phi^\alpha(A)=\max_W \frac{|\langle A, W\rangle| -\alpha \|W\|_1}{\Phi^*(W)}\,.$$
This means that  $\gamma_2^{2\epsilon}(A_f)=\max_M \frac{|\langle A_f, M\rangle| -2\epsilon \|W\|_1}{\gamma^*_2(W)}$ as claimed.
%
%
%
\end{proof}


For finite sets $X$,$Y$, and $E$, a function $f:E^n\rightarrow \{0, 1\}$, and a function $g: X\times Y\rightarrow E$, the {\em block composition} of $f$ and $g$ is the function $f\circ g^n: X^n\times Y^n \rightarrow \{0, 1\}$ defined by $(f\circ g^n)(x, y)=f(g(x^1, y^1), \ldots, g(x^n, y^n))$ where $(x^i, y^i)\in X\times Y$ for all $i=1, \ldots, n$. For any boolean function $f:\{0, 1\}^n\rightarrow \{0, 1\}$, let $f'$ be such that, for all $x\in \{0, 1\}^n$, $f'(x)=-1$ if $f(x)=0$ and $f'(x)=1$ otherwise. The {\em $\epsilon$-approximate degree} of $f$, denoted by $deg_\epsilon(f)$ is the least degree of a real polynomial $p$ such that $|f'(x)-p(x)|\leq \epsilon$ for all $x\in \{0, 1\}^n$.
%
%
We say that $g$ is {\em strongly balanced} if all rows and columns in the matrix $A_g$ sum to zero. For any $m$-by-$n$ matrix $A$, let $size(A)=m\times n$.
\danupon{Should we define these somewhere else?}
We now prove a ``server-model version'' of Lee and Zhang's theorem \cite[Theorem 8]{LeeZ10}. Our proof is essentially the same as their proof (also see \cite[Theorem 7.6]{LeeS09}).\danupon{To do: There is more detail in the proof of \cite[Theorem 7.6]{LeeS09}. We should add some details in there to make the proof below easier to follow.}
%
%
%

\begin{lemma}\label{lem:LeeZ}\danupon{This is slightly weaker than \cite{LeeZ10} as we have $\deg_{4\epsilon} ...$ instead of $\deg_{2\epsilon}$}
For any finite sets $X, Y$, let $g: X\times Y \rightarrow \{0, 1\}$ be any strongly balanced function. Let $f:\{0, 1\}^n\rightarrow \{0, 1\}$ be an arbitrary function. Then
$$Q^{*, \server}_\epsilon(f\circ g^n)\geq \deg_{4\epsilon}(f)\log_2 \left(\frac{\sqrt{\size{X}\size{Y}}}{\|A_g\|}\right)-O(1)$$
for any $0<\epsilon<1/4$.
\end{lemma}
\begin{proof}
\danupon{I put some simple facts here since I cannot find a place that collect them and I can point to easily.} We simply follow the proof of Lee and Zhang \cite{LeeZ10} and use Lemma~\ref{lem:xor to server} instead of Linial-Shraibman's theorem. First, we note the following inequality which follows from the definition of $\gamma_2$: For any $\delta\geq 0$ and $m$-by-$n$ matrix $A$,
$$\gamma_2^{\delta}(A)=\min_{B: \|B-A\|_\infty\leq \delta} \gamma_2(B)\geq \min_{B: \|B-A\|_\infty\leq \delta} \frac{\|B\|_{tr}}{\sqrt{size(B)}}=\frac{\|A\|_{tr}^\delta}{\sqrt{size(A)}}$$
where the first and last equalities are by definition of the approximate norm (see, e.g., \cite[Definition 4]{LeeZ10})\danupon{To do: Put this in prelim} and the inequality is by the definition of $\gamma_2$ norm (see, e.g., \cite[Definition 1]{LeeZ10}). Using $A=A_{f\circ g}$ which is an $|X|$-by-$|Y|$ matrix, we have
\begin{align}
\gamma_2^{\delta}(A_{f\circ g})\geq \frac{\|A_{f\circ g}\|_{tr}^\delta}{\sqrt{size(A_{f\circ g})}}\,.\label{eq:LeeZone}
\end{align}
%
The following claim is shown in the proof of Theorem 8 in \cite{LeeZ10}.
\begin{claim}[\cite{LeeZ10}]\danupon{Note that this is a bit weaker than \cite{LeeZ10}.}
\begin{align}
\frac{\|A_{f\circ g}\|_{tr}^{\delta}}{\sqrt{size(A_{f\circ g})}}\geq
\delta \left(\frac{\sqrt{|X\|Y|}}{\|A_g\|}\right)^{\deg_{2\delta}(f)}\,.\label{eq:LeeZtwo}
\end{align}
\end{claim}
\begin{proof}
We note the following lemma (noted as Lemma 1 in \cite{LeeZ10}) which shows that there exists a {\em dual polynomial of $f$} which is a polynomial $v$ which certifies that the approximate polynomial degree of $f$ is at least a certain value.
\begin{lemma}[\cite{Sherstov11,ShiZ09}] For any $f:\{0, 1\}^n\rightarrow \{0, 1\}$, let $f'$ be such that $f'(z)=(-1)^{f(z)}$ and $d=\deg_\delta(f)$\danupon{Same as $\deg_\delta(f')$}. Then, there exists a function $v:\{0, 1\}^n\rightarrow \reals$ such that
\begin{enumerate}
\item $\langle v, \chi_T\rangle =0$ for every character $\chi_T$ with $|T|<d$.\danupon{Should define character?}
\item $\|v\|_1=1.$
\item $\langle v, f'\rangle \geq \delta.$
\end{enumerate}
\end{lemma}
Let $v$ be a dual polynomial of $f$ as in the above lemma. We will use $B=(\frac{2^n}{size(A_g)})A_{v\circ g}$ as a ``witness matrix'', i.e.,\danupon{Did we define $size(A)$}
\begin{align}
B[x, y]=\frac{2^n}{size(A_g)^n} v(g(x_1, y_1), \ldots, g(x^n, y^n)).
\end{align}
It follows that
%
\begin{align}
\langle A_{f\circ g}, B\rangle &= \frac{2^n}{size(A_g)^n}\langle M_{f\circ g}, A_{v\circ g}\rangle\\
&= \frac{2^n}{size(A_g)^n}\sum_{x, y} f(g(x^1, y^1), \ldots, g(x^n, y^n))v(g(x^1, y^1), \ldots, g(x^n, y^n)) \\
&= \frac{2^n}{size(A_g)^n}\sum_{z\in \{0, 1\}^n} \left(f(z)v(z) \left(\sum_{\substack{x, y:\\ g(x^i, y^i)=z_i\\ \forall 1\leq i\leq n}} 1\right)\right)\\
&= \frac{2^n}{size(A_g)^n}\sum_{z\in \{0, 1\}^n} \left(f(z)v(z) \prod_{i=1}^n \left(\sum_{\substack{x^i, y^i:\\ g(x^i, y^i)=z_i}} 1\right)\right)\\
&= \frac{2^n}{size(A_g)^n}\sum_{z\in \{0, 1\}^n} \left(f(z)v(z) \left(\sum_{\substack{x', y':\\ g(x', y')=z_i}} 1\right)^n\right)\\
&= \sum_{z\in \{0, 1\}^n} f(z)v(z) \label{eq:need balance property}\\
&= \langle f, v\rangle \\
&\geq \delta\label{eq:dot product at least delta}
\end{align}
where Eq.\eqref{eq:need balance property} is because $g$ is strongly balanced which implies that $g$ is balanced, i.e. $g(x^i, y^i)$ is $0$ (and $1$) for half of its possible inputs (i.e. $size(A_g)/2$ entries of $A_g$ are $1$ (and $-1$)); thus, $$\sum_{\substack{x', y':\\ g(x', y')=z_i}} 1=size(A_g)/2.$$

A similar argument and the fact that $\|v\|_1=1$ can be used to show that \danupon{TO DO: Give details here.}
\begin{align}
\|B\|_1=1.\label{eq:L1 of B is one}
\end{align}
Now we turn to evaluate the spectral norm $\|B\|$. As shown in \cite{LeeZ10}\danupon{TO DO: WHERE? Theorem 7?}, the strongly balanced property of $g$ implies that the matrices $\chi_T\circ g^n$ and $\chi_S\circ g^n$ are orthogonal for distinct sets $S, T\subseteq \{0, 1\}^n$. Note the following fact (Fact 1 in \cite{LeeZ10}): For any matrices $A'$ and $B'$ of the same dimension, if $A'(B')^\dagger=(A')^\dagger B'=0$ then
%
%
$\|A+B\|=\max\{\|A\|, \|B\|\}$. Using this fact, we have \danupon{TO DO: Elaborate the inequalities below}
\begin{align}
\|B\| &= \frac{2^n}{size(A_g)^n} \|\sum_{T\subseteq [n]} \hat{v}_T A_{\chi_T\circ g^n}\|\\
&= \frac{2^n}{size(A_g)^n} \max_{T} |\hat{v}_T| \| \hat{v}_T A_{\chi_T\circ g^n}\| &\mbox{(by the fact above)}\\
&= \max_T 2^n |\hat{v}_T| \prod_i \frac{\|A_G^{T[i]}\|}{size(A_g)}\\
&\leq \max_{T: \hat{v}^T\neq 0} \prod_i \frac{\|A_G^{T[i]}\|}{size(A_g)}\label{eq:vT}\\
&= \left(\frac{\|A_g\|}{\sqrt{size(A_g)}}\right)^d\left(\frac{1}{size(A_g)}\right)^{n/2} \label{eq:last step of B}
\end{align}
where Eq.\eqref{eq:vT} is because $|\hat{v}_T|\leq 1/2^n$ as $\|v\|_1=1$\danupon{To do: need details} and Eq.\eqref{eq:last step of B} is because $\|J\|=\sqrt{size(A_g)}$\danupon{To do: Need detail}.

We note that for any $0\leq \epsilon<1$, norm $\Phi:\reals^n\rightarrow \reals$ and vector $v\in \reals^n$, the approximate norm is $\Phi^\epsilon(v)=\max_u \frac{|\langle v, u\rangle| - \epsilon \|u\|_1}{\Phi^*(u)}$ (see, e.g., \cite{LeeS09} and \cite[Proposition 1]{LeeZ10}). Note also that if $\Phi$ is the trace norm then its dual $\Phi^*$ is the spectral norm (this is noted in \cite{LeeZ10}). Thus,
\begin{align}
\|A_{f\circ g^n}\|^{\delta/2}_{tr} &=\max_{B'} \frac{|\langle A_{f\circ g^n}, B'\rangle| - (\delta/2) \|B'\|_1}{\|B'\|}\\
&\geq \frac{|\langle A_{f\circ g^n}, B\rangle| - \delta/2}{\|B\|} & \mbox{(by Eq.\eqref{eq:L1 of B is one})}\\
&\geq \frac{\delta - \delta/2}{\|B\|} &\mbox{(by Eq.\eqref{eq:dot product at least delta})}\\
&\geq (\delta/2)\left(\frac{\sqrt{size(A_g)}}{\|A_g\|}\right)^d\left(size(A_g)\right)^{n/2} &\mbox{(by Eq.\eqref{eq:last step of B})}\\
&\geq (\delta/2)\left(\frac{\sqrt{size(A_g)}}{\|A_g\|}\right)^d\left(\sqrt{size(A_{f\circ g})}\right)
\end{align}
%
%
where the last inequality is because $size(A_{f\circ g})=size(A_g)^n$\danupon{TO DO: Check this!}. This completes the proof of the claim.
\end{proof}
The lemma follows immediately from Eq.\eqref{eq:LeeZone} and Eq.\eqref{eq:LeeZtwo} by plugging in Lemma~\ref{lem:xor to server}:
$$4^{2Q^{*, \server}_\epsilon(f\circ g^n)}\geq \gamma_2^{2\epsilon}(A_{f\circ g^n})\geq \frac{\|A_{f\circ g}\|_{tr}^{2\epsilon}}{\sqrt{size(A_{f\circ g})}}\geq (2\epsilon)\left(\frac{\sqrt{|X\|Y|}}{\|A_g\|}\right)^{\deg_{4\epsilon}(f)}\,.$$
%
%
Lemma~\ref{lem:LeeZ} follows (the term $2\epsilon$ will contribute to the term ``$-O(1)$''\danupon{Check this!}).
\end{proof}

%


Now, we prove the lower bound for $\ipmodthree_n$. Our proof essentially follows Sherstov's proof \cite{Sherstov11} (also see \cite[Section 7.2.3]{LeeS09}). We can assume w.l.o.g. that $n$ is divisible by $4$. Consider the {\em promise version} of $\ipmodthree_n$ where any $n$-bit string input $x\in \cX$ and $y\in \cY$ has the property that for any $0\leq i\leq (n/4)-1$,
\begin{align*}
&x_{4i+1}x_{4i+2}x_{4i+3}x_{4i+4}\in \{0011, 0101, 1100, 1010\} ~~~\mbox{and}\\
&y_{4i+1}y_{4i+2}y_{4i+3}y_{4i+4}\in \{0001, 0010, 1000, 0100\}\,.
\end{align*}
Now we show that the claimed lower bound holds even in this case. This lower bound clearly implies the lower bound for the more general case of $\ipmodthree_n$ where no restriction is put on the input.

Observe that, for any $(x, y)\in \cX\times \cY$, the function $\ipmodthree$ can be written as
$$f\circ g^{n/4}(x, y)=f(g(x_1\ldots x_4, y_1\ldots y_4), g(x_5\ldots x_8, y_5\ldots y_8), \ldots, g(x_{n-3}\ldots x_n, y_{n-3}\ldots y_n))$$
where
$$g(x_{4i+1}\ldots x_{4i+4}, y_{4i+1}\ldots y_{4i+4})=(x_{4i+1}\wedge y_{4i+1}) \vee(x_{4i+2}\wedge y_{4i+2})\vee(x_{4i+3}\wedge y_{4i+3})\vee(x_{4i+4}\wedge y_{4i+4})$$
for all $0\leq i\leq (n/4)-1$, and $f(z_1, \ldots, z_{n/4})=1$ if $z_1+\ldots +z_{n/4}$ can be divided by 3 and 0 otherwise.
Note that $\ipmodthree(x, y)=f\circ g^{n/4}(x, y)$ since the promise implies that $g(x_{4i+1}\ldots x_{4i+4}, y_{4i+1}\ldots y_{4i+4})=1$ if and only if $x_{4i+1}y_{4i+1}+\ldots+x_{4i+4}y_{4i+4}=1$.
The matrix $A_g$ is
\[
A_g = \bordermatrix{~ & 0001 & 0010 & 1000 & 0100 \cr
                  0011 & -1 & -1 & 1 & 1 \cr
                  0101 & -1 & 1 & 1 & -1 \cr
                  1100 & 1 & 1 &-1 &-1  \cr
                  1010 & 1 & -1 & -1 & 1 \cr}
\]
which is clearly strongly balanced. It can be checked that this matrix has spectral norm $\|A_g\|=2\sqrt{2}$ (see, e.g., \cite[Section 7.2.3]{LeeS09}).
%
%
Moreover, by Paturi \cite{Paturi92} (see also \cite{deWolf10} and \cite[Theorem 2.6]{Sherstov11}), $\deg_{1/3}(f)=\Theta(n)$. Thus, Lemma~\ref{lem:LeeZ} implies that
\begin{align*}
Q^{*, \server}_{1/12}(f\circ g^n) &\geq \deg_{1/3}(f)\log_2 \left(\frac{\sqrt{4\times 4}}{\|A_g\|}\right)-O(1)\\
&= \deg_{1/3}(f)\log_2\sqrt{2}-O(1)\\
&= \Omega(n)\,.
\end{align*}

We note that the same technique can be used to prove many bounds in the server model similar to bounds in \cite{Razborov03,Sherstov11,LeeZ10}.

\section{Detail of Section~\ref{sec:server graph}} \label{sec:server graph Appendix}

First, let us recall that Alice and Bob construct a gadget $G_i$ using $x_i$ and $y_i$ as shown in Fig.~\ref{fig:Hamiltonian Gadget}. Fig.~\ref{fig:Hamiltonian Detail} shows how $G_i$ looks like for each possible value of $x_i$ and $y_i$. It follows immediately that $G_i$ always consist of three paths which connect $v_{i-1}^j$ to $v_{i}^{(j+x\cdot y) \mod 3}$, as in the following observation.


\begin{observation}[Observation~\ref{observation:ham} restated]\label{observation:ham appendix}
For any value of $(x_i, y_i)$, $G_i$ consists of three paths where $v_{i-1}^j$ is connected by a path to $v_i^{(j+x_i\cdot y_i)\mod 3}$, for any $0\leq j\leq 2$.
%
%
Moreover, Alice's (respectively Bob's) edges, i.e. thin (red) lines (respectively thick (blue) lines)  in Fig.~\ref{fig:Hamiltonian Gadget}, form a matching that covers all nodes except $v_i^j$ (respectively $v_{i-1}^j$) for all $0\leq j\leq 2$.
\end{observation}


Finally, we connect gadgets $G_i$ and $G_{i+1}$ together by identifying rightmost nodes of $G_i$ with leftmost nodes of $G_{i+1}$, as shown in Fig.~\ref{fig:Hamiltonian Graph} (gray lines represent the fact that we identify rightmost nodes of $G_n$ to leftmost nodes of $G_1$). 


\begin{lemma}[Lemma \ref{lem:ham} restated]\label{lem:ham appendix}
$G$ consists of three paths $P^0$, $P^1$ and $P^2$ where for any $0\leq j\leq 2$, $P^j$ has $v_0^j$ as one end vertex and $v_n^{(j+ \sum_{1\leq i\leq n} x_i\cdot y_i)\mod 3}$ as the other.
\end{lemma}
\begin{proof}
We will show that for any $2\leq k\leq n$ and $0\leq j\leq 2$, $P^j$ has $v_0^j$ as one end vertex and $v_k^{(j+ \sum_{1\leq i\leq k} x_i\cdot y_i)\mod 3}$ as the other. We prove this by induction on $k$. Our claim clearly holds for $k=2$ by Observation~\ref{observation:ham appendix}. Now assume that this claim is true for any $2\leq k\leq n-1$, i.e., $v_0^j$ is connected by a path to $v_k^{j'}$ where $j'=(j+ \sum_{1\leq i\leq k} x_i\cdot y_i)\mod 3$. By Observation~\ref{observation:ham appendix}, $v_k^{j'}$ is connected by a path to $v^{j''}$ where $j''=(j'+x_{k+1}\cdot y_{k+1}) \mod 3 = (j+ \sum_{1\leq i\leq k+1} x_i\cdot y_i)\mod 3$ as claimed.
\end{proof}

\begin{figure}
\center
\includegraphics[clip=true, trim=0cm 7cm 0cm 4.3cm, width=0.9\linewidth]{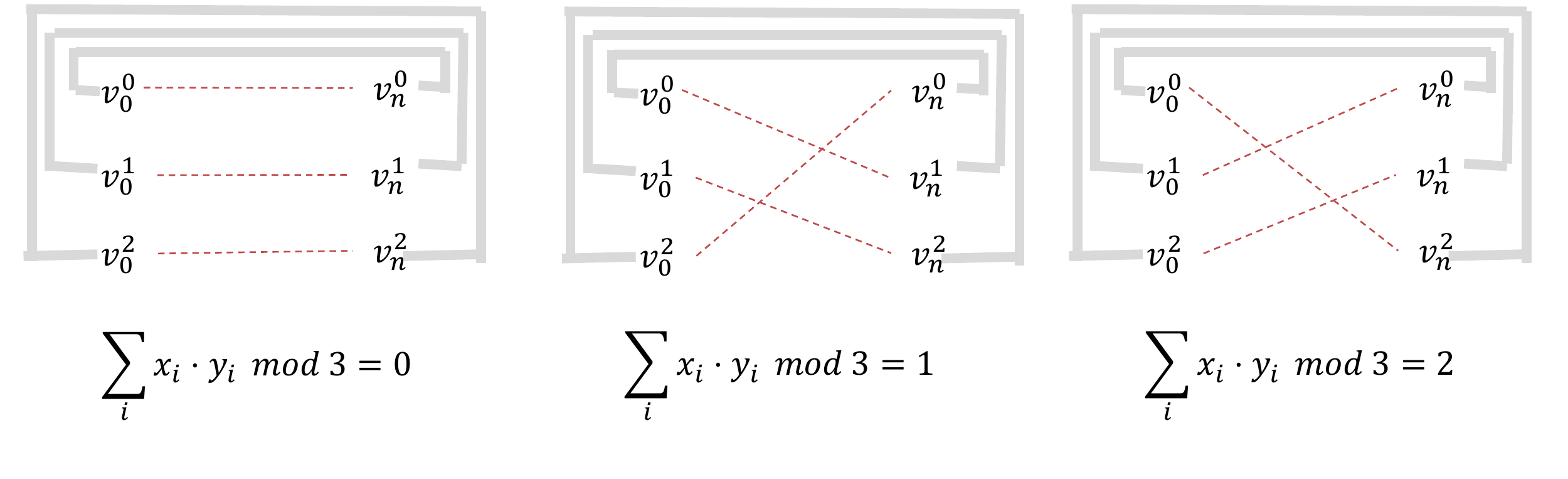}
\caption{The resulted graph $G$ in three situations depending on the value of $\sum_{1\leq i\leq n} x_i\cdot y_i \mod 3$. Dashed lines (in red) represent paths connecting $v_0^0, \ldots, v_0^2$ and $v_n^0, \ldots, v_n^2$. Thick lines (in gray) show the fact that we identify nodes on two sides, i.e. $v_0^j=v_n^j$ for all $0\leq j\leq 2$. Our main observation is that $G$ is a Hamiltonian cycle if and only if  $\sum_{1\leq i\leq n} x_i\cdot y_i \mod 3\neq 0$ (cf. Lemma~\ref{lem:ham if and only if appendix}).}\label{fig:ham three graph results Appendix}
\end{figure}

\begin{lemma}
\label{lem:ham if and only if appendix} Each player's edges form a perfect matching in $G$. Moreover, $G$ is a Hamiltonian cycle if and only if $\sum_{1\leq i\leq n} x_i\cdot y_i \mod 3\neq 0$.
\end{lemma}
\begin{proof} We consider three different values of $z=\sum_{1\leq i\leq n} x_i\cdot y_i \mod 3$ as shown in Fig.~\ref{fig:ham three graph results Appendix}. If $z=0$ then Lemma~\ref{lem:ham appendix} implies that $v_0^j$ will be connected to $v_n^j$ by a path, for all $j$. After we identify $v_0^j$ with $v_n^j$ we will have three distinct cycles, each containing a distinct $v_0^j=v_n^j$.
If $z=1$ then Lemma~\ref{lem:ham appendix} implies that $v_0^j$ will be connected to $v_n^{(j+1)\mod 3}$ by a path. After we identify $v_0^j$ with $v_n^j$ we will have one cycle that connects $v_0^0=v_n^0$ to $v_n^1=v_0^1$ then to $v_n^2=v_0^2$. Similarly, if $z=1$ then Lemma~\ref{lem:ham appendix} implies that $v_0^j$ will be connected to $v_n^{(j+2)\mod 3}$ by a path. After we identify $v_0^j$ with $v_n^j$ we will have one cycle that connects $v_0^0=v_n^0$ to $v_n^2=v_0^2$ then to $v_n^1=v_0^1$.
\end{proof}


\section{Detail of Section \ref{sec:from server to distributed}}\label{sec:full proof from server to distributed}

\begin{theorem}[Theorem \ref{theorem:from server to distributed} restated]
For any $B$, $L$, $\Gamma\geq \log L$, $\beta\geq 0$ and $\epsilon_0, \epsilon_1>0$, there exists a $B$-model quantum network $N$ of diameter $\Theta(\log L)$ and $\Theta(\Gamma L)$ nodes such that
\begin{align*}
\mbox{if } Q_{\epsilon_0, \epsilon_1}^{*, N}({\sf P}(N))\leq \frac{L}{2}-2 \mbox{ then } Q_{\epsilon_0, \epsilon_1}^{*, \server}({\sf P}_{\Gamma})= O((B\log L)Q_{\epsilon_0, \epsilon_1}^{*,N}({\sf P}(N)))
\end{align*}
where {\sf P} can be replaced by \ham and $(\beta\Gamma)\mbox{-}\conn$.
\end{theorem}

\subsection{Description of the network $N$} \label{sec:network description}

\begin{figure}
\centering
\includegraphics[width=0.6\linewidth]{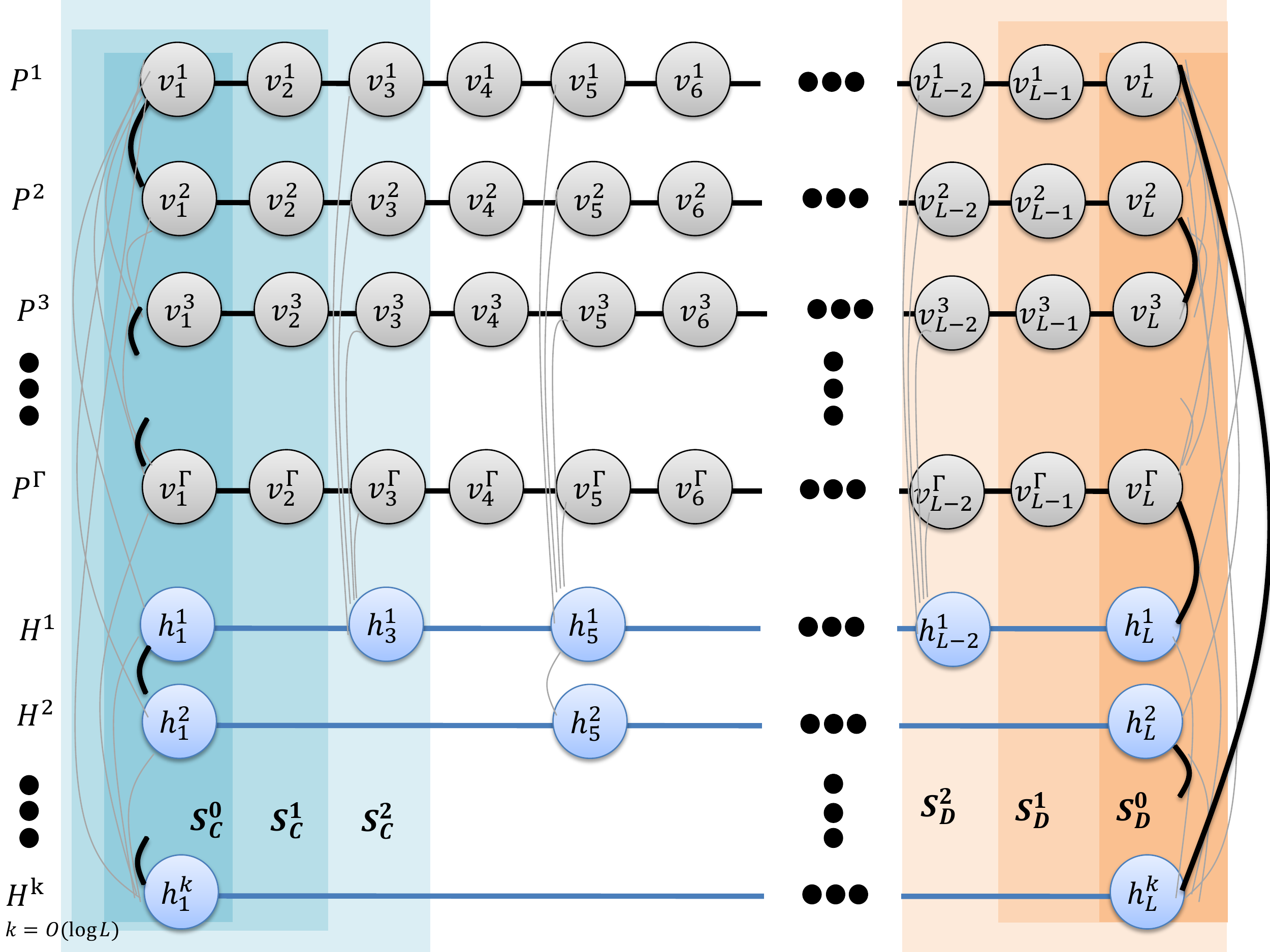}
\caption{(Fig.~\ref{fig:Network small diameter} reproduced) The network $N$ which consists of network $N'$ and some ``highways'' which are paths with nodes $h^i_j$ (i.e., nodes in blue). Bold edges show an example of subnetwork $M$ when the input perfect matchings are $E_C=\{(u_1, u_2), (u_3, u_4), \ldots, (u_{\Gamma+k-1}, u_{\Gamma+k}\}$ and $E_D=\{(u_2, u_3), (u_4, u_5), \ldots, (u_{\Gamma+k}, u_1)\}$. Pale edges are those in $N$ but not in $M$.}\label{fig:Network small diameter Appendix}
\end{figure}

In this section we describe the network $N$ as shown in Fig.~\ref{fig:Network small diameter Appendix}. We assume that $L=2^i+1$ for some $i$. This can be assumed without changing the theorem statement by simply increasing $L$ to the nearest number of this form.

The two basic units in the construction are {\em paths} and {\em highways}. There are $\Gamma$ paths, denoted by $P^1, P^2, \ldots, P^{\Gamma}$, each having $L$ nodes, i.e., for $j=1, 2, \ldots \Gamma$,
%
$$V(P^i)= \{v_1^i,\dots,v_{L}^i\}
~~~~~\mbox{and}~~~~~
E(P^i) = \{(v_j^i, v_{j+1}^i) \mid 1 \le j \le L-1\}\,.$$
We construct $k=\log_2 (L-1)$ highways, denoted by $H^1, \ldots, H^k$ where $H^i$ has the following nodes and edges.
\begin{align*}
V(H^i)&= \{h_{1+j2^i}^i \mid 0\leq j\leq \frac{L-1}{2^i}\}
~~~~~\mbox{and}~~~~~\\
E(H^i) &= \{(h^i_{1+j2^i}, h^i_{1+(j+1)2^i}) \mid 0 \le j \le \frac{L-1}{2^i}\}\,.\
\end{align*}
For any node $h^1_j$ we add an edge $(h^1_j, v^i_j)$ for any $j$. Moreover for any node $h^i_j$ we add an edge $(h^{i-1}_j, h^i_j)$.
Figure~\ref{fig:Network small diameter Appendix} depicts this network. We note the following simple observation.
%
%
\begin{observation}\label{lem:size}
The number of nodes in $N$ is $n = \Theta(L\Gamma)$ and its diameter is $\Theta(\log L)$.
\end{observation}

\subsection{Simulation}
For any $0\leq t\leq (L/2)-2$, we partition $V(N)$ into three sets, denoted by $S_C^t$, $S_D^t$ and $S_S^t$, as follows (also see Fig.~\ref{fig:Network small diameter Appendix}).
\begin{align}
S_C^t & =\{v^i_j, h^i_j \mid 1\leq i\leq \Gamma,~1\leq j\leq t+1\},\\
S_D^t & =\{v^i_j, h^i_j \mid 1\leq i\leq \Gamma,~L-t\leq j\leq L\},\\
S_S^t & =V(N)\setminus (S_C^t\cup S_D^t)\,.\label{eq:own set}
\end{align}
Let $\mathcal{A}$ be any quantum distributed algorithm on network $N$ for computing a problem {\sf P} (which is either \ham or $(\beta\Gamma)-\conn$). Let $T_{\mathcal{A}}$ be the worst case running time of algorithm $\mathcal{A}$ (over all inputs). We note that $T_\mathcal{A}\leq (L/2)-2$, as assumed in the theorem statement.
We show that Carol, David and the server can solve problem ${\sf P}$ on $(\Gamma+k)$-node input graph using small communication, essentially by ``simulating'' $\cA$ on some input subnetwork $M$ corresponding to $G=(U, E_C\cup E_D)$ in the following sense. When receiving $E_C$ and $E_D$, the three parties will construct a subnetwork $M$ of $N$ (without communication) in such a way that $M$ is a $1$-input of problem ${\sf P}$ (e.g., $M$ is a Hamiltonian cycle) if and only if $G=(U, E_C\cup E_D)$ is. Next, they will simulate algorithm $\cA$ in such a way that, at any time $t$ and for each node $v^i_j$ in $N$, there will be exactly one party among Carol, David and the server that knows {\em all information that $v^i_j$ should know in order to run algorithm $\cA$}, i.e., the (quantum) state of $v^i_j$ as well as the messages (each consisting of $B$ quantum bits) sent to $v^i_j$ from its neighbors at time $t$. The party that knows this information will pretend to be $v^i_j$ and apply algorithm $\cA$ to get the state of $v^i_j$ at time $t+1$ as well as the messages that $v^i_j$ will send to its neighbors at time $t+1$. We say that this party {\em owns $v^i_j$} at time $t$. Details are as follows.


We will define a server-model protocol $\cA'$ that guarantees that, at any time $t$, Carol, David and the server will own nodes in sets $S_C^t$, $S_D^t$ and $S_S^t$, respectively, at time $t$. That is, Carol's workspace, denoted by $H_{C, C}$, contains all qubits in $H_{vv'}$, for any $v\in S_C^t$ and $v'\in V(N)$, resulting from $t$ rounds of an execution of $\cA$. Similarly, David's (respectively the server's) workspace, denoted by $H_{D, D}$ (respectively $H_{S, S}$), contains all qubits in $H_{vv'}$ for any $v\in S_D^t$ (respectively $v\in S_S^t$) and $v'\in V(N)$ resulting from $t$ rounds of $\cA$. In other words, if after $t$ rounds of $\cA$ network $N$ has state
\begin{align*}
\ket{\psi^t_M}&=\sum_w \left(\alpha_w \bigotimes_{v, v'\in V(N)}\ket{\psi^t_{w, M}(v, v')}\right)\,,
\end{align*}
then we will make sure that the server model has state
$$\ket{\Psi^t_G}=\sum_w \left(\alpha_w \bigotimes_{i,j\in \{C, D, S\}}\ket{\Psi^t_{w, G}(i,j)}\right)$$
where $\ket{\Psi^t_{w, G}(i,j)}=\ket{0}$, for any $i, j\in \{C, D, S\}$ such that $i\neq j$, and for any $i\in \{C, D, S\}$
\begin{align}
\ket{\Psi^t_{w, G}(i,i)}=\bigotimes_{v\in S_i^t, v'\in V(N)}\ket{\psi^t_{w, M}(v',v)}\,.\label{eq:reduction state}
\end{align}

Let $\Gamma'=\Gamma+k$. Fix any $\Gamma'$-node input graph $G=(U, E_C\cup E_D)$ of problem {\sf P} where $E_C$ and $E_D$ are edges given to Carol and David respectively. Let $U=\{u_1, \ldots, u_{\Gamma'}\}$. For convenience, for any $1\leq j\leq k$, let $v^{\Gamma+j}_1=h^{j}_1$ and $v^{\Gamma+j}_L=h^{j}_L$ We construct a subnetwork $M$ of $N$ as follows.
For any $i\neq j$, we mark $v^i_1v^j_1$ as participating in $M$ if and only if $u_iu_j\in E_C$. Note that this knowledge must be kept in qubits in $H_{v^i_1v^i_1}$ and $H_{v^j_1v^j_1}$ in network $N$ as we require each node to know whether edges incident to it are in $M$ or not. This means that this knowledge must be stored in $H_{C, C}$ since $v^i_1,v^j_1\in S_C^0$. This can be guaranteed without any communication since Carol knows $E_C$.
Similarly, we mark $v^i_Lv^j_L$ as participating in $M$ if and only if $u_iu_j\in E_D$, and this information can be stored in $H_{D, D}$ without communication. Finally, we let all edges in all paths and highways be in $M$. This information is stored in $H_{S, S}$. An example of network $M$ is shown in Fig.~\ref{fig:Network small diameter Appendix}.
To conclude, if the initial state of $N$ with this subnetwork $M$ is
\begin{align*}
\ket{\psi^0_M}&=\sum_w \left(\alpha_w \bigotimes_{v, v'\in V(N)}\ket{\psi^0_{w, M}(v, v')}\right)\,.
\end{align*}
then the server model will start with state $\ket{\Psi^0_G}=\sum_w \left(\alpha_w \bigotimes_{i,j\in \{C, D, S\}}\ket{\Psi^0_{w, G}(i,j)}\right)$ where $\ket{\Psi^0_{w}(i,j)}=\ket{0}$, for any $i, j\in \{C, D, S\}$ such that $i\neq j$, and for any $i\in \{C, D, S\}$
\begin{align*}
\ket{\Psi^0_{w, G}(i,i)}=\bigotimes_{v\in S_i^0, v'\in V(N)}\ket{\psi^0_{w, M}(v',v)}\,.
\end{align*}
%
%
Thus Eq.\eqref{eq:reduction state} holds for $t=0$.
We note the following simple observation.
\begin{observation}\label{observation:G is ham if and only if M is ham Appendix}
$G=(U, E_C\cup E_D)$ is a Hamiltonian cycle if and only if $M$ is a Hamiltonian cycle. $G$ is connected if and only if $M$ is connected, and for any $\delta$, $G$ is $\delta$-far from being connected if and only if $M$ is $\delta$-far from being connected.
\end{observation}
\danupon{Should write a proof?}
Thus, Carol, David and the server can check whether $G$ is a Hamiltonian cycle if they can check whether $M$ is a Hamiltonian cycle. Similarly, they can check if $G$ is connected or $(\beta \Gamma)$-far from being connected by checking $M$. So, if Eq.\eqref{eq:reduction state} can be maintained until $\cA$ terminates then we are done since each server-model player can pretend to be one of the nodes they own and measure the workspace of such node to get the property of $M$.

Now suppose that Carol, David and the server have maintained this guarantee until they have executed $\cA$ for $t-1$ steps, i.e., player $i$ owns the nodes in $S_i^{t-1}$ at time $t-1$. They maintain the guarantee at step $t$ as follows. First, each player simulate the internal computation of $\cA$ on nodes they own. That is, for each node $v\in V(N)$, the player $i$ such that $v\in S_i^{t-1}$ applies the transformation $C_{t, v}$  (cf. Section~\ref{sec:formal definition of network}) on qubits in workspace $\bigotimes_{v'\in V(N)} H_{v'v}$ which is maintained in $H_{i, i}$ at time $t-1$. This means that if after the internal computation $N$ has state
$\ket{\upsilon^t_M}=\sum_w \left(\alpha_w \bigotimes_{v, v'\in V(N)}\ket{\upsilon^t_{w, M}(v, v')}\right)$
then the server model will have state $\ket{\Upsilon^t_G}=\sum_w \left(\alpha_w \bigotimes_{i,j\in \{C, D, S\}}\ket{\Upsilon^t_{w, G}(i,j)}\right)$ where $\ket{\Upsilon^t_{w}(i,j)}=\ket{0}$, for any $i\neq j$, and
$\ket{\Upsilon^t_{w, G}(i,i)}=\bigotimes_{v\in S_i^t, v'\in V(N)}\ket{\upsilon^t_{w, M}(v',v)}$
for any $i$. Note that the server model players can simulate the internal computation of $\cA$ without any communication since a player that owns node $v$ has all information needed to simulate an internal computation of $v$ (i.e., the state of $v$ as well as all messages $v$ received at time $t-1$).

At this point, for any $i\in\{C, D, S\}$, player $i$'s space contains the current state and out-going messages of every node $v\in S_i^{t-1}$. They will need to receive some information in order to guarantee that they own nodes in $S_i^t$. First, consider Carol. Let $S'_C$ be the set of rightmost nodes in the set $S_C^{t-1}$, i.e. $S'_C$ consists of $v^i_{t+1}$ and $h^i_j$ for all $i$ and $j=\arg\max_j \{h^i_j\in S_C^{t-1}\}$.

Note that Carol already has the workspace and all incoming messages of nodes in $S_C^{t-1}\setminus S'_C$ at time $t$.  This is because for any $v\in S_C^{t-1}\setminus S'_C$, Carol already has qubits in $H_{v'v}$ for all $v'\in V(N)$. For each $v\in S'_C$, Carol is missing the messages sent from $v$'s right neighbor; i.e., Carol does not have qubits in $H_{v^i_{t+2}v^i_{t+1}}$ and $H_{h^i_{j'}h^i_j}$ for all $i$, $j=\arg\max_j \{h^i_j\in S_C^{t-1}\}$ and $j'=\arg\min_{j'} \{h^i_{j'}\notin S_C^{t-1}\}$. Since $S'_C\subseteq S_C^{t}$, we need to make sure that Carol has all information of nodes in $S'_C$ at time $t$.

For a non-highway node $v^i_{t+1}$, for all $i$, this can be done by letting the server who owns $v^i_{t+2}$ send to Carol a message sent from $v^i_{t+2}$ to $v^i_{t+1}$ at time $t$, i.e., qubits in $H_{v^i_{t+2}v^i_{t+1}}$. For highway node $h^i_{j}$ for all $i$ and $j=\arg\max_j \{h^i_j\in S_C^{t-1}\}$, its right neighbor $h^i_{j'}$, where $j'=\arg\min_{j'} \{h^i_{j'}\notin S_C^{t-1}\}$, might be owned by David or the server. In any case, we let the owner of $h^i_{j'}$ send to Carol the message sent from $h^i_{j'}$ to $h^i_j$ at time $t$, i.e., qubits in $H_{h^i_{j'}h^i_j}$. The cost of doing this is zero if $h^i_{j'}$ belongs to the server and at most $B$ if $h^i_{j'}$ belongs to David since the message size is at most $B$. In any case, the total cost will be at most $Bk$ since there are $k$ highways. We can thus make sure that Carol gets the information of nodes in $S_C^{t-1}$ at time $t$ at the total cost of at most $Bk$.
%

In addition to this, Carol needs to get information of nodes in $S_C^{t}\setminus S_C^{t-1}$ at time $t$. This means that, for any $v\in S_C^{t}\setminus S_C^{t-1}$ she has to receive the qubits stored in $H_{v'v}$ for all $v'\in V(N)$. For any non-highway node $v^i_{t+2}\in S_C^{t}\setminus S_C^{t-1}$, it can be checked from the definition that $v^i_{t+2}$ and all its neighbors are in $S_C^{t-1}\cup S_S^{t-1}$. So, we can make sure that Carol owns $v^i_{t+2}$ by letting the server send to Carol the workspace of $v^i_{t+2}$ and messages sent to $v^i_{t+2}$ by its neighbors in $S_S^{t-1}$ (i.e. qubits in $H_{v'v^i_{t+2}}$ for all $v'\in S_S^{t-1}$). This communication is again free.
For a highway node $h^i_j$ in $S_C^{t}\setminus S_C^{t-1}$, it can be checked from the definition that $h^i_j$ as well as all its {\em non-highway} neighbors are in $S_C^{t-1}\cup S_S^{t-1}$. The only neighbor of $h^i_j$ that might be in $S_D^{t-1}$ is its right neighbor, say $h^i_{j'}$, in the highway. If $h^i_{j'}$ is in $S_D^{t-1}$ then David has to send to Carol the message sent from $h^i_{j'}$ to $h^i_j$. This has cost at most $B$. So, Carol can obtain the workspace of $h^i_j$ as well as all messages sent to $h^i_j$ at the cost of $B$. Since there are $k$ highway nodes in $S_C^{t}\setminus S_C^{t-1}$, the total cost for Carol to obtain information needed to maintain nodes in $S_C^{t}\setminus S_C^{t-1}$ is $Bk$.
We conclude that Carol can obtain all information needed to own nodes in $S_C^t$ at time $t$ at the cost of $2Bk$.

We can do the same thing to guarantee that David owns all nodes in $S_D^t$ at time $t$ at the cost of $2Bk$.
Now we make sure that the server own nodes in $S_S^t$. First, observe that the server already has the workspace of all nodes in $S_S^t$ since $S_S^t\subseteq S_S^{t-1}$. Moreover, the server already has all messages sent to all non-highway nodes in $S_S^t$ (i.e. $v^i_j$ for all $t+2\leq j\leq L-t-1$ and $1\leq i\leq \Gamma$) since all of their neighbors are in $S_S^{t-1}$. Additionally, each leftmost highway node $h^i_j\in S_S^t$, for any $i$ and $j=\arg\min_j \{h^i_j\in S_S^t\}$, has at most one neighbor in $S_C^{t-1}$ (i.e., its right neighbor in the highway). Similarly, each rightmost highway node $h^i_j\in S_S^t$, for any $i$ and $j'=\arg\max_{j'} \{h^i_{j'}\in S_S^t\}$, has at most one neighbor in $S_D^{t-1}$ (i.e., its right neighbor in the highway). Thus, the server needs to obtain from Carol and David at most $2B$ qubits to maintain $h^i_j$ and $h^i_{j'}$. Since there are $k$ highways, the server needs at most $2kB$ qubits total from Carol and David.
We thus conclude that the players can maintain Eq.\eqref{eq:reduction state} at the cost of $6kB=O(B\log L)$ qubits per round as desired.

As noted earlier, the server-model players will simulate $\cA$ until $\cA$ terminates. Then they can measure the workspace of nodes they own to check whether $M$ is a $0$- or $1$-input of problem {\sf P}. Observation~\ref{observation:G is ham if and only if M is ham Appendix} implies that they can use this answer to answer whether $G$ is a $0$- or $1$-input with the same error probability. Since each round of simulation requires a communication complexity of $O(B\log L)$ and the simulation is done for $T_\cA\leq Q^{*, N}_{\epsilon_0, \epsilon_1}({\sf P}(N))$ rounds, the total communication complexity is $O((B\log L)Q^{*, N}_{\epsilon_0, \epsilon_1}({\sf P}(N)))$ as claimed.

  \let\oldthebibliography=\thebibliography
  \let\endoldthebibliography=\endthebibliography
  \renewenvironment{thebibliography}[1]{%
    \begin{oldthebibliography}{#1}%
      \setlength{\parskip}{0ex}%
      \setlength{\itemsep}{0ex}%
  }%
  {%
    \end{oldthebibliography}%
  }
{
\bibliographystyle{alpha}
\bibliography{quantum}
}

\end{document}